\documentclass[12pt]{article}

\usepackage{natbib}
\usepackage{booktabs}
\usepackage{threeparttable}
\usepackage{longtable}
\usepackage{rotating}
\usepackage{epsfig}
\usepackage{graphicx}
\usepackage{comment}
\usepackage{lscape}
\usepackage{setspace}
\usepackage{amsthm, amssymb, amsmath, amsthm, mathtools}
\usepackage{enumerate}
\usepackage{epstopdf}
\usepackage{tikz}
\usetikzlibrary{arrows}
\usetikzlibrary{decorations.markings}
\usepackage{graphicx, xcolor}
\usepackage{dcolumn}
\usepackage{placeins,enumitem}
\usepackage{refcount}
\usepackage{pdflscape,rotating}

\usepackage{appendix}
\usepackage{hyperref}
\hypersetup{
  colorlinks   = true, 
  urlcolor     = black, 
  linkcolor    = black, 
  citecolor   = black 
}
\usepackage{makecell}
\usepackage{pict2e}
\usepackage{arydshln}

\usepackage{titling}
\usepackage[T1]{fontenc}
\usepackage[latin9]{inputenc}
\usepackage{array}
\usepackage{multirow}

\usepackage[font=footnotesize]{caption} 
\usepackage[font=footnotesize]{subcaption}

\makeatletter

\usepackage{bbm}

\newtheorem{definition}{Definition}

\newtheorem{ex}{Example}

\newtheorem{thm}{Theorem}
\newtheorem{prop}{Proposition}

\newtheorem{cor}{Corollary}[prop]

\newtheorem{lemma}{Lemma}
\newtheorem{customcor}{Corollary}

\makeatother
\usepackage[bottom, multiple]{footmisc}

\makeatother
\makeatletter
\renewcommand\footnoterule{%
  \vspace{0.8em}
 \kern-3\p@\hrule\@width.4\columnwidth%
   \kern2.6\p@}
\makeatother
\usepackage[margin=1in]{geometry}

\def\sym#1{\ifmmode^{#1}\else\(^{#1}\)\fi}

\usepackage[semibold, lf, tabular, scale=0.91]{FiraSans} \usepackage[lf, scale=0.91]{FiraMono}	
\usepackage[charter, cal, greekuppercase=italicized, sfscaled=false, ttscaled=false]{mathdesign}	
\usepackage[scaled=.96, lf, sups]{XCharter}	
\usepackage[protrusion=true, expansion=false, kerning=true]{microtype}
\DeclareMathOperator {\corr}{\text{Cor}}

\DeclarePairedDelimiter\set\{\}

\begin{document}

\title{\textsc{Tradeoffs and Comparison Complexity}\thanks{We are indebted to Benjamin Enke, Matthew Rabin, Joshua Schwartzstein, and Tomasz Strzalecki for their excellent supervision and guidance. We also thank Katie Coffman, John Conlon, Bnaya Dreyfuss, Ignacio Esponda, Xavier Gabaix, Thomas Graeber, David Laibson, Shengwu Li, Daniel Martin, Paulo Natenzon, Gautam Rao, Alex Rees-Jones, Andrei Shleifer, Keyu Wu, and Harvard PhD workshop participants for helpful comments and suggestions. We are grateful for financial support from the Chae Family Economics Research Fund. Shubatt: Department of Economics, Harvard University. cshubatt@g.harvard.edu. Yang: Department of Economics, University of California, Santa Barbara. jeffrey.yang97@gmail.com.}}


\author{Cassidy Shubatt \and Jeffrey Yang}

\date{\small{First posted: January 13, 2024\\
Updated: \today}\\
Latest version \href{https://jeffreyyang97.github.io/personalwebsite/CC_Draft.pdf}{\textcolor{blue}{here}}}

\normalsize
\maketitle

\thispagestyle{empty}

\vspace*{-.0cm}

\begin{abstract}

Using theory and experiments, this paper shows that the difficulty of making tradeoffs offers a parsimonious explanation for a wide range of behavioral phenomena. We develop a model of imprecise comparisons applicable to multiattribute, lottery, and intertemporal choice, which formalizes the idea that comparisons are difficult when they involve pronounced tradeoffs. The model rationalizes a range of documented regularities, such as context effects, preference reversals, apparent probability weighting and hyperbolic discounting, and generates new implications for behavior. We assess the explanatory power of our model in a series of choice experiments. Our model explains a large share of the variation in choice inconsistency across problems, and we document that manipulating tradeoffs reverses classic behavioral regularities, in line with its predictions. 

\end{abstract}

\vspace*{1cm}

\small{

}
\normalsize



\newpage
\setcounter{page}{1}

\section{Introduction}
\label{SEC:intro}

There is increasing interest in studying the complexity of economic decisions. Recent work has argued that a number of behavioral regularities are in part driven by cognitive frictions, from canonical anomalies in  risky choice \citep{khaw_cognitive_2021,frydman_source_2025,enke_cognitive_2023,oprea_decisions_2024} and intertemporal choice \citep{gabaix_myopia_2017,vieider_noisy_2021,enke_complexity_2025} to a range of patterns in other domains \citep{de_clippel_caution_2025, arrieta_procedural_2023,guan_beyond_2025,ba_over-_2022,augenblick_overinference_2025,enke_behavioral_2024}. This evidence points to the potential for a unifying explanation of these regularities. However, it is less well-understood what exactly makes a decision complex, which limits our ability to predict when and which complexity-driven distortions will arise in a given setting. 

Using theory and experiments, this paper shows how a single source of complexity---the difficulty of making tradeoffs---can offer a parsimonious explanation for a broad set of behavioral phenomena and yields new predictions supported by the data. Table 1 compiles these phenomena, which span multiple choice domains and decision problems. They include classic anomalies in behavioral economics: decoy effects in multiattribute choice, probability weighting and hyperbolic discounting in the valuation of risk and time, as well as choice instabilities such as preference reversals and inconsistencies in probability weighting across elicitation procedures. While some of these patterns have individual explanations in the literature, from classic accounts of non-standard preferences to recent models of cognitive frictions,  existing models only speak to a subset of anomalies and domains. We show how they can be rationalized by a common mechanism: the need to aggregate tradeoffs---prices against product features, payoffs against probabilities, money against time---makes some options hard to compare, generating distortions depending on which comparisons are relevant to the decision.


We formalize this idea in a model of imprecise comparisons applicable to multiattribute, lottery, and intertemporal choice. The model has two components: complexity measures that specify how tradeoffs govern comparison difficulty in each domain, and a general account of how this difficulty distorts choice. Together, these components rationalize the  regularities in Table \ref{table:anomalies}, and also generate new predictions on how tradeoffs shape noise and systematic patterns in choice and valuation. We take these predictions to the data. First, we use large-scale binary choice data to assess whether our complexity measures indeed predict the difficulty of comparisons. In all three domains, our measures strongly predict choice inconsistency and other indicators of choice difficulty, such as subjective uncertainty and response times. Second, we test the framework's prediction that manipulating tradeoffs can reverse or eliminate canonical choice patterns. We find these predicted effects in domains where they have previously been documented and, crucially, in new ones, providing out-of-sample evidence for the framework.


\begin{table}[!t]
\begin{center}
\caption{Behavioral regularities}
\label{table:anomalies}
\scalebox{.8}{
\begin{threeparttable}
    \begin{tabular}{lll}
    \toprule\addlinespace
    \textbf{Domain} & \textbf{Empirical regularity} & \textbf{Example reference}  \\\addlinespace \hline\addlinespace
Multiattribute & Decoy/asymmetric dominance effects & \cite{huber_adding_1982} \\
Risk & Inverse-S probability weighting in certainty equivalents & \cite{kahneman_prospect_1979} \\ 
Risk & Absence of fourfold pattern in binary choice & \cite{harbaugh_fourfold_2010} \\ 
Risk & Reversed probability weighting in probability equivalents & \cite{feldman_certain_2024} \\ 
Risk & Lottery preference reversals & \cite{tversky_anomalies_1990}\\
Risk & Elimination of preference reversals in probability equivalents & \citet{butler_imprecision_2007}\\
Risk & Dissimilarity-based noise in binary choice & \cite{enke_quantifying_2023}\\
Time & Hyperbolic discounting in present value equivalents & \cite{cohen_measuring_2020} \\ 
Time & Invariance of hyperbolic discounting to front-end delays & \cite{cohen_measuring_2020} \\ 
Time & Intertemporal preference reversals & \cite{tversky_causes_1990} \\ 
\addlinespace
\hline
\addlinespace
Multiattribute & Tradeoff-driven noise in binary choice & New prediction \\
Time & Tradeoff-driven noise in binary choice & New prediction \\
Time & Reversed hyperbolic discounting in time equivalents & New prediction\\
Time & Elimination of preference reversals in time equivalents & New prediction\\
\addlinespace \bottomrule
    \end{tabular}
\end{threeparttable}}
    \end{center}
\end{table}

\noindent\textbf{\textit{Contribution}}. This paper builds on ideas from distinct literatures in cognitive economics, decision theory, and psychology, reviewed at the end of this section. We make two main contributions. First, we show how these ideas can be brought together in a single framework to provide a unifying explanation for behavioral phenomena spanning multiple choice domains and decision problems. We achieve this scope precisely by connecting ideas that had previously been developed separately: as we discuss in the related literature below, each of the antecedent accounts we build on can only speak to a subset of these phenomena, whereas our framework expands their reach to a far broader range of behavior. Second, we provide systematic evidence for our mechanism by evaluating the explanatory power of our framework across a range of choice domains and decision problems. Taken together, our parallel results across these settings provide evidence that these seemingly disjoint phenomena can be explained by a common source.\\

\noindent\textbf{\textit{Modeling comparison complexity}}. 
Our paper is motivated by the idea that aggregating tradeoffs is difficult, and that not all comparisons require the same degree of aggregation. To illustrate, consider the options in Figure \ref{fig:opening_examples}. Across these domains, the comparison $(x,y)$ is simple---there is little need to make tradeoffs to see that $x$ is better than $y$. On the other hand,  $(x',y)$ is less obvious: 1a) involves a tradeoff between fixed and variable fees, 1b) involves trading off a higher maximum payout against a lower payout probability, and 1c) involves a tradeoff between money and delays. Tradeoffs complicate comparisons for at least two reasons: decision-makers may be uncertain over the relative importance of different features, and even absent this uncertainty, may find it difficult to aggregate advantages and disadvantages across features.


We formalize this in a model of imprecise comparisons, in which a decision-maker is uncertain over the value of each option $(v_x)_{x\in X}$, and chooses based on noisy signals of how they compare. The precision of each signal, $\tau_{xy}$, captures the ease of comparison between $x$ and $y$.  We specify $\tau_{xy}$ in multiattribute, lottery, and intertemporal choice, developing concrete measures of comparison complexity that reflect the difficulty of making tradeoffs in each domain.


These measures are grounded by two shared properties: similarity and dominance. First, holding fixed value differences, options are easier to compare if they are more \textit{similar} across features. Formally, the ease of comparison $\tau_{xy}$ is increasing in the \textit{value-dissimilarity ratio} $|v_x-v_y|/d(x,y)$, where $d(x,y)$ is a distance metric. Intuitively, similar options require less aggregation of tradeoffs to compare, as one can divert attention from congruent features and more easily assess differences. This principle is motivated by insights from psychology \citep{tversky_substitutability_1969} and our formalism follows recent work in decision theory \citep{he_moderate_2024}. To put structure on $d(x,y)$, we appeal to our second principle: options are maximally easy to compare when there is \textit{dominance}---that is, when there are no tradeoffs. Domain-specific dominance notions give rise to tractable dissimilarity measures in each domain. 
\begin{figure}[t!]
\small
\begin{subfigure}[t]{0.31\textwidth}
\centering
    \caption{\centering Multiattribute }
    \vspace{0.5em}
$
    \begin{array}{ll}
        x: &  \text{\$11/month, \$3.45/GB}\\
        x': &  \text{\$32/month, \$1.6/GB} \\
        y: &  \text{\$10.95/month, \$4.45/GB} 
    \end{array}
$
\end{subfigure}
\begin{subfigure}[t]{0.44\textwidth}
\centering
\caption{\centering Lottery }
    \vspace{0.5em}
$
    \begin{array}{ll}
        x: &  \text{\$27 w.p. 25\%, \$3 w.p. 75\%}\\
        x': &  \text{\$9 for sure} \\
        y: &  \text{\$20 w.p. 20\%, \$3.2 w.p. 80\%}\\
    \end{array}
$
\end{subfigure}
\begin{subfigure}[t]{0.20\textwidth}
\centering
\caption{\centering Intertemporal}
    \vspace{0.5em}
$
\begin{array}{ll}
        x: &  \text{\$60 in 61 days} \\        
        x': &  \text{\$100 in 3 years} \\
        y: &  \text{\$40 in 60 days} 
\end{array}
$
\end{subfigure}
\captionsetup{font=small}
\caption{Choice Domains. Comparisons between a) phone plans characterized by a monthly and data use fee, b) monetary lotteries, and c) payoff flows.
}
\label{fig:opening_examples}
\end{figure}

These complexity measures have transparent behavioral implications for binary choice. Our model belongs to the moderate utility class axiomatized by \citet{he_moderate_2024}, where binary choice probabilities depend on value differences normalized by a distance metric. Building on this decision-theoretic framework, we show that axioms on binary choice behavior corresponding to similarity and dominance, together with other easily-understood axioms, characterize the implications of our complexity measures for binary choice in each domain. \\


\noindent\textit{\textbf{Rationalizing behavioral regularities}}.  Using our model of imprecise comparisons, we show that tradeoff-based comparison complexity can rationalize a broad set of regularities central to behavioral economics and recent work on complexity, and generates novel testable implications. 

These regularities follow from two key predictions of the model. First, in binary choice, comparison complexity leads to noisy, but unbiased, choice. Second, in richer menus, comparison complexity generates systematic distortions: a) context effects in choice, which occur when competing alternatives are hard to compare to each other but differ in their comparability to other options in the menu, and b) ``pull-to-center'' effects in the valuation of options in a multiple price list: that is, when the option being valued is hard to compare to a range of prices, its valuations are compressed towards the center of that range. 


In conjunction with our measures of comparison complexity, these forces rationalize a range of seemingly-unrelated behavioral regularities that cut across choice domains and decision problems. Given our multiattribute complexity measure, the context effects in our model generate familiar decoy/asymmetric dominance effects. With the structure of our lottery and intertemporal measures, which predict which risky and intertemporal prospects are hard to compare to prices, our model can rationalize documented biases and instabilities in valuation: preference reversals and apparent probability weighting and hyperbolic discounting, specifically because these patterns can be explained by a pull-to-center bias whose magnitude depends on tradeoffs.

To illustrate, consider a standard paradigm for measuring risk preferences: eliciting certainty equivalents of a lottery $l$ that pays \$$\overline{w}$ with probability $p$. As $p$ approaches 0 or 1 in the limit, $l$ has a near-dominance relationship with the payments in the price list, so valuations converge to accuracy. Away from these boundaries, tradeoffs between payoffs and probabilities make $l$ harder to compare to prices, producing a pull-to-center: small (large) probabilities are overvalued (undervalued). Taken together, this produces inverse-S probability weighting. The same logic generates apparent hyperbolic discounting in the valuation of delayed payments. 

Importantly, these distortions are not fixed biases in our model, but instead arise from tradeoffs in comparing options to prices. As such, our framework not only provides foundations for complexity-based accounts of patterns like probability weighting, but also rationalizes why these patterns are often unstable across choice settings. Suppose that instead of valuing lotteries against certain payments, the decision-maker values certain payments against lotteries: assessing the payoff probability $p$ that makes $l$ indifferent to $\$w\in[0,\overline{w}]$ for sure. The pull-to-center effects that produce probability weighting in certainty equivalents here generate the \textit{opposite} pattern: low (high) certain payments are overvalued (undervalued), which can explain documented inconsistencies between the valuation methods \citep{feldman_certain_2024}. Furthermore, the same mechanism operates in intertemporal choice, yielding the novel prediction that hyperbolic discounting can be reversed by similarly inverting the valuation task. 

Since these pull-to-center distortions affect valuations but not direct choice, our framework can also explain puzzling inconsistencies between valuation and choice, such as the preference reversal phenomenon \citep{tversky_anomalies_1990}, where risk preferences seemingly flip across direct choice and valuations, and predicts analogous reversals in intertemporal choice. At face value, these reversals challenge the existence of stable preferences. Our model instead rationalizes them as expressions of the same underlying preferences under tradeoff-driven noise. Here, the model also predicts that these reversals can be \textit{eliminated} by manipulating the relative ease of comparing each option to prices---specifically by changing the numeraire against which options are valued. The model thus provides a formal explanation for documented shifts in lottery reversal patterns across  valuation formats \citep{butler_imprecision_2007} and delivers novel predictions on how changing the numeraire can similarly eliminate reversals in intertemporal choice.\\


\noindent\textit{\textbf{Experimental evidence.}} We conduct two sets of experiments. First, we use large-scale binary choice datasets to assess whether our complexity measures predict comparison difficulty in our three domains: choice between induced-values multiattribute goods, lotteries, and time-dated payoff streams. In all domains, our measures strongly predict choice inconsistency---subjects are 15 percentage points more likely to make inconsistent choices across repeat instances of a problem with the highest vs. lowest level of tradeoff complexity---and are similarly predictive of three other markers of choice difficulty:  subjective uncertainty over choices, response times, and choice errors vis-a-vis a preference benchmark. We show that potential confounds such as the proximity to indifference or heterogeneity cannot fully account for these relationships.

In each domain, our model of noise explains a large share of variation in choice rates not captured by existing preference models, and improves the predictive power of leading models in lottery and intertemporal choice by 10-22\%. We quantify our model's \textit{completeness} and \textit{restrictiveness} \citep{fudenberg_measuring_2022, fudenberg_how_2023} in our datasets. In terms of completeness, our model explains 70\% of the predictable variation in multiattribute choice, and over 90\% in lottery and intertemporal choice. At the same time, our model is \textit{more} restrictive than leading models in the latter domains, suggesting that its performance does not come at substantial costs to parsimony.

Second, we run a series of experiments to test our framework's predictions that preference reversals and canonical patterns in lottery and intertemporal valuation can be eliminated or reversed by manipulating tradeoffs. In our preference reversal experiments, we reproduce classic reversals in lottery choice and document analogous reversals in intertemporal choice, and show that they can be eliminated by manipulating the relative ease of comparing options to prices. In line with the model, aggregate reversals between valuations and direct choice disappear when subjects instead value lotteries (delayed payments) using probability equivalents (time equivalents). Finally, we document that two classic valuation patterns, inverse-S probability-weighting and hyperbolic discounting, reverse when we manipulate tradeoffs in the valuation task. In line with model predictions and existing evidence, we find that using probability equivalents rather than certainty equivalents to estimate risk preferences produces apparent \textit{underweighting} of small probabilities and \textit{overweighting} of large probabilities. The analogous prediction holds in intertemporal choice: relative to the discounting in present value equivalents, time equivalents  \textit{overvalue} delays close to the present and \textit{undervalue} longer delays.

These experiments provide novel evidence that the same instabilities in lottery choice---the elimination of preference reversals and flipping valuation distortions via changes in the numeraire---also emerge in intertemporal choice, just as our model predicts. The intertemporal results thus provide an out-of-sample test of our framework, and the parallel effects we document across the two domains point to a common mechanism behind these patterns.\\


\noindent\textit{\textbf{Related literature.}} This paper contributes to a literature on complexity and the cognitive foundations of economic behavior. While recent work has argued that various behavioral regularities are driven by cognitive frictions, such as prospect theory-like behavior \citep{khaw_cognitive_2021,frydman_source_2025,steiner_perceiving_2016,enke_cognitive_2023,oprea_decisions_2024} and hyperbolic discounting \citep{gabaix_myopia_2017,enke_complexity_2025}, we lack a formal understanding of the sources of choice difficulty that produce these regularities. 
This paper shows how a single mechanism---the difficulty of making tradeoffs---can rationalize phenomena as distinct as probability weighting, hyperbolic discounting over money, preference reversals, context effects, and the large variation of noise in binary choice, and also predicts new regularities that find support in the data. While different formal explanations have been proposed for some of these regularities, ours is the first to our knowledge to accommodate them simultaneously. 

This paper most closely relates to two strands within this literature. First, the pull-to-center effects in our model resemble predictions of noisy cognition models \citep{khaw_cognitive_2021,frydman_efficient_2021,enke_cognitive_2023,gabaix_myopia_2017}, wherein cognitive frictions compress valuations towards a prior. By modeling tradeoffs as one source of these frictions, our model puts empirically-supported structure on how the magnitude of compression effects, often a key degree of freedom in these models, varies across options.\footnote{In doing so, our framework also provides a structured mechanism for central tendency effects documented in valuation tasks \citep{andersen_elicitation_2006,charness_complexity_2018,beauchamp_measuring_2020}.} Second, our focus on tradeoff complexity connects to \citet{enke_quantifying_2023}, which finds that a measure of ``excess dissimilarity,'' related to our lottery complexity measure, emerges as a key predictor of the difficulty of binary lottery choice. We extend and test this dissimilarity notion in other choice domains, demonstrating that its predictive power is not specific to risk but instead reflects a domain-general mechanism of tradeoff difficulty. Crucially, we also show that this mechanism can rationalize a far broader set of regularities beyond noise in binary lottery choice.

Our approach not only rationalizes known empirical patterns, but also helps organize a body of evidence on their instability across elicitation methods \citep{tversky_anomalies_1990,seidl_preference_2002,harbaugh_fourfold_2010,bouchouicha_cognitive_2026,imai_meta-analysis_2025,feldman_certain_2024}---findings that pose fundamental issues for economists hoping to apply insights from behavioral research to novel settings, or to draw inferences from choice data. While past work has argued that option comparability could produce inconsistencies across elicitation methods \citep{tversky_anomalies_1990,butler_imprecision_2007,hsee_preference_2009}, these accounts do not provide a general measure of comparability or a formal model of how it maps into behavior. Our framework provides this formal structure, yielding sharp predictions that can help researchers anticipate when and how behavior changes across elicitation methods, both existing and new.\footnote{See \cite{shubatt_valuations_2026} for follow-up work using our framework to organize patterns across a broad range of valuation paradigms in the literature, and to characterize which paradigms are robust to complexity.}



Our model builds on a stochastic choice literature \citep{natenzon_random_2019,he_random_2023,he_moderate_2024} studying the moderate utility class of binary choice rules. We bring these models closer to data in two ways. First, we propose concrete moderate utility specifications which reflect an underlying mechanism---the difficulty of tradeoffs---and test both these specifications and the  mechanism using large-scale datasets. Second, we expand the empirical reach of moderate utility to a broad set of regularities beyond binary choice. We extend this class to general decision problems by embedding our specifications in a model of imprecise ordinal comparisons, yielding a framework that rationalizes numerous behavioral anomalies beyond the existing scope of this literature. Our learning-based extension draws on ideas from the Bayesian Probit model due to \cite{natenzon_random_2019}, and we discuss the relationship between these models in Section \ref{SEC:existing_models}.





Section \ref{SEC:theory} develops our specification of comparison complexity. Section \ref{SEC:multinomial} applies the framework to explain empirical regularities in choice and valuation. Section \ref{SEC:experiments} describes the experimental tests of the model. Section \ref{SEC:existing_models} discusses relationships to existing models, and Section \ref{SEC:conclusion} concludes. Appendix \ref{APP:proofs} contains proofs of all results in the main text.

\section{Theory of Comparison Complexity}
\label{SEC:theory}

We model a decision-maker who faces imprecise comparisons and responds in a Bayesian fashion. This section develops our theory of how tradeoffs make a comparison difficult. To exposit this theory, we first consider the restriction of our model to binary choice, where the DM faces a single comparison. In the subsequent section we extend the model to decisions involving multiple comparisons, and show how tradeoff difficulty can explain a range of empirical regularities. 


Let $X$ denote a set of options and $v_x$ denote the value of each $x\in X$. Consider a decision-maker (DM) who is uncertain about the value of each option, and when faced with a binary choice $\set{x,y}$, chooses based on a noisy signal on which option is better. The DM has continuous, i.i.d. priors over each $(v_x)_{x\in X}$ given by the symmetric distribution $Q$, and receives a signal $s_{xy}$ on the value comparison between $x$ and $y$:
\begin{align*}
s_{xy}&=\text{sgn}(v_x - v_y)+\frac{1}{\sqrt{\tau_{xy}}}\epsilon_{xy}\\
&\epsilon_{xy}\sim N(0,1).
\end{align*}
The DM chooses the option with the highest posterior expected value, randomizing in the case of ties; this yields choice probabilities  $\rho(x,y)=\Phi(\text{sgn}(v_x-v_y)\sqrt{\tau_{xy}})$, for $\Phi$ the standard normal CDF.\footnote{Normality of $\epsilon_{xy}$ is not crucial to our results; see Appendix \ref{APP:proofs_multinomial}.} Here, the precision $\tau_{xy}$ governs the discriminability or the \textit{ease of comparison} between $x$ and $y$: the higher is $\tau_{xy}$, the greater the likelihood of choosing the higher-value option. 

This framework follows recent models of Bayesian cognitive imprecision \citep[e.g.][]{khaw_cognitive_2021,frydman_source_2025,vieider_noisy_2021,natenzon_random_2019}, in which the DM behaves as-if she has imperfect access to underlying values, and forms posteriors based on signals that reflect  internal deliberation. We make two departures. First, signals are over value \textit{comparisons} rather than individual values. This is natural given our focus on modeling tradeoffs, which are inherently comparative, and aligns with a tradition in psychology of modeling  comparative value judgments \citep{thurstone_law_1927}. Second, these signals convey only ordinal value information. This reflects how one can often have a precise idea of which option is better without knowing ``by how much,'' and lets the model capture situations where two options are perfectly comparable to a third, yet hard to compare to each other: if learning was instead over cardinal value differences, the DM cannot both learn $v_x-v_z$ and $v_y-v_z$ without also learning $v_x-v_y$. This feature of our ordinal signal structure will be crucial to modeling the empirical regularities we consider in Section \ref{SEC:multinomial}, which involve multiple comparisons of varying difficulty. 

Our i.i.d. priors assumption reflects an ``ignorance prior'': before observing anything about the options, the DM has no basis for judging one to be better than the other. Abstracting from asymmetric priors ensures that our predictions are driven by our specification of $\tau_{xy}$. 

We now turn to developing these specifications. We propose functional forms for $\tau_{xy}$ in multiattribute, lottery, and intertemporal choice which reflect the difficulty of making tradeoffs.\\

\noindent \textbf{\textit{Shared Principles}}. Across our choice domains, our specifications of $\tau$ satisfy two common properties reflecting the difficulty of tradeoffs: similarity and dominance. 

First, options are easier to compare when they are more \textit{similar} feature-by-feature, holding fixed their value difference. Across our domains, the ease of comparison $\tau_{xy}$ is an increasing transformation of the \textit{value-dissimilarity} ratio
\begin{align*}
    \tau_{xy}=H\left(\frac{|v_x-v_y|}{d(x,y)}\right).
\end{align*}
The numerator contains the value difference, capturing the intuition that options further from indifference are easier to compare.\footnote{This is well-documented in work studying response times \citep[e.g.][]{krajbich_visual_2010,forstmann_sequential_2016}.} The denominator is a distance metric measuring the dissimilarity of the options' features. Intuitively, if options are more feature-dissimilar, the DM faces larger tradeoffs across those features, and thus a harder comparison. The idea that dissimilarity reduces comparability has been emphasized in work in psychology \citep{tversky_substitutability_1969}, and we follow recent work in decision theory \citep{he_moderate_2024} in our formalism.



Second, we place structure on $d(x,y)$ by positing that options are maximally easy to compare when there are no tradeoffs---that is, when they have a \textit{dominance} relationship. As we discuss below, this gives rise to specific distance metrics in each of our three choice domains of interest, given the domain-relevant dominance notion.

\subsection{Multiattribute Choice}
\label{SEC:theory_mac}

Consider the domain of multiattribute choice, where each option $x\in X\equiv  X_1\times...\times X_n$ is defined on $n$ real-valued attributes, i.e., $X_i=\mathbb{R}$. Utility is linear\footnote{Linearity of preferences can be relaxed; see Appendix \ref{APP:theory_axioms}.} in attributes, where the value of each option $x$ is given by $v_x=U(x)=\sum_{k} \beta_kx_k$ for attribute weights $\beta\in \mathbb{R}^n$. We specify $\tau$ as follows.

\begin{definition}
\label{def:L1_complexity}
    $\tau$ has an $L_1$-\textit{complexity representation} $\tau^{L1}$ if there exists $\beta_k\neq 0$ such that
\begin{align*}
    \tau_{xy}=H\left(\frac{|U(x)-U(y)|}{d_{L1}(x,y)}\right)
\end{align*}
for $H$ continuous, strictly increasing with $H(0)=0$, where $d_{L1}(x,y)=\sum_{k}|\beta_k(x_k-y_k)|$. 
\end{definition} 

The ease of comparison is governed by the value-dissimilarity ratio: the aggregated value difference, normalized by a metric equal to the summed feature-by-feature value differences. This measure satisfies {similarity} and {dominance}. Holding fixed value differences, $\tau^{L1}$ is increasing in similarity as measured by a distance metric, where the choice of this metric ensures that $\tau^{L1}$ is maximized when there is dominance. Specifically, if $x$ attribute-wise dominates $y$, i.e. $\beta_k x_k\geq \beta_k y_k$ for all $k$, the ease of comparison $\tau_{xy}$ takes on its maximal value of $H(1)$. 



$\tau^{L1}$ also satisfies a \textit{simplification} property, wherein reducing the number of attributes along which there is a value difference increases the ease of comparison. To illustrate, suppose $n=3$ and $\beta=(1,1,1)$, and consider the following comparisons:
\begin{align*}
\begin{array}{ll}
&(x,y)\\
    x&=(10,7,9)\\
    y&=(3,15,5)
\end{array}
\qquad
\begin{array}{ll}
&(x',y)\\
    x'&=(3,14,9)\\
    y&=(3,15,5)
\end{array}
\end{align*}
Note that $(x',y)$ is formed by eliminating the value difference along the first attribute in $(x,y)$ and redistributing it to the second attribute. Our measure predicts that $(x',y)$ is easier to compare than $(x,y)$, i.e., 
$\tau_{x'y}>\tau_{xy}$. More generally, concentrating value differences into fewer attributes makes comparisons easier: that is, for $x'$ satisfying $x'_i=y_i$ for some $i$, $x'_j\neq x_j$ for at most one $j\neq i$, and $U(x')=U(x)$, we have $\tau_{x'y}\geq\tau_{xy}$. This property again reflects tradeoff complexity: if aggregating tradeoffs is difficult, we might expect that an operation of this kind, where some of that aggregation is done for the individual, simplifies the comparison.

Notice that the same attribute weights $\beta_k$ that govern value also enter the dissimilarity measure. This captures the intuition that the extent to which an attribute matters for tradeoffs depends on its value-relevance. However, this feature abstracts from the possibility that some attributes may be more or less difficult to trade off, independent of their utility weights. \href{https://jeffreyyang97.github.io/personalwebsite/CC_OA.pdf}{Supplemental Appendix F.2} generalizes the measure to allow for this distinction.\\

\noindent \textbf{\textit{Axiomatic Foundations}}. Our specification of $L_1$-complexity puts transparent restrictions on binary choice behavior. The binary choice probabilities induced by $\tau^{L1}$ take the form $\rho(x,y)=G\left(\frac{U(x)-U(y)}{d_{L1}(x,y)}\right)$, where $G:[-1,1]\to[0,1]$ is continuous, strictly increasing, and satisfies $G(t)=1-G(-t)$. This representation is a special case of the moderate utility class axiomatized in \cite{he_moderate_2024}, in which the utility function and distance metric are fully general. In Appendix \ref{APP:l1_axioms}, we provide an axiomatic characterization for this representation: a binary choice rule takes the above form if and only if it satisfies Moderate Transitivity, an axiom that \cite{he_moderate_2024} show characterizes the moderate utility class, Continuity and Linearity axioms, as well as Dominance and Simplification---direct translations of the two corresponding properties of $\tau^{L1}$ discussed above into restrictions on binary choice probabilities. 




\subsection{Risky and Intertemporal Choice}

\textbf{\textit{Lottery Choice}}. 
Here, each option $x$ is a monetary lottery described by the mass function $f_x:\mathbb{R}\to[0,1]$ where $f_x(w)> 0$ for finitely many $w$. Let $F_x$ and $F_x^{-1}$ denote the CDF and quantile function of $x$. Tastes are given by expected utility: $v_x=EU(x)\equiv \sum_{w}u(w)f_x(w)$.

\begin{definition}
\label{def:CDF_complexity}
    $\tau$ has a \textit{CDF-complexity representation} $\tau^{CDF}$ if for $u$ strictly increasing,
\begin{align*}
    \tau_{xy}=H\left(\frac{|EU(x)-EU(y)|}{d_{CDF}(x,y)}\right)
\end{align*}
for $H$ continuous, strictly increasing with $H(0)=0$, where $d_{CDF}$ is given by
\begin{align*}
    d_{CDF}(x,y)=\int_{0}^{1} |u(F^{-1}_x(q))-u(F^{-1}_y(q))|\,dq. 
\end{align*}
\end{definition} 

As with $L_1$-complexity, $\tau_{xy}^{CDF}$ is governed by the value-dissimilarity ratio. The specific distance metric $d_{CDF}(x,y)$ measures the area between the utility-valued CDFs of $x$ and $y$, and intuitively captures how similarly the payoffs in $x$ and $y$ are distributed. For linear $u$, $d_{CDF}$ reduces to the 1-Wasserstein metric, which has been shown to predict noise in lottery choice \citep{enke_quantifying_2023,erev_combining_2008,buschena_generalized_2008}. \\


\noindent \textbf{\textit{Intertemporal Choice}}. 
Here, each option $x$ is a payoff stream described by the \textit{payoff function} $m_x:[0,\infty)\to\mathbb{R}$, where $m_x(t)\neq 0$ for finitely many $t$. Here,  $m_x(t)$ describes how much $x$ pays at time $t$. Let $M_x(t)=\sum_{t'\leq t}m_x(t')$ denote the \textit{cumulative payoff function} of $x$, which describes the total payoff of $x$ up to time $t$. Utility is given by exponential discounting, with $v_x=PV(x)=\sum_{t}\delta^t m_x(t)$.

\begin{definition}
\label{def:CPF_complexity}
    $\tau$ has a \textit{CPF-complexity representation} $\tau^{CPF}$ if for $\delta<1$,
\begin{align*}
    \tau_{xy}= H\left(\frac{|PV(x)-PV(y)|}{d_{CPF}(x,y)}\right)
\end{align*}
for $H$ continuous, strictly increasing with $H(0)=0$, where  $d_{CPF}$ is given by
\begin{align*}
    d_{CPF}(x,y)=\ln(1/\delta)\int_{0}^{\infty}|M_x(t)-M_y(t)|\cdot\delta^t\,dt.
\end{align*}
\end{definition} 

Here, the distance metric in the value-dissimilarity ratio is proportional to the present value of the difference between the cumulative payoff functions of $x$ and $y$, and captures how similarly $x$ and $y$ distribute their payoffs across time.  \\

\noindent\textbf{\textit{Shared Properties and Axiomatic Foundations}}. As with our multiattribute  measure, $\tau^{CDF}$ and $\tau^{CPF}$ satisfy similarity and dominance. Holding fixed value differences, both are increasing in the similarity between options, as measured by a distance metric that respects dominance: $\tau^{CDF}_{xy}$ is maximized when $x$ first-order stochastically dominates $y$ (i.e. $F_x\leq F_y$), and $\tau^{CPF}_{xy}$ is maximized when $x$ temporally dominates $y$ (i.e. $M_x\geq M_y$---that is, if $x$ pays out more in total than $y$ at any point in time). Both also satisfy analogs of the simplification property for $\tau^{L1}$: as we formalize in Appendix \ref{APP:theory_axioms}, concentrating differences into fewer quantile regions and fewer time periods increases the ease of comparison under $\tau^{CDF}$ and $\tau^{CPF}$, respectively.

The behavioral implications of $\tau^{CDF}$ and $\tau^{CPF}$ for binary choice can be characterized using a parallel set of axioms on choice probabilities. These axioms are direct analogs of the axioms that characterize the binary choice implications of $\tau^{L1}$ in multiattribute choice, exposing commonalities in how tradeoff difficulty can be modeled across these three domains. In the Appendix, we state and prove these characterization results.\footnote{\citet{fishburn_probabilistic_1978} axiomatizes a related representation for binary choice probabilities in lottery choice. We discuss the relationship between these axiom systems in Appendix \ref{APP:theory_axioms}, and stress here that our contribution is to propose and provide evidence for a domain-general mechanism underlying this noise specification---tradeoff complexity---and show how this mechanism can explain a broad set of behaviors beyond binary lottery choice. } Moreover, our complexity measures in risk and time not only share the same properties of $L_1$-complexity, but can also be equivalently formulated as $L_1$-complexity applied to a suitably defined common attribute space for lotteries and payoff streams. We discuss this connection in \href{https://jeffreyyang97.github.io/personalwebsite/CC_OA.pdf}{Supplemental Appendix F.4}.

\subsection{Identification of Model Parameters}
\label{SEC:theory_params}
In each domain, binary choice probabilities under our model take the form $\rho(x,y)=G\left(\frac{v_x-v_y}{d(x,y)}\right)$,
where the signed value-dissimilarity ratio is specified according to Definitions \ref{def:L1_complexity}, \ref{def:CDF_complexity}, and \ref{def:CPF_complexity}, and $G$ is an increasing transformation satisfying $G(r)=1-G(-r)$. The model is parameterized by the preference parameters that enter the ratio---$\beta$ in multiattribute choice, $u$ in lottery choice, and $\delta$ in intertemporal choice---as well as the transformation $G$. Each of these objects are fully identified from binary choice data, as stated in Theorems \ref{THM:representation}, \ref{THM:representation_risk}, and \ref{THM:representation_time} in the Appendix. Our preferred specification of $G$ is given by
\begin{align}\label{eq:G_param}
        G(r)=\begin{dcases}(1-\kappa)-(0.5-\kappa)(1-r)^{\gamma} & r\geq 0\\
            \kappa+(0.5-\kappa)(1+r)^{\gamma} & r<0
            \end{dcases}
\end{align}
where $\kappa$ is a tremble parameter that governs the error rates at dominance, and $\gamma$ governs the curvature in the relationship between choice rates and the ratio.


\subsection{Limitations and Extensions} 
In the interest of parsimony, these measures abstract from several considerations. First, under the interpretation that the fixed precisions $\tau_{xy}$ capture the intrinsic ease of comparison, our choice model does not allow for endogenous information acquisition. Our approach therefore complements an established literature on costly information acquisition \citep{mackowiak_rational_2023}; while that literature models how decision-makers allocate effort in response to information costs, our focus is on modeling why some comparisons are more costly to resolve than others. Second, we impose standard underlying preferences in defining our measures of tradeoffs, i.e. expected utility and exponential discounting, which rules out classic binary choice phenomena such as the Allais paradox and stationarity violations.  These restrictions can be relaxed; Appendix \ref{APP:theory_axioms} extends our intertemporal measure to general discount functions, while \href{https://jeffreyyang97.github.io/personalwebsite/CC_OA.pdf}{Supplemental Appendix F.3}  extends the lottery measure to rank-dependent preferences. Finally, the Dominance property of our measures abstracts from the possibility that some dominance relations are more transparent than others \citep[e.g.][]{birnbaum_new_2008}. 

Despite these simplifications, we will show in the subsequent sections that our framework nevertheless has substantial predictive power over binary choice, and when extended to settings with multiple comparisons, rationalizes a range of empirical regularities beyond binary choice.


\section{Rationalizing Behavioral Regularities}
\label{SEC:multinomial}

We now show how tradeoff complexity can rationalize a range of empirical regularities, including context effects, preference reversals, and apparent probability weighting and hyperbolic discounting, and develop novel predictions of our framework. To study these settings, we extend our model to decisions that involve multiple comparisons.

\subsection{Multinomial Choice Extension}
As before, there is a set of options $X$, and the DM has continuous, i.i.d. priors over $v_x$ for all $x\in X$, distributed according to a symmetric distribution $Q$. Let $\mathcal{A}$ denote the collection of non-empty finite subsets of $X$, and let $\mathcal{C}=\mathcal{A}\cup\set{\varnothing}$. The DM faces a \textit{choice problem} $(A,C)\in \mathcal{A}\times\mathcal{C}$, comprised of a \textit{menu} of options  $A$ and a \textit{choice context} $C$---a set of options the DM observes but are not available for the current choice. For instance, in our application to choices from multiple price lists, $A$ is a single choice between the option being valued and a price in the list, and the choice context $C$ consists of the prices the DM observes in other rows of the list. The DM chooses from $A$ based on signals on how each pair of options in $A\cup C$ compares. 


In particular, for each pair of distinct options $x,y\in A\cup C$, the DM observes the signal $s_{xy}=\text{sgn}(v_x - v_y)+\frac{1}{\sqrt{\tau_{xy}}}\epsilon_{xy}$, where $\epsilon_{xy}\sim N(0,1)$, where $\tau_{xy}$ retains its interpretation as the ease of comparison. Let $s$ denote the collection of these signals; given a signal realization, the DM chooses the option $x\in A$ with the maximal posterior expected value $\mathbb{E}[v_x|s]$. We are interested in the resulting choice probabilities, which are given by\footnote{In the case of ties, we assume a symmetric tiebreaking rule. See Appendix \ref{APP:theory_tiebreaking} for details.}
\begin{align*}
    \rho(x,A|C)=\mathbb{P}(\set{s: \mathbb{E}[v_x|s]>\mathbb{E}[v_y|s]\,\forall\, y\in A/\set{x} }\,|\,v).
\end{align*}
Note that when restricted to binary menus, this model reduces to the model in Section \ref{SEC:theory}. Let $\rho(x,y)=\rho(x,\set{x,y}|\varnothing)$ denote binary choice probabilities, and let $\rho(x,y|C)=\rho(x,\set{x,y}|C)$ denote binary choice probabilities under a choice context $C$. 


This model has two key implications, developed below. First, comparison complexity leads to noisy but unbiased choice in binary menus---the DM may err, but more often chooses the higher-value option. Second, complexity produces systematic biases in richer choice settings: when facing hard-to-compare alternatives, the DM relies on information from other, easier comparisons, which can distort choice. This generates 1) context effects in multinomial choice, and 2) pull-to-center distortions in valuation tasks, caused by the difficulty of comparing the option being valued against prices. When combined with our theory of $\tau_{xy}$, which makes sharp predictions on comparison difficulty, these forces rationalize a broad set of empirical regularities. \\


\noindent\textbf{\textit{Context Effects.}} Consider an example in which $X=\set{x,y,z}$, with $v_x>v_y>v_z$ and $\tau_{xy}=\tau_{xz}=0$, $\tau_{yz}=\infty$. That is, the DM has no idea how $x$ compares to $y$ and $z$, but knows $y$ is better than $z$. Here, the model predicts that $\rho(y,x|\set{z})=1$: the presence of $z$ in the choice context provides information that rules out posterior beliefs over $(v_x,v_y,v_z)$ in which $v_y< v_z$, thus distorting the DM's choice in favor of the inferior option $y$. This is generalized in the following proposition, which says that if $x$ and $y$ are sufficiently hard to compare, the presence of an inferior option $z$ that is easier to compare to $y$ distorts choice in favor of that option.

\begin{prop}
    \label{PROP:context}
    Let $v_x,v_y>v_z$. If $\tau_{yz}>\tau_{xz}$, then there exists $\epsilon>0$ such that if $\tau_{xy}<\epsilon$, $\rho(y,x|\set{z})>1/2$. 
\end{prop}

This result is closely related to \citet{natenzon_random_2019}, which develops a model that also produces context effects through Bayesian learning over options that vary in comparability. While we compare the two models in Section \ref{SEC:existing_models}, we stress that here, our contribution is to specify how option features determine their comparability, which allows the framework to generate predictions from observable characteristics. In particular, combining the above result with the structure of $L_1$ complexity predicts how context effects vary with the options' configuration in attribute space, consistent with documented decoy effects \citep[e.g.][]{huber_adding_1982}, where adding a decoy $z$ that is dominated or near-dominated by $x$, but not by $y$, boosts the choice share of $x$. In Appendix \ref{APP:decoy_effects}, we discuss how $L_1$ complexity predicts documented patterns of decoy effects that existing context-dependent models in multiattribute choice cannot capture. \\

\noindent \textbf{\textit{Compression Effects in Valuations}}. We model the valuation of an option by applying our framework to choices within a multiple price list, a standard experimental method for eliciting valuations. Here, an option $x\in X$ is valued against a \textit{price list} $Z=\set{z^1,z^2,..,z^n}\subseteq X$: a set of options whose ranking $v_{z^1}>v_{z^2}>...>v_{z^n}$ is unambiguous, i.e., $\tau_{z^iz^j}=\infty$ for all $z^i,z^j\in Z$. For each price $z^k$, the DM chooses between $x$ and $z^k$, revealing her valuation of $x$ in terms of $Z$. 


Specifically, in a \textit{valuation task} $(x,Z)$, the DM receives signals $s$ for each pairwise comparison among the options in $\set{x}\cup Z$, forms a joint posterior over $(v_x,v_{z^1},...,v_{z^n})$, and makes each binary choice $\set{x,z^1},...,\set{x,z^n}$ according to those posteriors.\footnote{Formally, this is a special case of an extension of our choice model to \textit{menu sequences}, in which the DM faces a sequence of menus evaluated under a shared signal realization; see Appendix \ref{APP:theory_tiebreaking}.} Since the DM perfectly learns the ranking of prices in $Z$, these choices yield a switching point: for any signal realization, there is an index $R\in \set{1,...,n,n+1}$ for which the DM chooses the option $x\in \set{x,z^k}$ for $k\geq R$, and the price $z^k\in \set{x,z^k}$ for $k<R$. $R$ reveals where the DM believes the object $x$ falls within the ranking of prices, i.e. the subject's valuation in terms of $Z$. We will be interested in the distribution over $R$ induced by $s$, which we denote by $R(x,Z)$. With slight abuse of notation, denote $v_k \equiv v_{z^k}$ and $\tau_{xk} \equiv \tau_{xz^k}$.  


Our model predicts that when $x$ is hard to compare to prices, valuations will be systematically compressed towards the center of the price list. Below, we formalize this as a limit result in the case of uniform comparability. In the sequel, we impose structure on $\tau$ using our specification of comparison complexity, and show via simulation how our framework explains a range of empirical regularities in valuation and choice. 

\begin{prop}
    \label{PROP:valuation} Consider a valuation task $(x,Z)$ with $R^*(x,Z)\in\set{1,...,n+1}$ such that $v_x>v_k$ for $k\geq R^*(x,Z)$ and $v_x<v_k$ otherwise. If $\tau_{xk}=\tau$ and $v_x\neq v_k$ for all $k$, 
    \begin{enumerate}[label = (\roman*)]
        \item As $\tau\to 0$, $\mathbb{E}[R(x,Z)]\to (n+2)/2$.
        \item As $\tau\to \infty$, $R(x,Z)\to_p R^*(x,Z)$.
    \end{enumerate}
\end{prop}

Intuitively, when $x$ is incomparable to prices, i.e., $\tau = 0$, the DM receives no information on where $x$ falls within the ranking of prices---her posterior puts equal probability on each ranking, so she values $x$ in the middle of the list, a ``pull to center'' bias. As $x$ becomes increasingly comparable, i.e., $\tau\to \infty$, valuations converge to the truth.


These pull-to-center distortions accord with  Bayesian cognitive noise models, in which uncertainty compresses perceptions of value toward an intermediate default \citep[][e.g.]{khaw_cognitive_2021,frydman_efficient_2021}, and with documented central tendency effects in multiple price lists \citep{andersen_elicitation_2006,charness_complexity_2018,beauchamp_measuring_2020}. Our framework advances these accounts in two ways. First, it disciplines the magnitude of pull-to-center distortions, a key degree of freedom in these accounts, by tying it to the comparability of options to prices. Because our complexity measures predict which risky and intertemporal prospects are difficult to compare to prices---for instance, that a lottery paying $\$10$ with 20\% chance is harder to compare to sure payments than a degenerate lottery paying $\$2$---our framework captures regularities in how pull-to-center distortions vary across options, as the subsequent applications illustrate. Second, our framework predicts instabilities in these distortions across decision problems: since they arise from the difficulty of comparing options against the numeraire, they depend on what numeraire the options are valued against, and are not present in direct choice. Our framework thus not only accounts for valuation patterns previously attributed to pull-to-center effects, like probability weighting and hyperbolic discounting, but also rationalizes documented instabilities in these patterns and preference reversals. 



\subsection{Preference Reversals}
\label{SEC:reversals}
Consider the preference reversal phenomenon in risky choice. Lottery $x$ pays a high amount with a low probability, while $y$ pays a modest sum with a high probability, e.g.
\begin{align*}
    &x:\quad\$23.50 \text{ with 19\% chance}\\
   &y:\quad\$4.75\text{ with 94\% chance}
\end{align*}
Most subjects choose $y$ over $x$ in direct choice, yet state a higher valuation for $x$. While taken at face value, these reversals challenge the existence of stable preferences, our model offers a simple rationalizing explanation: valuations are distorted when options are hard to compare to money, and tradeoffs make some options harder to compare to money than others. In particular, because our model predicts that $y$ is easier to compare to money than $x$, differential pull-to-center effects distort their valuations, generating reversals relative to direct choice. As we show below, the same mechanism generates analogous reversals in intertemporal choice. This perspective contrasts with existing formal accounts of preference reversals, such as salience and regret-based explanations (see Appendix \ref{APP:alt_models}), which are specific to lottery choice and so do not speak to intertemporal reversals. Furthermore, our mechanism formalizes how manipulating tradeoffs in the comparison to prices can eliminate reversals; this organizes evidence on the instability of lottery reversals, and yields novel parallel predictions in intertemporal choice.

\subsubsection{Lottery Reversals}
Consider the lottery domain, where $v_x=\sum_{w}u(w)f_x(w)$ and $\tau$ has a CDF-complexity representation $\tau_{xy}^{CDF}=H\left(\frac{|EU(x)-EU(y)|}{d_{CDF}(x,y)}\right)$ for which $H(1)=\infty$; that is, the DM perfectly learns the ranking between lotteries with a dominance relationship. Call $l=(w_l,p_l)$ a simple lottery if it pays out $w_l$ with probability $p_l$, and nothing otherwise. 

\begin{ex}\label{EX:cequiv_reversals} \normalfont (Classic preference reversals). Suppose the DM is risk-averse, and again take the lotteries $x=(\$23.5,0.19)$, $y=(\$4.75,0.94)$. First, consider binary choice between the lotteries. Since any risk-averse DM weakly prefers $y$ to $x$, as $x$ is a mean-preserving spread of $y$, our model predicts that the DM is more likely to choose $y$.

Now, suppose the DM values the lotteries. Formally, the DM faces a valuation task $(l,Z)$ in which the lottery $l=(w_l,p_l)$ is valued against a price list $Z=\set{z^1,...,z^n}$, where each $z^k=(w_k,1)$ is a sure payment. We restrict to ``adapted'' price lists, where $Z$ is \textit{adapted} to $l$ if $Z$ contains equal-sized steps and contains the minimal and maximal support points of $l$ (i.e., $w_k-w_{k+1}$ is constant in $k$, and $w_n=0$, $w_1=w_l$).
Recall that each valuation task $(l,Z)$ induces a distribution of switching points $R(l,Z)$; denote by $CE(l,Z)=1/2\left[w_{R(l,Z)-1}+w_{R(l,Z)}\right]$ the resulting distribution over certainty equivalents.

\begin{figure}[t!]
    \centering
    \small
    \begin{subfigure}[t]{0.48\textwidth}
        \includegraphics[width=\linewidth]{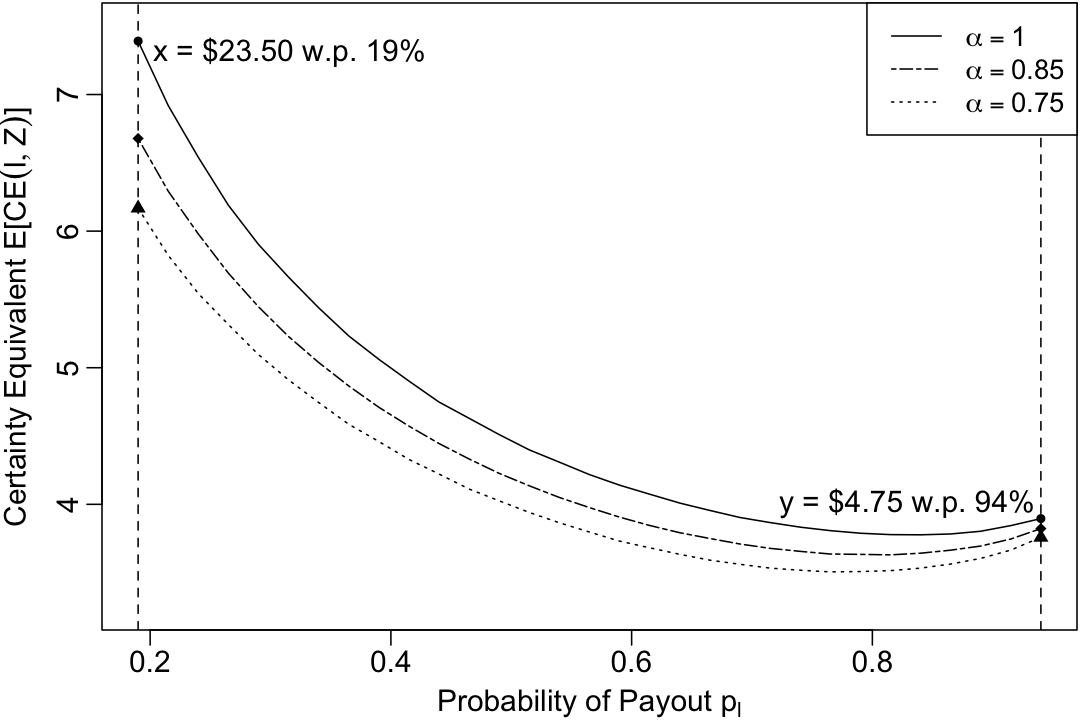}
        \caption{Certainty Equivalents $\mathbb{E}[CE(l,Z)]$}
        \label{fig:ce_sims_pr}
    \end{subfigure}
    \hspace{1em}
    \begin{subfigure}[t]{0.48\textwidth}
        \includegraphics[width=\linewidth]{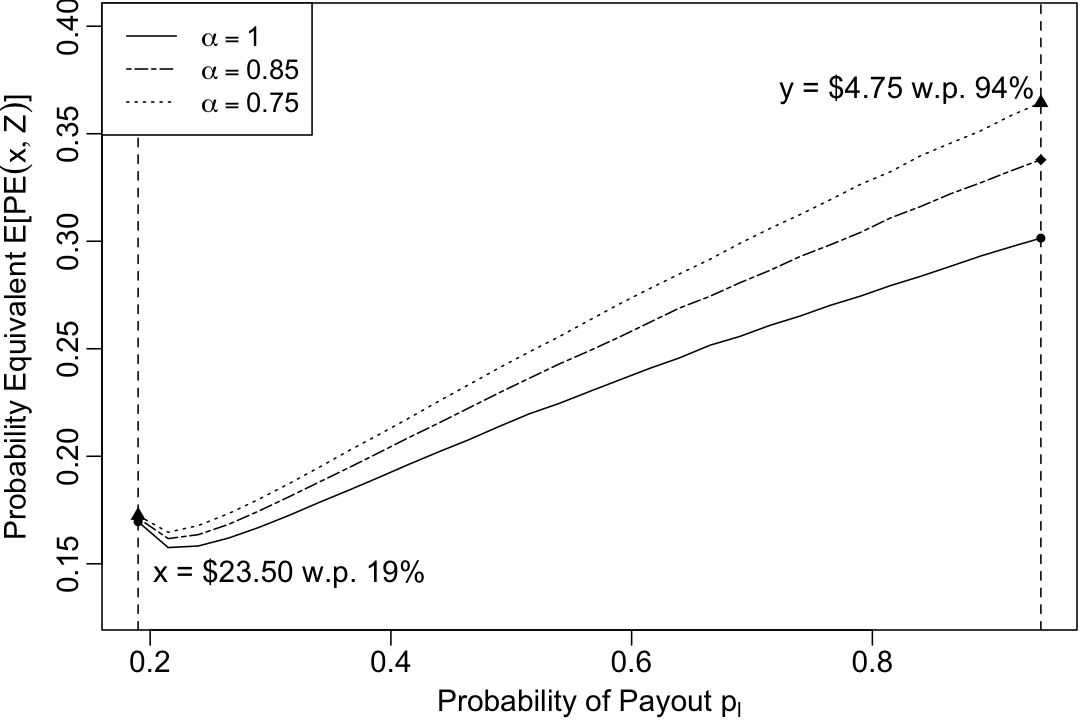}
        \caption{Probability Equivalents $\mathbb{E}[PE(l,Z)]$}
        \label{fig:pe_sims_pr}
    \end{subfigure}
    \caption{Simulated average certainty equivalents and probability equivalents for simple lotteries $l=(w_l,p_l)$ with expected value equal to that of $x=(23.50, 0.19)$ as a function of $p_l$. $Z$ is adapted to $l$ and we set $|Z|=15$. $\tau$ has a CDF-complexity representation parameterized by $u(w)=w^{\alpha}$ and $H(r)=(\Phi^{-1}(G(r)))^2$, for $G$ given by \eqref{eq:G_param} with $\kappa=0,\gamma=0.5$. Priors are distributed $Q\sim U[0,1]$.}
    \label{fig:lottery_reversal}
\end{figure}

Figure \ref{fig:ce_sims_pr} plots the average certainty equivalents $\mathbb{E}[CE(l,Z)]$ simulated from our model for lotteries $l$ with the same expected value as $x$ and $y$. Notice that the riskier lottery $x$ is valued higher than $y$ on average, despite the fact that $y$ is weakly preferred to $x$ for our risk-averse DM. Intuitively, $x$ is dissimilar to and therefore difficult to compare to money, so its valuation is inflated toward the midpoint of the undominated range of prices $[0,w_x]$. On the other hand, $y$ is easier to compare to money, so its valuation is less distorted, and if anything is pulled \textit{downward} toward the midpoint of undominated prices $[0,w_y]$.  We have a reversal: $\rho(x,y)\leq 1/2$ and yet $\mathbb{E}[CE(x,Z)]>\mathbb{E}[CE(y,Z)]$. The rest of the figure traces our model's predictions for preference reversals in general: for a sufficiently high-risk lottery $x'$ and a sufficiently low-risk lottery $y'$, we  have $\mathbb{E}[CE(x',Z)]>\mathbb{E}[CE(y',Z)]$---even though $y'$ is preferred to $x'$ and so $\rho(x',y')\leq 1/2$.

These systematic pull-to-center distortions produce reversals in median as well as average valuations. This prediction distinguishes our model from noise-based explanations of reversals such as \citet{blavatskyy_preference_2009}, which impose a strong underlying consistency between choice and valuation that rules out systematic reversals in median valuations (see Appendix \ref{APP:alt_models}). As we discuss in Section \ref{SEC:reversal_exps}, our experiments support this distinguishing prediction.

\end{ex}




In our model, reversals result from the differential ease of comparing lotteries to money. This relates to past work arguing that the difficulty of valuing lotteries against an incongruent response scale can generate preference reversals \citep{tversky_causes_1990, gilovich_compatibility_2002, butler_imprecision_2007}. Unlike previous work, however, we provide a  formal account of both what makes options hard to value, and how this difficulty distorts valuation. This structure yields sharp predictions about when reversals arise and how they vary across valuation formats in multiple choice domains. In particular, it formalizes why aggregate preference reversals in lottery choice disappear when the valuation numeraire changes \citep{butler_imprecision_2007}, and delivers novel predictions on how intertemporal reversals can be similarly eliminated.

\begin{ex}\label{EX:pequiv_reversals} \normalfont (Reversals in probabilities).
Consider lotteries $x=(\$23.50,19\%)$ and $y=(\$4.75,94\%)$ from Example \ref{EX:cequiv_reversals}. Instead of valuing $x$ and $y$ against money, suppose the DM assesses their \textit{probability-equivalents}: the probability $p$ that makes the lottery $z=(\$24,p)$  indifferent to each. Whereas $y$ was easier to compare to money, $x$ is easier to compare to this new numeraire, which is more similar to $x$ than to $y$. Our model predicts that this change in numeraire \textit{reverses} the distortions in the valuation of $x$ and $y$.

Formally, the DM now values $l=(w_l,p_l)$ against a \textit{probability list}: a price list of lotteries $Z = \set{z^1, ..., z^n}$, where each $z^k=(24,p_k)$; denote
by $PE(l, Z) = 1/2[p_{R(l,Z)-1} + p_{R(l,Z)
}]$ the DM's probability equivalents.\footnote{As in Example \ref{EX:cequiv_reversals}, we restrict to \textit{adapted} lists containing equal-sized steps with $p_1=p_l$, $p_n=0$.} Figure \ref{fig:pe_sims_pr} plots the simulated probability equivalents $\mathbb{E}[PE(l,Z)]$ for the same simple lotteries as in Figure \ref{fig:ce_sims_pr}. Intuitively, $y$ is harder to compare to prices, so its valuation is compressed upward toward the middle of the range of undominated probabilities, whereas $x$ is easier to compare, so its valuation is closer to the truth. Thus, $\mathbb{E}[PE(y,Z)] > \mathbb{E}[PE(x,Z)]$, eliminating the aggregate reversal.


\end{ex}
\subsubsection{Intertemporal Reversals} 
As our model generates preference reversals due to tradeoffs in comparing options to prices, it predicts that these reversals are not specific to risk, but also occur in intertemporal choice. Because this logic closely parallels the lottery case, we relegate formal details to Appendix \ref{APP:theory_time_reversals} and sketch the intuition here. Consider a DM who values the following options:
\begin{align*}
x &: \$27 \text{ in 750 days}\\
y &: \$8.25 \text{ in 30 days}
\end{align*}
The high-delay option $x$ is harder to compare to money today than the low-delay option $y$, so pull-to-center effects inflate $x$'s valuation toward the middle of the range $[\$0,\$27]$, while $y$'s valuation is more accurate, and if anything is compressed downward. This can cause $x$ to be valued above $y$, even if the DM in truth prefers $y$ to $x$ --- producing a reversal between valuations and choice. The model also makes the novel prediction that this reversal can be eliminated by changing the numeraire, analogous to lottery choice: when options are instead valued using suitably chosen time equivalents---eliciting the delay $t$ that makes, for instance, $\$27.5$ in $t$ days indifferent to each option---the relative ease of comparison flips, reversing valuation distortions and eliminating the inconsistency with direct choice.

\subsubsection{Robustness of Model Predictions}
The simulations above show that our model can rationalize preference reversals and their elimination across valuation formats. In Appendix 
 \ref{APP:sim_robustness}, we show that these patterns are not an artifact of our particular choice of parameter values. Not only does the model predict these patterns over a wide range of alternative parameter values, including those we estimate in our binary choice experiments in Section \ref{SEC:experiments}; it is also \textit{restrictive} in the sense that it never generates reversals in the opposite direction.
 

\subsection{Probability Weighting and Hyperbolic Discounting}
\label{SEC:valuation_biases}
The pull-to-center effects that produce preference reversals in our model also generate canonical patterns in the valuation of risk and time: probability-weighting and hyperbolic discounting. The same logic underlies both domains. Consider a simple lottery with payout probability $p$, or a delayed payment with delay $t$. Away from the boundaries, tradeoffs make these options hard to compare to money, inducing a pull-to-center: valuations of small (large) probabilities are compressed upward (downward), generating probability weighting; and valuations of small (large) delays are compressed downward (upward), generating apparent hyperbolic discounting. 

While a range of explanations have been proposed for probability weighting and hyperbolic discounting as separate phenomena (see Appendix \ref{APP:alt_models} for a discussion), our framework provides a parsimonious explanation for both that accords with recent evidence suggesting that these patterns are driven by cognitive frictions and complexity \citep{oprea_decisions_2024,frydman_source_2025,enke_cognitive_2023,enke_complexity_2025}. Furthermore, our framework moves beyond simply microfounding these patterns by making novel predictions on how they are unstable across contexts: that they can be \textit{reversed} by inverting the role of the numeraire.

Below, we show via simulation how our model rationalizes these patterns. Appendix \ref{APP:sim_robustness} shows that these predictions are robust across a comprehensive range of model parameters. 

\subsubsection{Apparent Probability Weighting in Lottery Valuation}

Consider the standard paradigm used to estimate risk preferences: eliciting certainty equivalents of simple lotteries $l=(\overline{w},p_l)$. 
As $p_l$ approaches 0 or 1 in the limit, $l$ is easy to compare to prices and so valuations converge to accuracy; away from the limit, tradeoffs between payouts and probabilities make $l$ harder to compare to prices, producing pull-to-center distortions: small probabilities are overvalued and large probabilities are undervalued. These two forces generate apparent inverse-S probability weighting in the DM's valuations, as seen in Figure \ref{fig:pwf_sim_ce}, which plots our model's predicted normalized certainty equivalents $\mathbb{E}[CE(l,Z)]/\overline{w}$ as a function of $p_l$. 

\begin{figure}[b!]
  \centering
    \small
    \begin{subfigure}[t]{0.48\textwidth}
        \includegraphics[width=\linewidth]{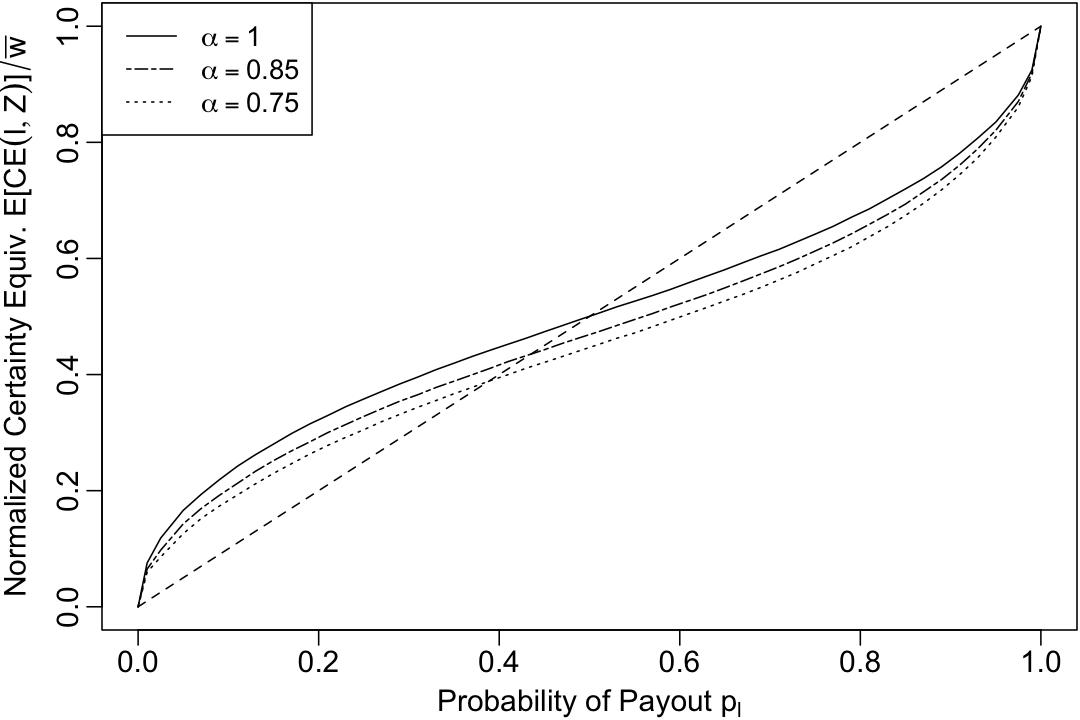}
        \caption{Simulated normalized average certainty equivalents $\mathbb{E}[CE(l,Z)]/\overline{w}$ for simple lotteries $l=(\overline{w},p_l)$ as a function of $p_l$.}
        \label{fig:pwf_sim_ce}
    \end{subfigure}
    \hspace{1em}
    \begin{subfigure}[t]{0.48\textwidth}
        \includegraphics[width=\linewidth]{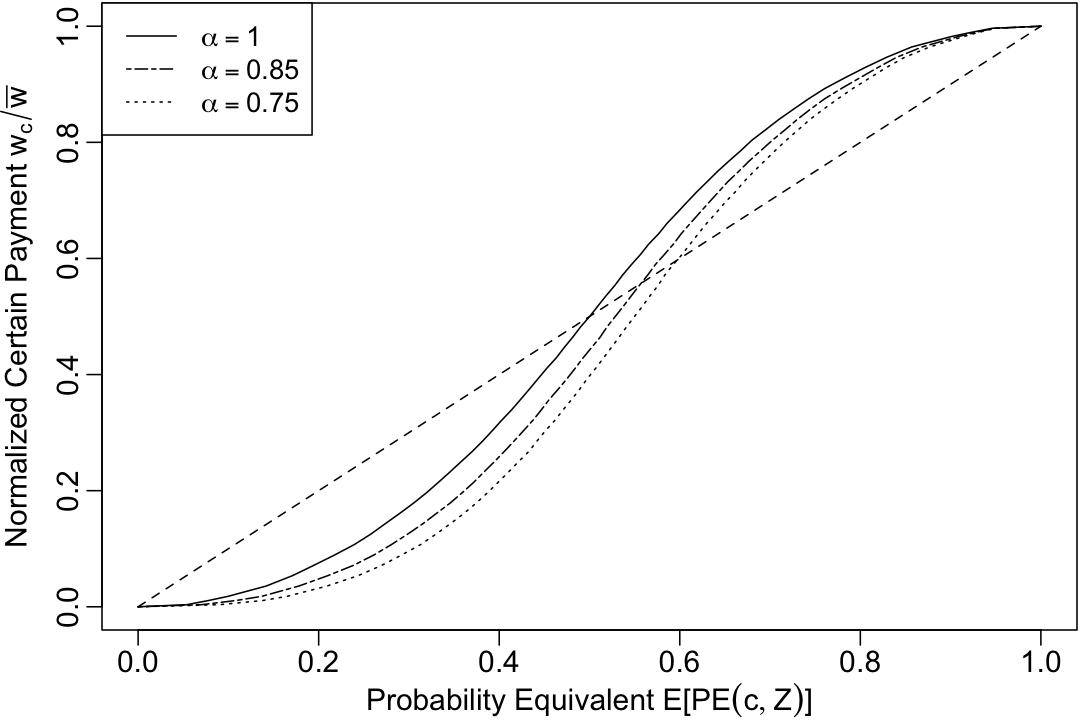}
        \caption{Simulated average probability equivalents $\mathbb{E}[PE(c,Z)]$ for certain payments $c=(w_c,1)$ as a function of  $w_c/\overline{w}$.}
        \label{fig:pwf_sim_pe}
    \end{subfigure}   
    \caption{Simulated certainty equivalents (left) and probability equivalents (right). $Z$ is adapted to $l$ with $|Z|=15$. $\tau$ has a CDF-complexity representation with $u(w)=w^{\alpha}$ and $H(r)=(\Phi^{-1}(G(r)))^2$, for $G$ given by \eqref{eq:G_param} with $\kappa=0,\gamma=0.5$. Priors are distributed $Q\sim U[0,1]$.}
    \label{fig:pwf_sim}
\end{figure}

Importantly, our model does not predict that probability weighting is a generic distortion; instead, it results from tradeoffs in comparing a lottery to a price list of sure payments. This is important for two reasons. First, our model can generate probability weighting in valuations even absent such distortions in binary choice, and can thus reconcile why the fourfold pattern is far less pronounced in direct choice than in valuation tasks \citep{harbaugh_fourfold_2010,bouchouicha_cognitive_2026}. Second, we predict that apparent probability weighting will be sensitive to the specific valuation paradigm. In particular, we predict that it is possible to \textit{reverse} the pattern of apparent probability-weighting with an appropriate choice of price list currency.

Consider an alternative paradigm for eliciting risk preferences using probability equivalents: the probability $p$ that makes the lottery $z=(\overline{w},p)$ indifferent to $c=(w_c,1)$. Tracing out probability equivalents as a function of the normalized certain payments $w_c/\overline{w}$ should---under stable preferences---reveal the same preferences as  certainty equivalents. Figure \ref{fig:pwf_sim_pe} plots the predicted relationship between $w_c/\overline{w}$ (y-axis) and the associated probability equivalent $\mathbb{E}[PE(c,Z)]$ (x-axis) simulated from our model.  Here we see a reversal of the inverse-S pattern: the difficulty of comparing sure payments against the numeraire produces a pull-to-center in probability equivalents, generating \textit{underweighting} of small probabilities and \textit{overweighting} of large probabilities. This rationalizes evidence of such inconsistencies between certainty and probability equivalents \citep{sprenger_endowment_2015,feldman_certain_2024}.\footnote{Both papers offer an explanation for these inconsistencies based on loss aversion with stochastic reference points; see Appendix \ref{APP:alt_models} for a discussion of how these explanations relate to our model and data.}



\subsubsection{Apparent Hyperbolic Discounting in Intertemporal Valuation}

Consider a standard paradigm used to estimate time discounting: valuing delayed payments $\upsilon=(\overline{m},t_{\upsilon})$ in terms of money today. Figure \ref{fig:hbd_sim_pve} plots the normalized  valuations $\mathbb{E}[PVE(\upsilon,Z)]/\overline{m}$ simulated from our model as a function of $t_{\upsilon}$.\footnote{See Appendix \ref{APP:theory_time_reversals} for formal definitions of $PVE(\upsilon,Z)$ and $TE(c,Z)$.} Here, we see apparent hyperbolic discounting: pull-to-center effects cause  undervaluation of low delays and overvaluation of high delays, with valuations approaching accuracy as $t_{\upsilon}\to 0$.



\begin{figure}[t!]
  \centering
    \small
    \begin{subfigure}[t]{0.48\textwidth}
        \includegraphics[width=\linewidth]{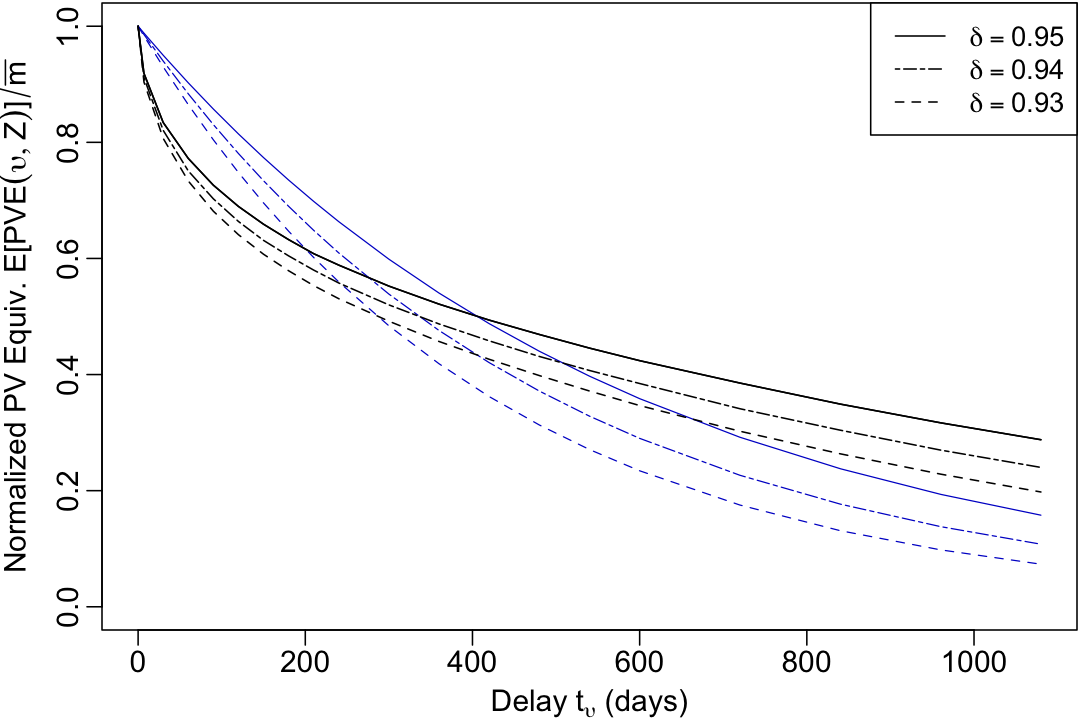}
        \caption{Simulated average present value equivalents $\mathbb{E}[PVE(\upsilon,Z)]$ (in black) for delayed payments $\upsilon=(\overline{m},t_{\upsilon})$ as a function of $t_{\upsilon}$.}
        \label{fig:hbd_sim_pve}
    \end{subfigure}
    \hspace{1em}
    \begin{subfigure}[t]{0.48\textwidth}
        \includegraphics[width=\linewidth]{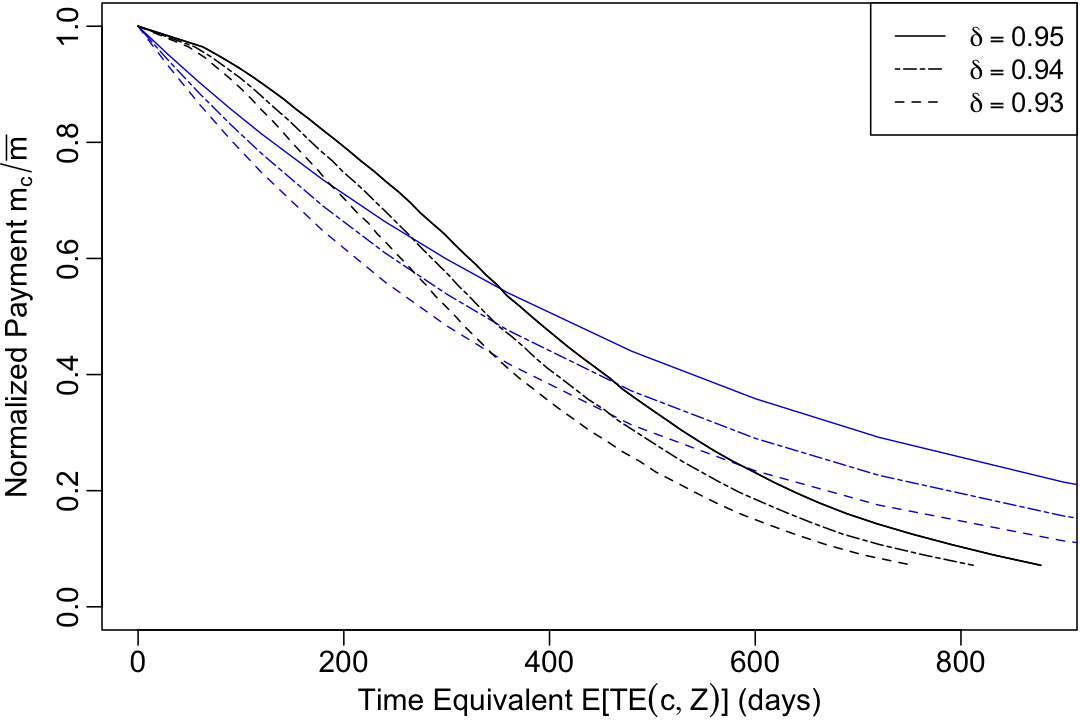}
        \caption{Relationship between simulated average time equivalents $\mathbb{E}[TE(c,Z)]$ (in black) for immediate payment $c=(m_c,0)$ and normalized amount $m_c/\overline{m}$.}
        \label{fig:hbd_sim_te}
    \end{subfigure}   
    \caption{Simulated present value equivalents (left) and time equivalents (right). For PVEs, $Z$ is adapted to $\upsilon$ with $|Z|=15$. For TEs, $Z=\set{z^1,...,z^n}$, where $z^k=(\overline{m},t_k)$, for $(t_1,...,t_n)=(0, 7, 30, 60, 90, 120, 180, 240, 300, 360, 480, 600, 720, 900, 1080)$ days. Blue curves plot distortion-free discount functions given the true discount rate $\delta$. $\tau$ has a CPF-complexity representation parameterized by $H(r)=(\Phi^{-1}(G(r)))^2$, for $G$ given by \eqref{eq:G_param} with $\kappa=0,\gamma=0.5$. Priors are distributed $Q\sim U[0,1]$.}
    \label{fig:hbd_sim}
\end{figure}


As our model generates this pattern via tradeoff difficulty rather than present-focused preferences, it predicts that hyperbolicity persists even when the valuation task involves front-end delays (see \href{https://jeffreyyang97.github.io/personalwebsite/CC_OA.pdf}{Supplemental Appendix F}). Our model can thus reconcile an empirical puzzle in the literature \citep{cohen_measuring_2020} that valuations of delayed payments often exhibit marked hyperbolicity and yet are stationary with respect to front-end delays.


Consider now an alternative valuation paradigm in which the DM assesses \textit{time equivalents}: the delay $t$ that makes $(\overline{m},t)$ indifferent to the immediate payment $c=(m_{c},0)$. Figure \ref{fig:hbd_sim_te} plots the predicted relationship between the immediate payment amount $m_c/\overline{m}$ (y-axis) and associated time equivalents $\mathbb{E}[TE(c,Z)]$ (x-axis). As with the reversal of probability weighting, the model predicts a reversal of hyperbolic discounting: pull-to-center effects generate overvaluation of payments close to the present and undervaluation of payments with longer delays.

Here we do not claim that observed hyperbolic discounting is purely a complexity-driven distortion; rather, we provide a mechanism explaining why canonical valuation tasks may overstate the degree of hyperbolic discounting present in direct choice. In \href{https://jeffreyyang97.github.io/personalwebsite/CC_OA.pdf}{Supplemental Appendix F}, we repeat the simulation exercise in Figure \ref{fig:hbd_sim} for a DM with a hyperbolic discount function, and show how comparison complexity magnifies the extent of true hyperbolic discounting in present value equivalents, and reduces it in time equivalents.  



%

\section{Experimental Evidence}
\label{SEC:experiments}

We turn to evaluating the explanatory power of our framework across choice domains and decision problems. We test whether 1) our tradeoff complexity measures predict the difficulty of binary comparisons; 2) preference reversals can be eliminated by manipulating the ease of comparing options to prices; and 3) probability weighting and hyperbolic discounting can be reversed with parallel manipulations, as our model predicts. We then quantify the predictive power of our binary choice model relative to benchmark models.


\subsection{Tests of Complexity Measures}
\label{SEC:binary_exps}
Using binary choice experiments in multiattribute, lottery, and intertemporal choice, we assess the extent to which our complexity measures predict the difficulty of comparison. Our main indicator of comparison difficulty is within-subject choice inconsistency---a direct test of the model's behavioral prediction that tradeoffs generate noisier comparisons. To probe our proposed mechanism, we also evaluate whether our measures predict three additional indicators of choice difficulty used in the literature: self-reported confidence in decisions, response times, and choice errors relative to a preference benchmark. To measure choice errors, we take a two-pronged approach. In multiattribute choice, we induce subjects' preferences over options so that errors are directly observable. In lottery and intertemporal choice, we estimate preferences from choice data and use departures from estimated preferences to proxy for errors.\footnote{In Appendix \ref{APP:axiom_tests}, we also report tests of the choice axioms of our model where our data allows.}



We carry out these analyses in three binary choice datasets. We run new experiments in multiattribute and intertemporal choice and compile existing data from \citet{enke_quantifying_2023} and \citet{peterson_using_2021} to study lottery choice. In our experiments, we recruit participants through an online survey platform to make 50 incentivized binary choices. To measure choice consistency, 10 of these problems are repeated in the survey. We also elicit participants' subjective certainty in the optimality of each of their decisions. We collect an average of 37 choices for each of 662 multiattribute choice problems and 1,097 intertemporal problems---a total of more than 66,000 decisions. The lottery dataset includes nearly 10,000 choice problems (over 1 million decisions) and includes similar measures of choice consistency and cognitive uncertainty.\footnote{Cognitive uncertainty and response times are elicited only in \citet{enke_quantifying_2023}, and problems are only repeated in \citet{peterson_using_2021}. Unlike our experiments, problems here are repeated in direct succession.}\\

\begin{sidewaysfigure}[htbp]
    \small
    \begin{subfigure}[t]{0.245\textwidth}
        \caption{\centering Multiattribute Inconsistency}
        \label{subfig:multi_incons}
        \vspace{0.5em}
        {\includegraphics[width=\linewidth]{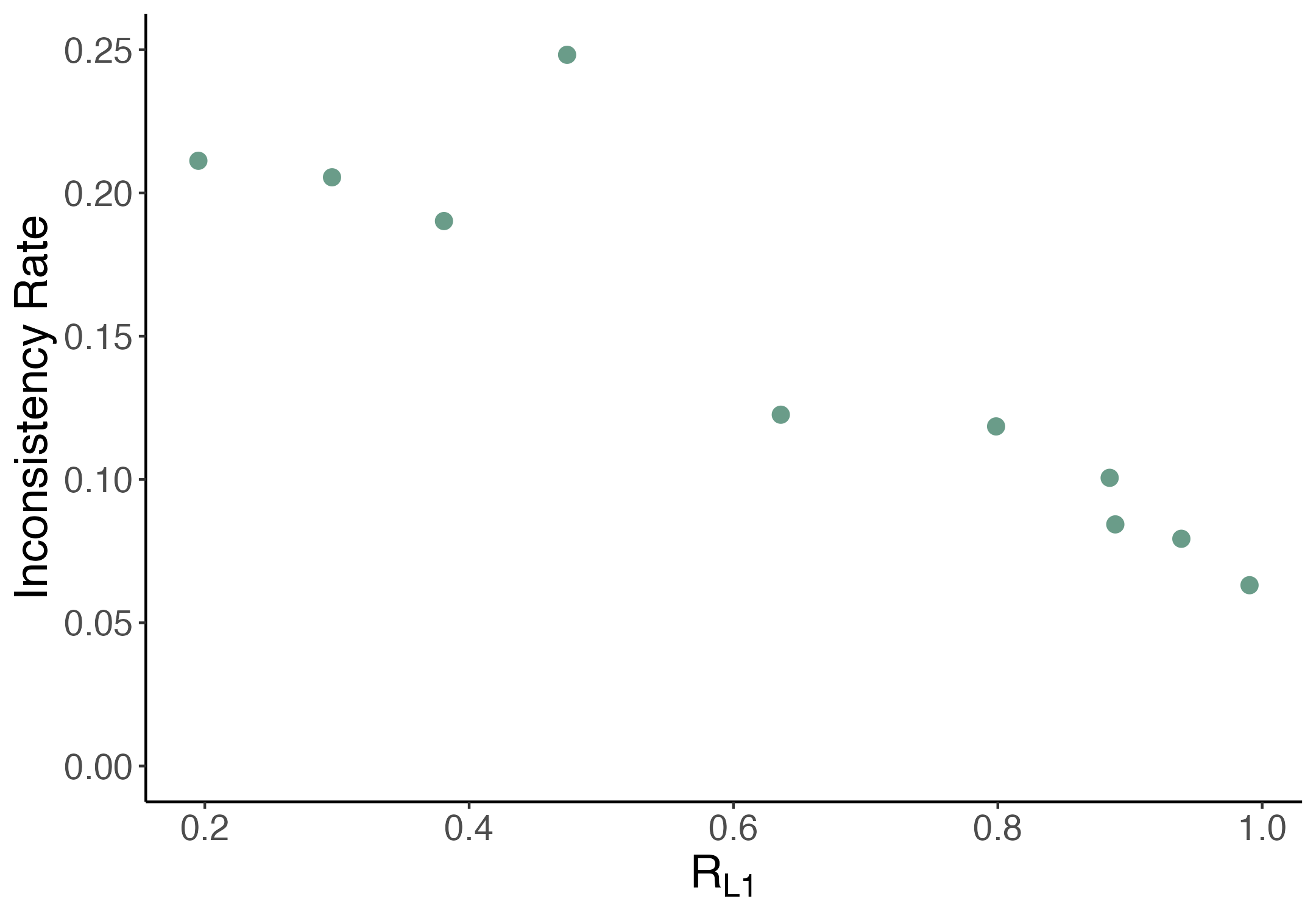}}
    \end{subfigure}
    \begin{subfigure}[t]{0.245\textwidth}
        \caption{\centering Multiattribute Error Rates}
        \label{subfig:multi_errors}
        \vspace{0.5em}
        {\includegraphics[width=\linewidth]{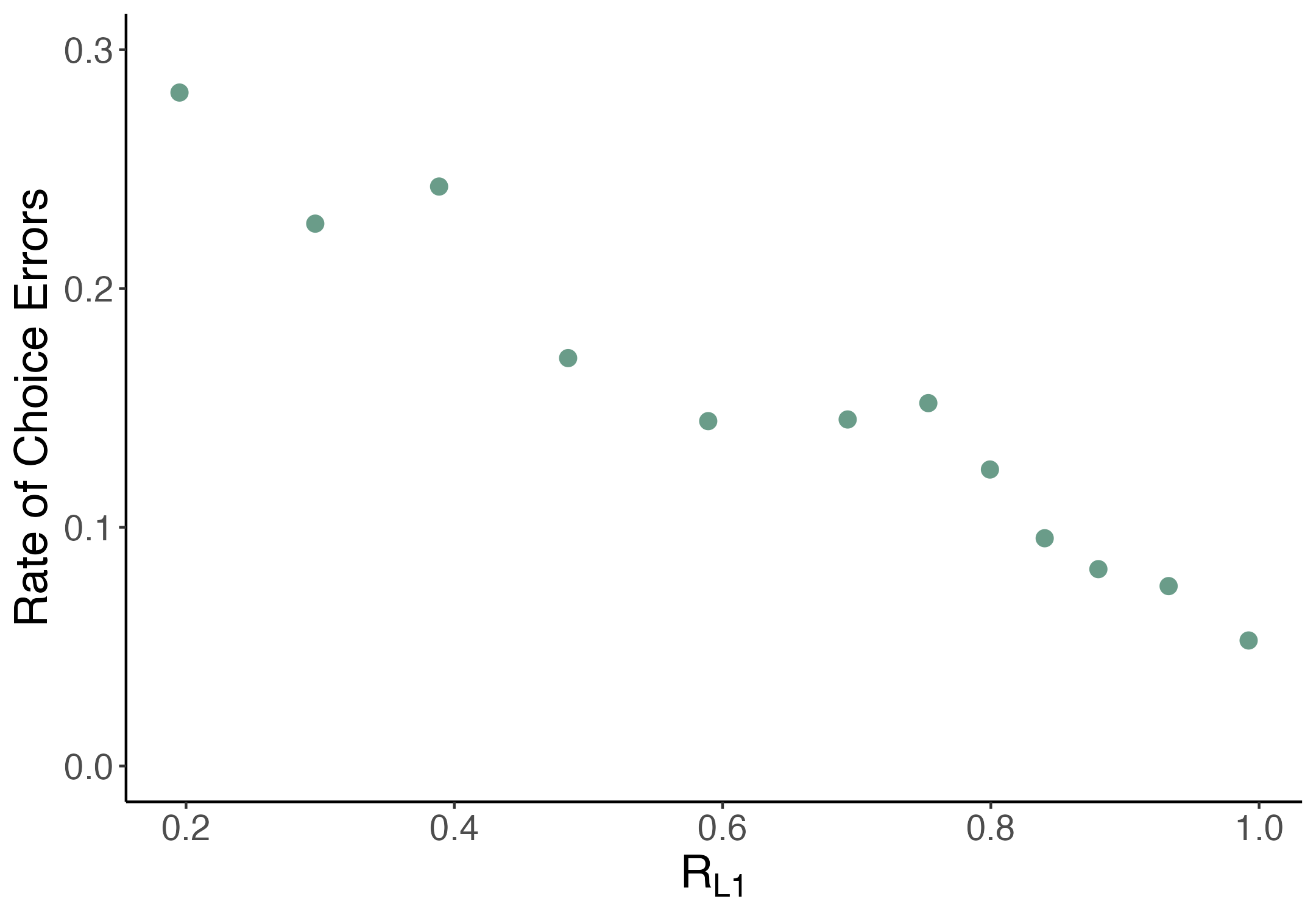}}
    \end{subfigure}
    \begin{subfigure}[t]{0.245\textwidth}
        \caption{\centering Multiattribute Uncertainty }
        \label{subfig:multi_cu}
        \vspace{0.5em}
        {\includegraphics[width=\linewidth]{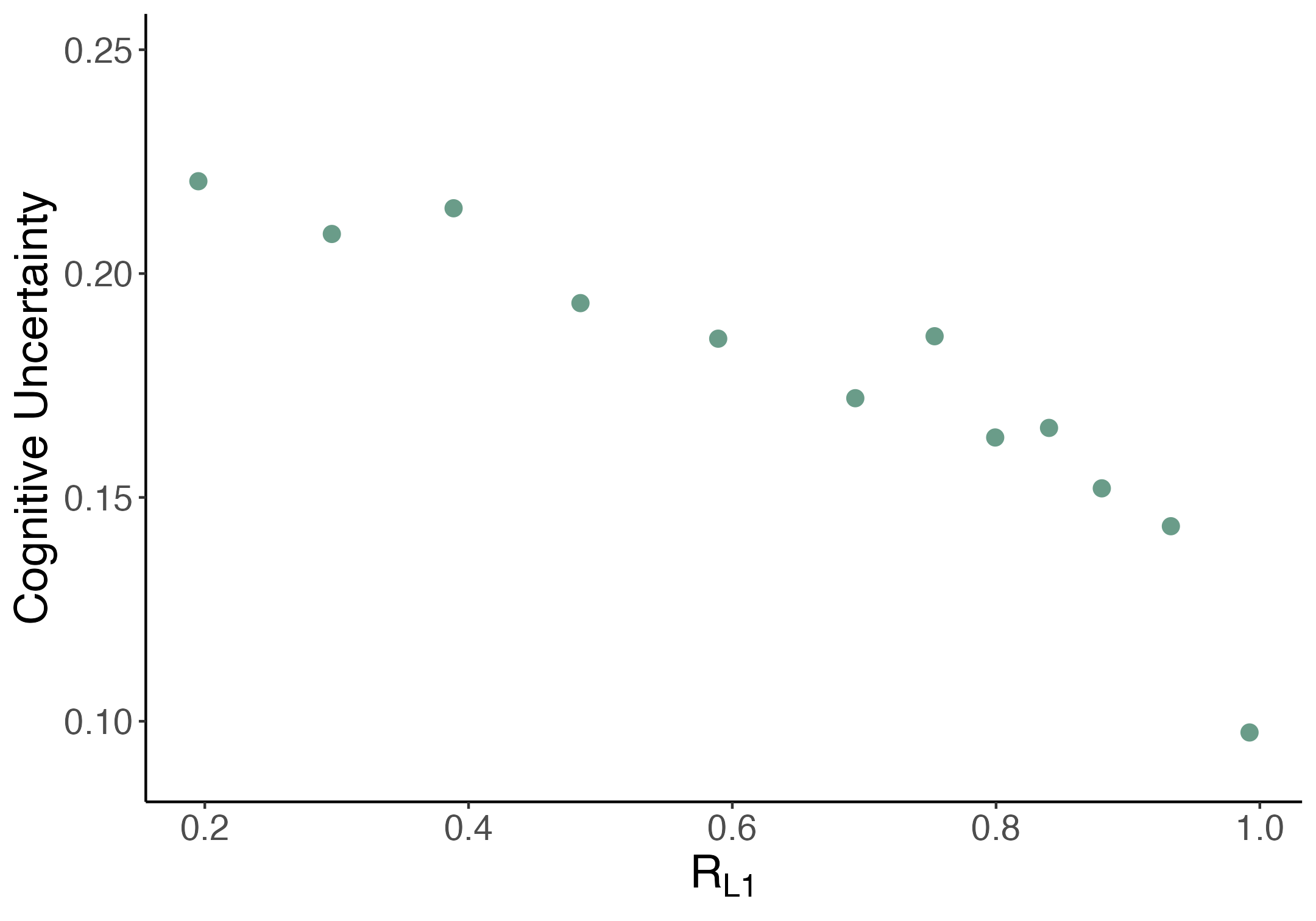}}
    \end{subfigure}
    \begin{subfigure}[t]{0.245\textwidth}
        \caption{\centering Multiattribute 
        Resp. Times}
        \label{subfig:multi_rt}
        \vspace{0.5em}
        {\includegraphics[width=\linewidth]{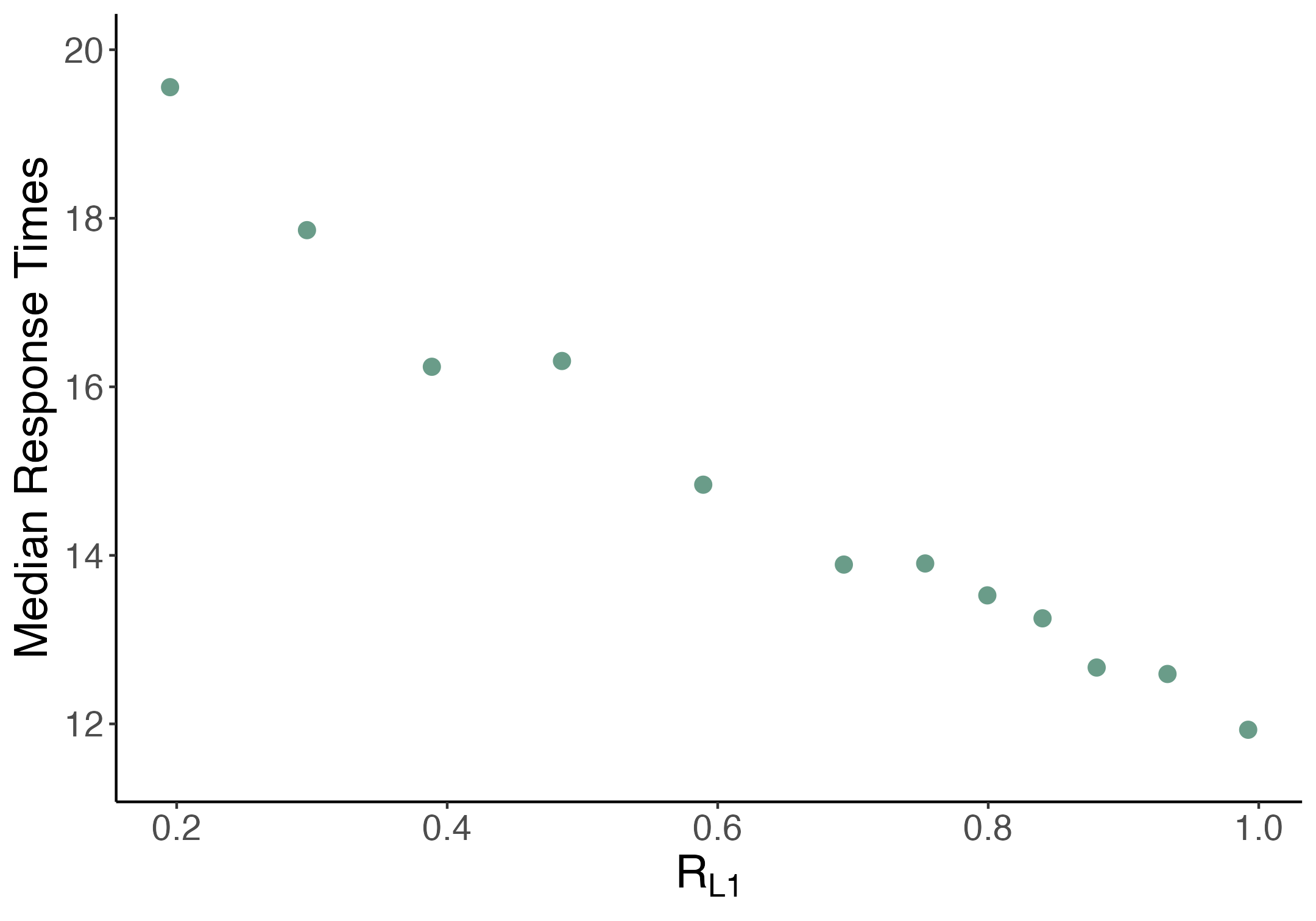}}
    \end{subfigure}
    
    \small
        \begin{subfigure}[t]{0.245\textwidth}
    \centering
    \caption{\centering Lottery Inconsistency}
        \vspace{0.5em}
        {\includegraphics[width=\linewidth]{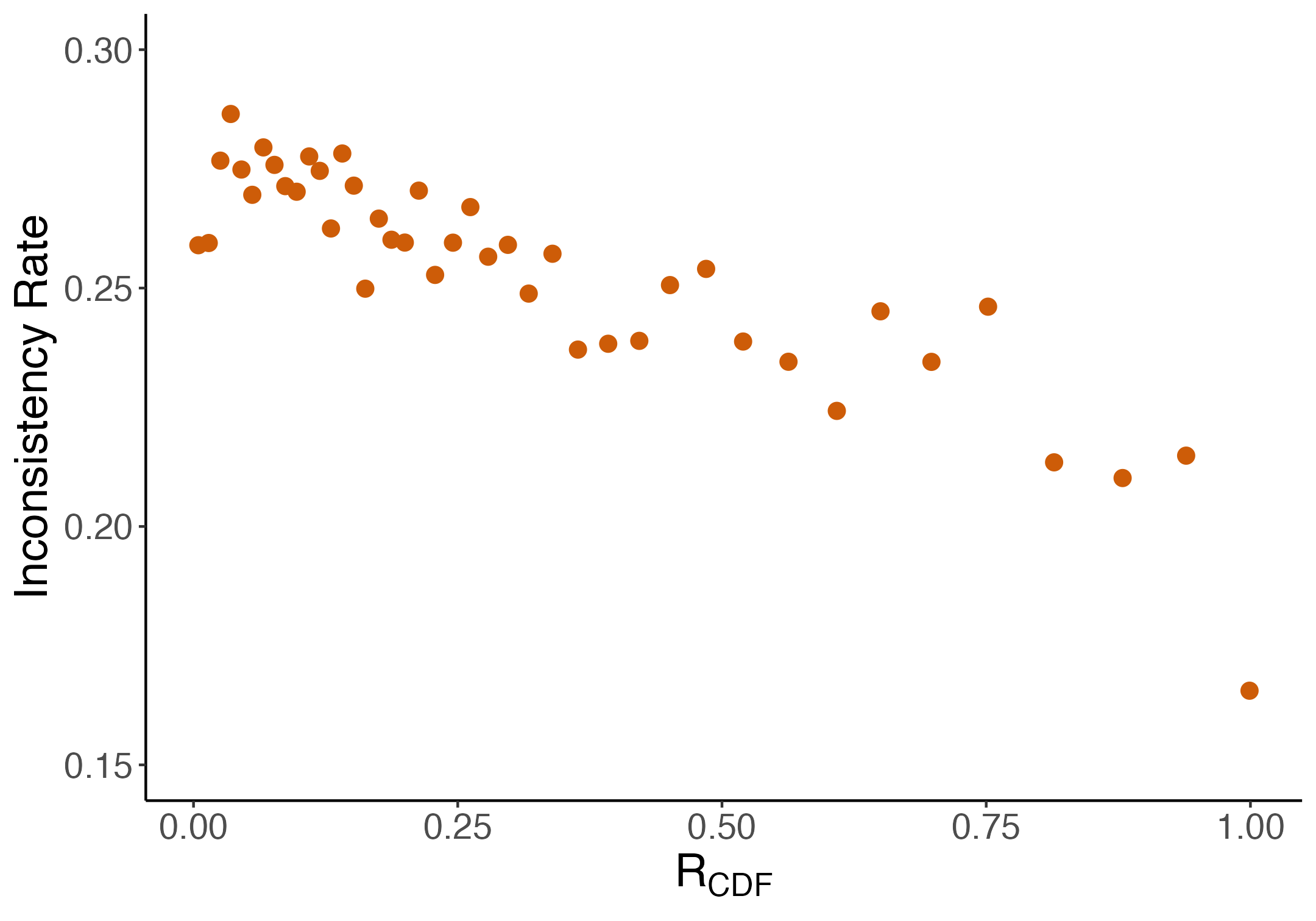}}
        \label{subfig:lottery_incons}
    \end{subfigure}
    \begin{subfigure}[t]{0.245\textwidth}
    \centering
        \caption{\centering Lottery ``Error'' Rates}
        \vspace{0.5em}
        {\includegraphics[width=\linewidth]{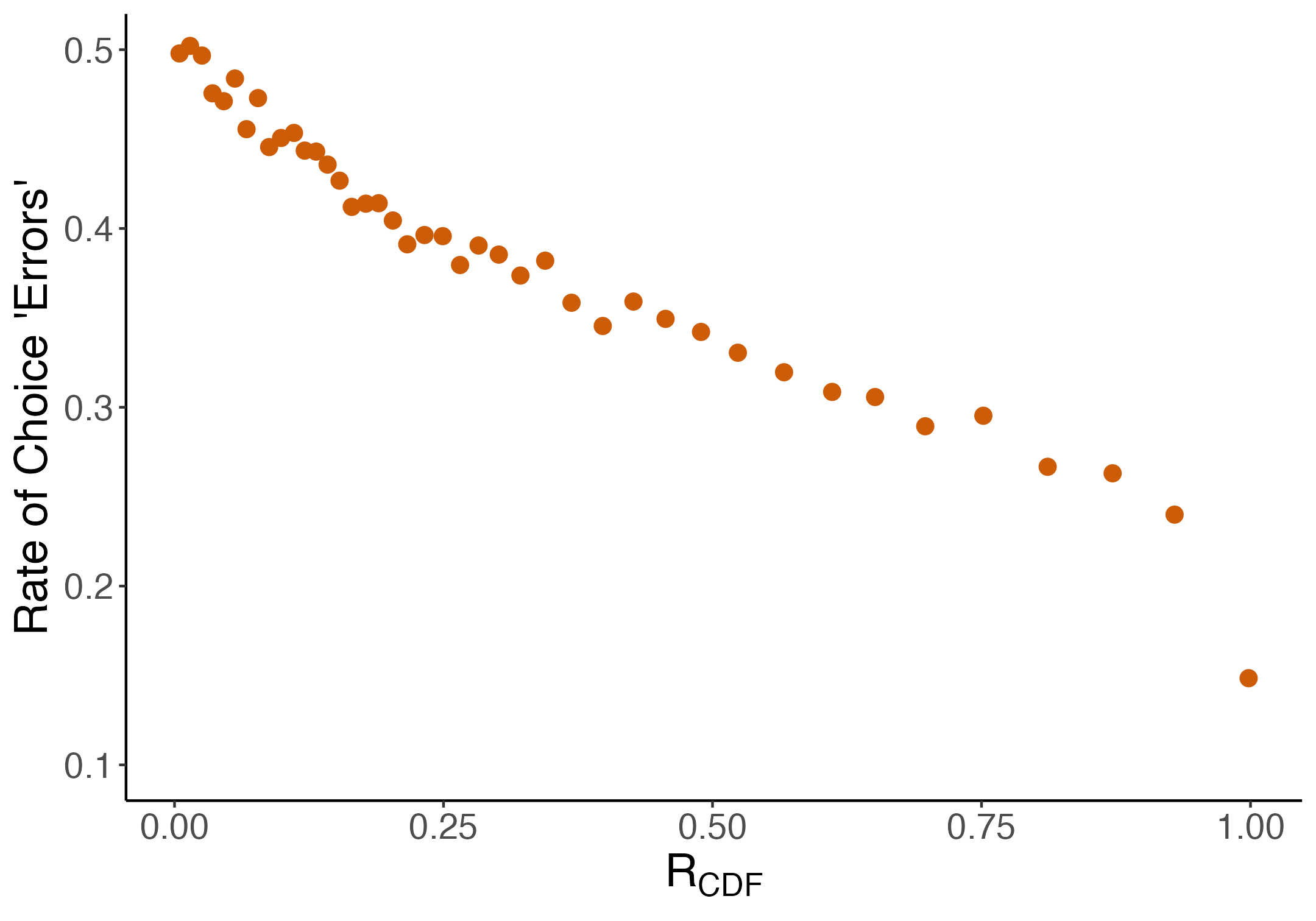}}
        \label{subfig:lottery_errors}
    \end{subfigure}
    \begin{subfigure}[t]{0.245\textwidth}
    \centering
    \caption{\centering Lottery Uncertainty }
        \vspace{0.5em}
        {\includegraphics[width=\linewidth]{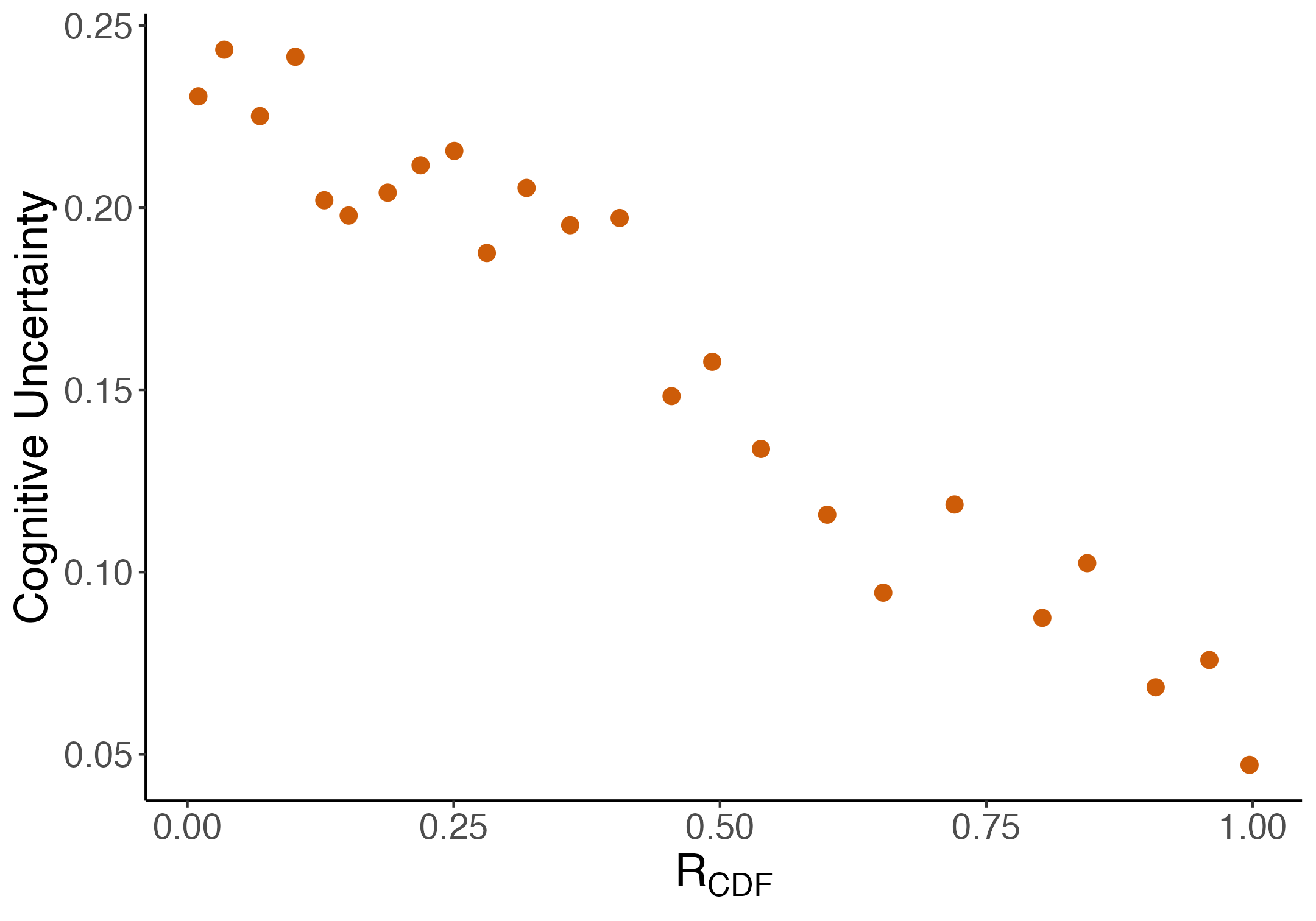}}
        \label{subfig:lottery_cu}
    \end{subfigure}
    \begin{subfigure}[t]{0.245\textwidth}
    \centering
    \caption{\centering Lottery Resp. Times }
        \vspace{0.5em}
        {\includegraphics[width=\linewidth]{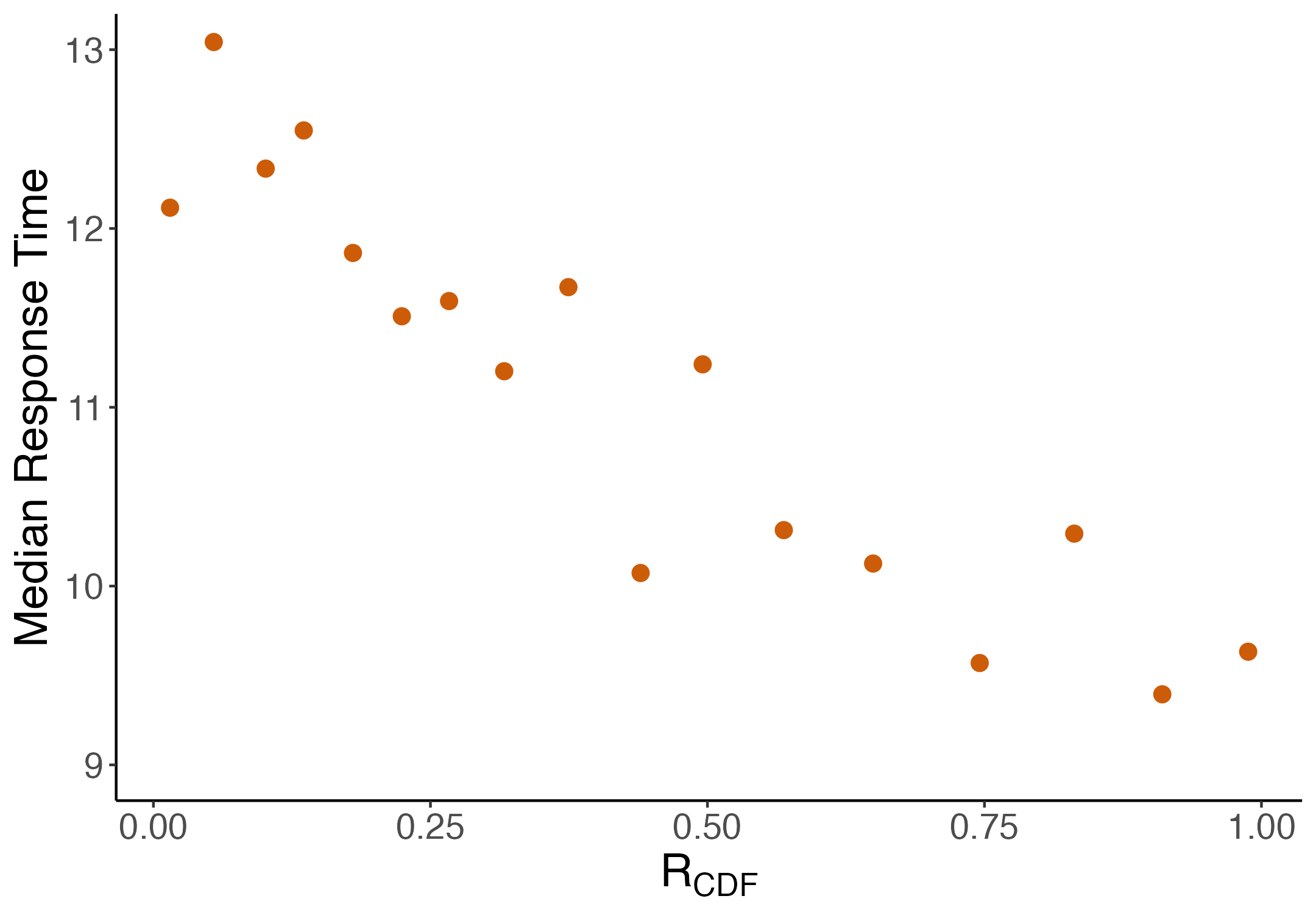}}
        \label{subfig:lottery_rt}
    \end{subfigure}

    \small
    \begin{subfigure}[t]{0.245\textwidth}
    \caption{\centering Intertemporal Inconsistency}
        \vspace{0.5em}
        {\includegraphics[width=\linewidth]{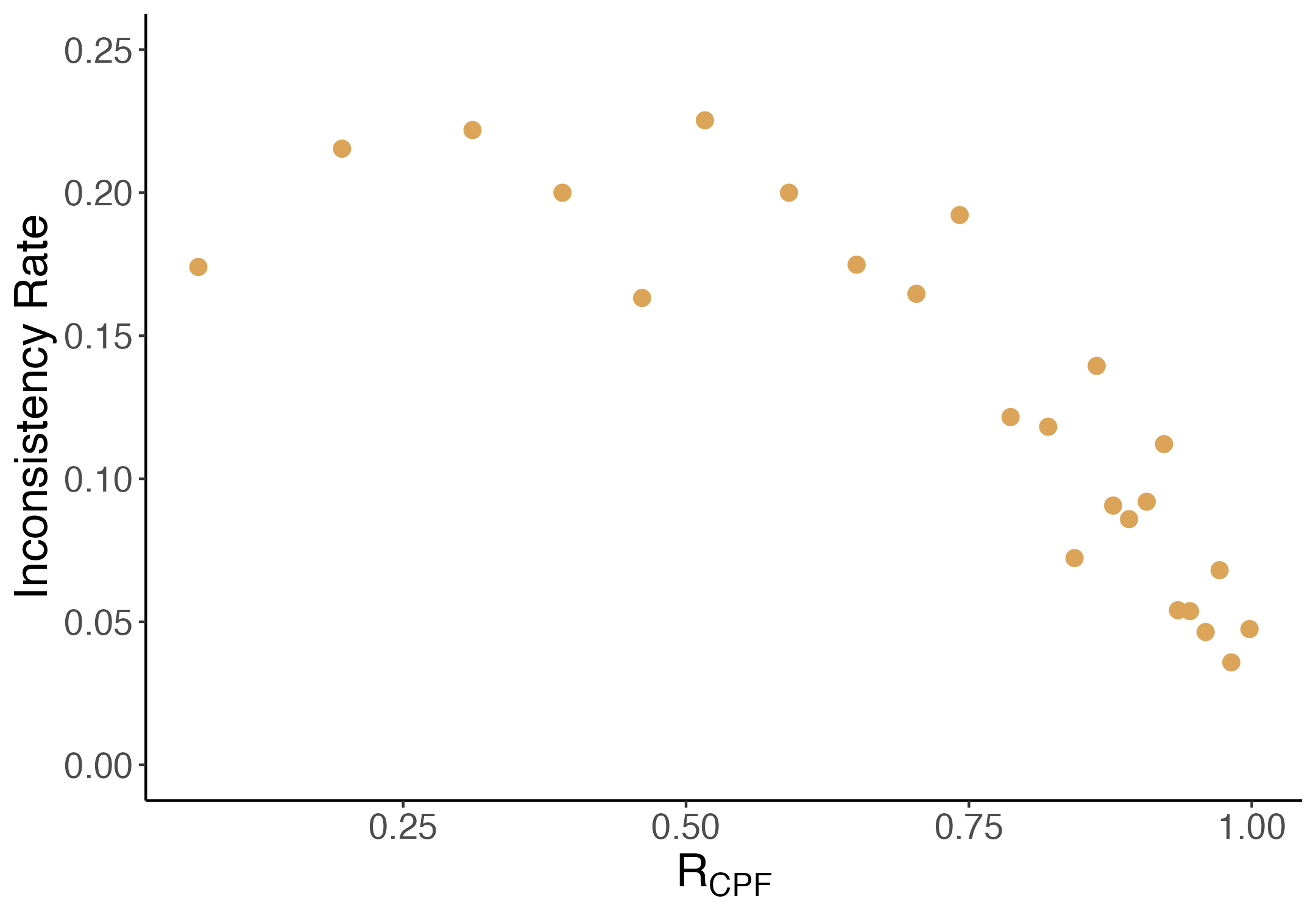}}
        \label{subfig:temp_incons}
    \end{subfigure}
    \begin{subfigure}[t]{0.245\textwidth}
        \caption{\centering Intertemporal ``Error'' Rates}
        \vspace{0.5em}
        {\includegraphics[width=\linewidth]{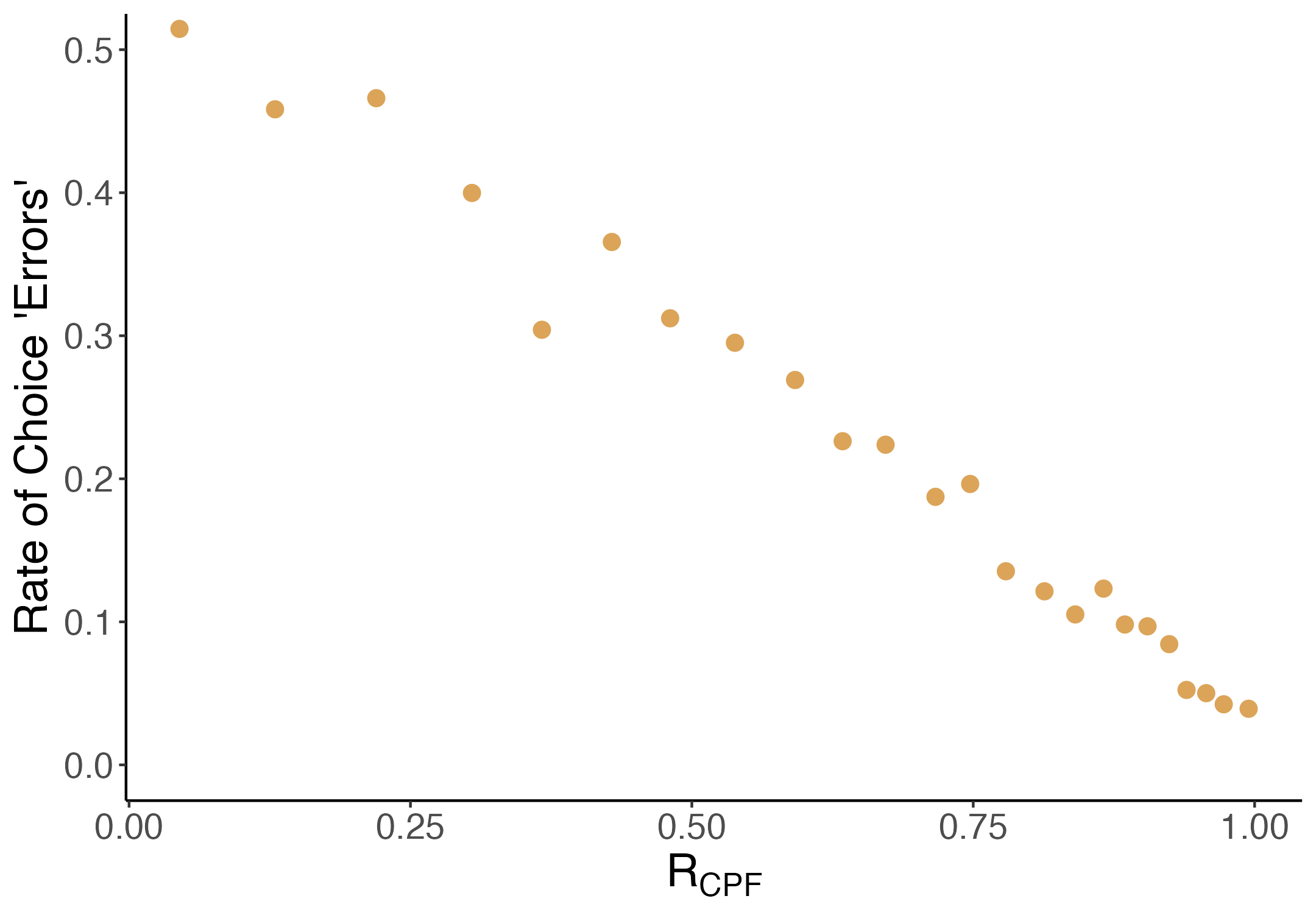}}
        \label{subfig:temp_errors}
    \end{subfigure}
    \begin{subfigure}[t]{0.245\textwidth}
    \caption{\centering Intertemporal Uncertainty }
        \vspace{0.5em}
        {\includegraphics[width=\linewidth]{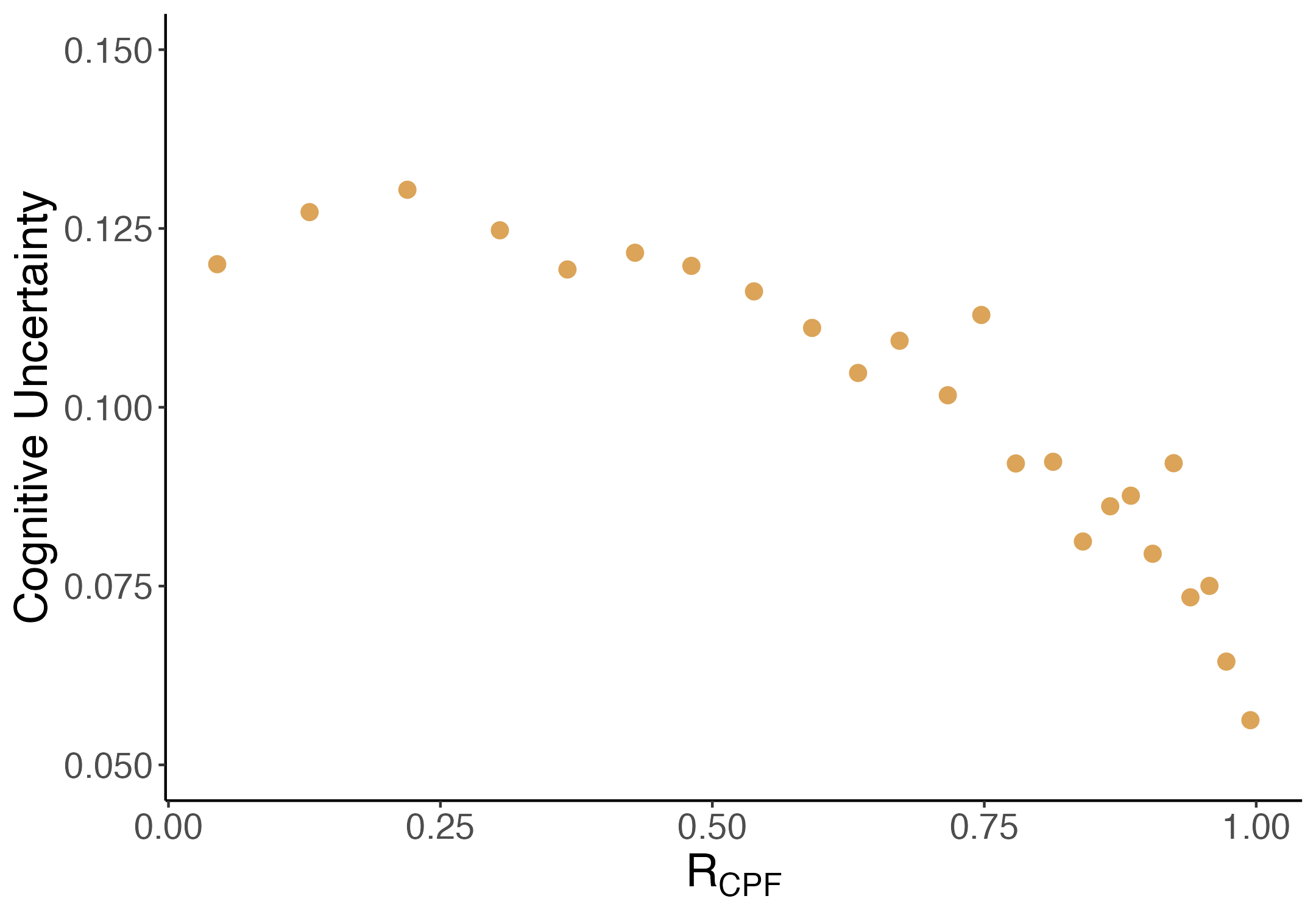}}
        \label{subfig:temp_cu}
    \end{subfigure}
    \begin{subfigure}[t]{0.245\textwidth}
    \centering
    \caption{\centering Intertemporal Resp. Times }
        \vspace{0.5em}
        {\includegraphics[width=\linewidth]{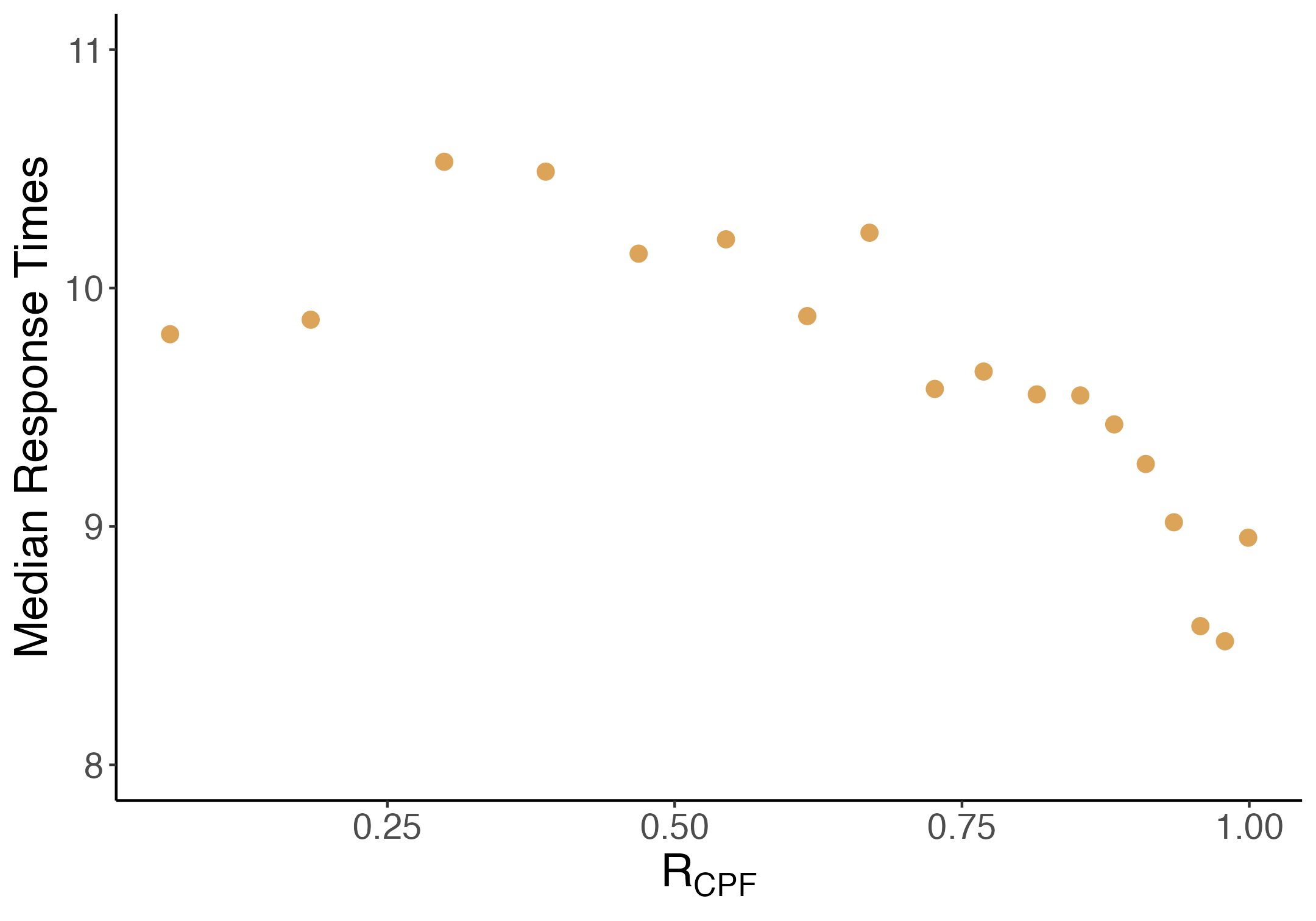}}
        \label{subfig:temp_rt}
    \end{subfigure}
    \vspace{-1em}
    \captionsetup{font=small}\caption{Binscatters of problem-level indicators of choice difficulty against the value-dissimilarity ratio. Inconsistency rates in panels (a), (e), and (i) are the probability of choosing differently across problem repeats. Error rates in (b), (f), and (j) are the rate of choosing, respectively, the lower-value option, the option disfavored by the best-fit expected utility model, and the option disfavored by the best-fit exponential discounting model. Cognitive uncertainty in (c), (g), and (k) is the average subjective likelihood subjects assign to choosing incorrectly. Panels (d), (h), (l) report median response times. $R_{L1}$, $R_{CDF}$, and $R_{CPF}$ follow Definitions \ref{def:L1_complexity}, \ref{def:CDF_complexity}, and \ref{def:CPF_complexity}, using known attribute weights $\beta$, best-fit EU preferences $u$, and the best-fit discount rate $\delta$, respectively. See Appendix Tables \ref{tab:l1regs}, \ref{tab:cdfregs_global}, and \ref{tab:cpfregs_global} for corresponding regression analyses.}
    \label{fig:experiment_plots}
\end{sidewaysfigure}


\noindent \textbf{\textit{Multiattribute Choice}}. Participants make binary choices between hypothetical phone plans with either two, three, or four attributes, which include a device cost, a monthly flat fee, a data usage fee, and a quarterly Wi-Fi fee. Each choice problem has an objective payoff-maximizing answer, which allows us to observe choice errors: subjects choose on behalf of a hypothetical consumer with a known budget and data usage, and are incentivized to choose the plan that minimizes costs for the consumer. If a participant is selected to earn a bonus (1 in 2 chance), we select one of their choices at random and pay them based on the money saved. Across all choice problems, the payoff-maximizing plan yields an average bonus that is \$4.78 higher than the alternative. For more detail on the design and pre-registration, see Appendix \ref{APP:experiments}.\footnote{We pre-registered analyses restricting to subjects who do not report using a calculator in the experiment (82.5\% of the sample). Quantitative results are virtually unchanged when restricting to this sample.}

Figure $\ref{subfig:multi_incons}$ plots the relationship between the $L_1$ ratio and choice inconsistency rates. As predicted, this relationship is decreasing---choices are noisier when they involve greater tradeoffs, as measured by the ratio. This relationship is quantitatively meaningful: the average inconsistency rate, around 5\% for problems near-dominance, increases more than four-fold for problems with the lowest values of the $L_1$ ratio. To provide evidence that this relationship reflects comparison difficulty, Figures \ref{subfig:multi_cu}, \ref{subfig:multi_rt}, and \ref{subfig:multi_errors} relate the $L_1$ ratio to our other indicators of difficulty: cognitive uncertainty, response times, and errors. All three measures are strongly decreasing in the ratio. Importantly, all these relationships are unchanged when controlling for value differences (see Appendix Table \ref{tab:l1regs}), demonstrating that they are driven by variation in the value-dissimilarity \textit{ratio}, rather than simply the proximity to indifference.\\


%

\noindent\textbf{\textit{Lottery Choice}}. In the lottery choice experiments, participants choose between lotteries which pay off different amounts with known probabilities. If selected to receive a bonus, they receive the outcome of a lottery they chose in a randomly selected decision. Here, the CDF ratio takes as input a preference object---the Bernoulli utility function, which we estimate from the data using a CRRA parameterization (see Appendix \ref{APP:structural_risk}).\footnote{Results are robust to a range of alternative preferences used to compute the ratio. For instance, results using risk neutral preferences are similar: the resulting ratio has a correlation of 0.98 with the one used in our analyses.} 

Figure \ref{subfig:lottery_incons} plots the relationship between the CDF ratio and choice inconsistency. As in multiattribute choice, our measure of tradeoff complexity strongly predicts inconsistency rates, which range from 16\% to nearly 30\% for problems with the highest vs. lowest value of the CDF ratio. Figures \ref{subfig:lottery_cu}, \ref{subfig:lottery_rt}, and \ref{subfig:lottery_errors} show that the CDF ratio also strongly predicts all three other indicators of choice difficulty: cognitive uncertainty, response times, and ``errors'', where we code a choice as an error if it departs from the estimated risk preferences used to compute the ratio. The ratio is similarly predictive of errors classified using individually-estimated preferences, which accounts for preference heterogeneity (see Appendix Table \ref{tab:cdf_error_indiv}).\footnote{We restrict this analysis to the \citet{enke_quantifying_2023} data since the \citet{peterson_using_2021} data contain relatively few unique choice problems per subject (75\% of subjects face $\leq 14$ unique problems).} All these relationships are unchanged when controlling for value differences (see Appendix Table \ref{tab:cdfregs_global}), which again shows that they are driven by variation in the value-dissimilarity \textit{ratio}, rather than value differences alone. 

Finally, we repeat the ``errors'' analysis on data from a separate binary choice experiment in \citet{enke_quantifying_2023}, where subjects face incentives that induce risk-neutral preferences over lotteries, which allows us to directly observe objective choice errors. In this data as well, errors are strongly decreasing in the ratio (see Appendix Table \ref{tab:prediction_error}).\\


\noindent\textbf{\textit{Intertemporal Choice}}. Participants choose between time-dated payoff streams. Each option has up to two payoffs ranging between \$1 and \$40, to be received at delays ranging between the present and 2 years in the future. If a participant is selected to earn a bonus (1 in 5 chance), we select a decision at random and pay out the chosen payoff stream on the specified dates. For more details on the design and pre-registration, see Appendix \ref{APP:experiments}. In this domain, the CPF ratio takes as input a discount factor $\delta$, which we estimate from the data (see Appendix \ref{APP:structural_time}).\footnote{The discount factor we estimate should \textit{not} be interpreted as subjects' ``pure'' time preferences \citep{cohen_measuring_2020}: it could reflect, for instance, access to outside credit. Rather than identifying time preferences, we are interested in studying which problems are hard when subjects face tradeoffs between money and delays.}
 
Figure \ref{subfig:temp_incons} plots the relationship between the CPF ratio and choice inconsistency. Here again, the ratio strongly predicts choice inconsistency, which is around 5\% of problems with the highest value of the CPF ratio and increases four-fold for problems with the lowest value. Similarly, Figures \ref{subfig:temp_cu}, \ref{subfig:temp_rt}, and \ref{subfig:temp_errors} show that the ratio again predicts cognitive uncertainty, response times, and choice ``errors'', where here we proxy for errors using our estimated discount factor $\hat\delta$; the ratio is similarly predictive of errors classified using individually-estimated discount factors $\hat{\delta}_i$ (see Appendix Table \ref{tab:cpf_error_indiv}). As in our other domains, all of these relationships persist when controlling for estimated value differences in the comparison (see Appendix Table \ref{tab:cpfregs_global}).\\

\noindent\textbf{\textit{Potential Confounds}}. Our complexity measures strongly predict the rate of choice inconsistency. One concern is that choices are noisier for low-ratio comparisons simply because they are closer to indifference. This cannot explain the results in our induced-values multiattribute experiment: here we can control for the exact value difference in each comparison, and doing so leaves the relationship unchanged. While we do not observe value differences in lottery and intertemporal choice, the relationships in these domains are similarly robust to controlling for estimated value differences, as discussed above; Appendix \ref{APP:incons_model_controls} shows that this is the case even when value differences are estimated under a suite of flexible preference models, such as cumulative prospect theory and hyperbolic discounting.  A second concern is that subject heterogeneity drives the aggregate relationship between complexity and inconsistency. This is ruled out by design: choice inconsistency is measured within-subject, and each subject receives a randomly selected set of problems. Any subject-level differences in inconsistency are orthogonal to the variation in problem complexity we induce via random assignment.

\subsection{Tests of Preference Reversal Predictions}
\label{SEC:reversal_exps}


We run experiments mirroring the simulations in Section \ref{SEC:reversals} to test our model's predictions on preference reversals. We reproduce classic lottery reversals and show that they can be eliminated by manipulating the ease of comparing options to prices. We also validate the novel cross-domain prediction that a similar manipulation eliminates reversals in intertemporal choice.

\begin{table}[b!]
\begin{center}
\begin{tabular}{l|l|l|l}
                   & \textit{Lottery}  &                     & \textit{Intertemporal} \\ \hline
                   & \$4.50 w.p. 98\%  &                     & \$8.25 in 30 days      \\
\textit{Low-risk/delay}  & \$4.75 w.p. 94\%  &  & \$9.50 in 90 days      \\
                   & \$5.00 w.p. 90\%     &                     & \$11.00 in 150 days       \\ \hline
                   & \$19.50 w.p. 23\% &                     & \$24.00 in 630 days       \\
\textit{High-risk/delay} & \$21.25 w.p. 21\% & & \$25.50 in 690 days    \\
                   & \$23.50 w.p. 19\% &                     & \$27.00 in 750 days  \\
\end{tabular}
\end{center}
\caption{Base options used in reversal experiments. We consider every possible pairing of a low-risk, high-risk option for lotteries; and every possible pairing of a low-delay, high-delay option for delayed payments.}
\label{tab:reversals}
\end{table}

In each domain, we construct two sets of ``base'' options: 3 high-risk and 3 low-risk lotteries, and 3 high-delay and 3 low-delay payoff flows (see Table \ref{tab:reversals}). We elicit binary choices between these options, as well as multiple price list (MPL) valuations. In these valuation tasks, we manipulate the ease of comparing each option to the price list: for lotteries, we elicit both certainty equivalents and probability equivalents using a yardstick lottery that pays \$24 with $p$\% chance; for intertemporal choice, we elicit both present value equivalents and time equivalents using a yardstick that pays \$27.5 in $t$ days. We elicit choices and valuations for both the base options as well as a ``scaled-up'' version of each option, where we multiply payouts by a factor of 1.6.


Recall the predictions developed in Section \ref{SEC:reversals}. First, binary choice probabilities will favor low-risk (low-delay) options over high-risk (high-delay) options.\footnote{We calibrated these options so that under the preferences estimated from our binary choice data, the low-risk/low-delay options are preferred.} Second, valuations will be higher for high-risk (high-delay) options compared to low-risk (low-delay) options when they are valued in terms of money. Third, these relative valuations will flip when options are instead valued using probability equivalents (time equivalents). 

We run separate incentivized lottery and intertemporal choice experiments on an online survey platform with 151 and 152 subjects, respectively. Subjects complete two parts, in random order: \textit{Binary Choice} and \textit{Valuation}. In \textit{Binary Choice}, subjects make 16 direct choices between option pairs. In \textit{Valuation}, participants value all 12 options (6 base options at 2 scale factors) using MPLs corresponding to one of two randomly assigned valuation modes: certainty equivalents or probability equivalents in the lottery experiment, and present value equivalents or time equivalents in the intertemporal experiment. See Appendix \ref{APP:experiments_val} for additional design details.

\begin{figure}[b!]
    \centering
    \begin{subfigure}[t]{0.495\textwidth}
        \includegraphics[width=\linewidth]{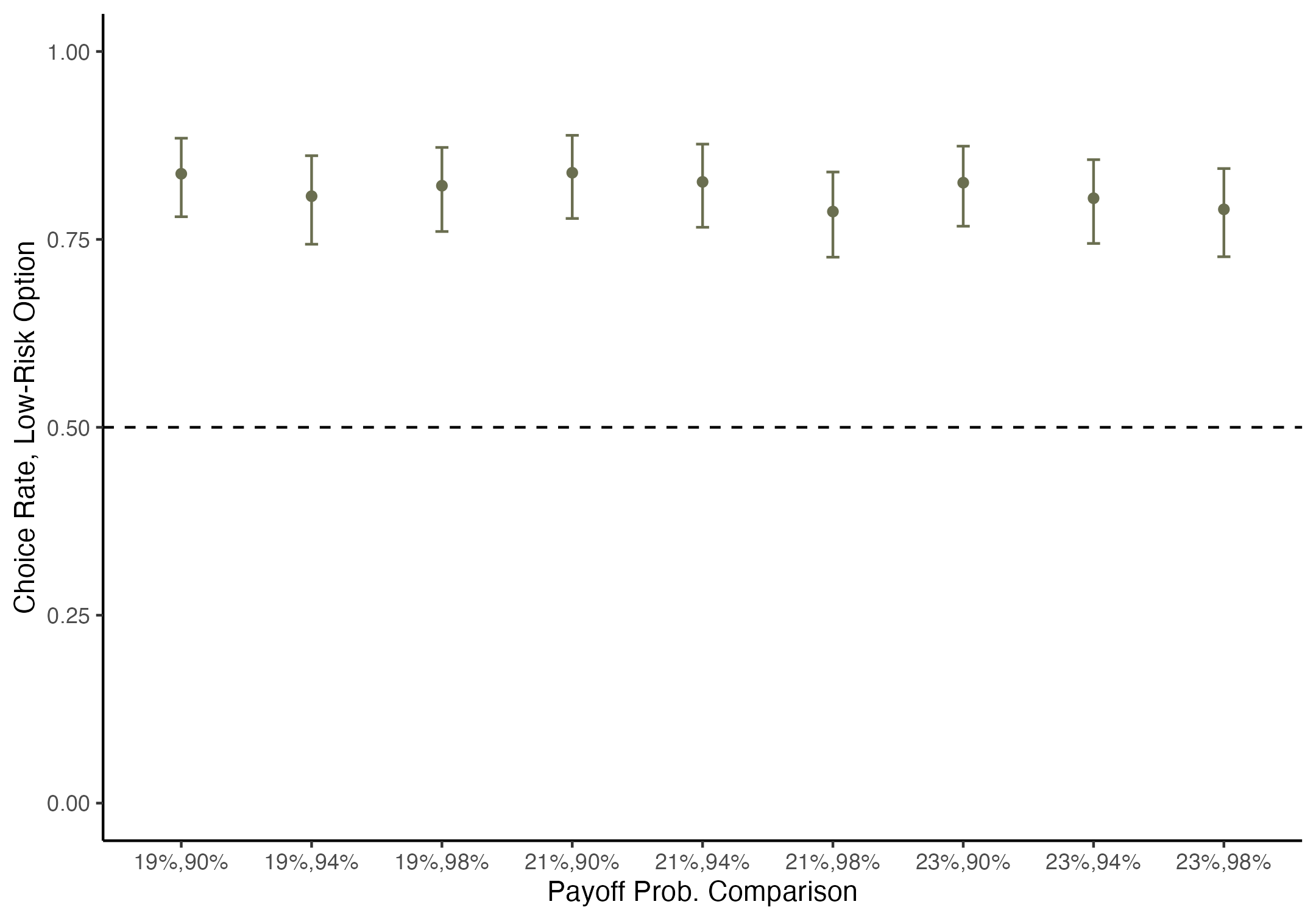}
        \caption{Binary Choice Probabilities}
        \label{fig:risk_pr_choice}
    \end{subfigure}
    \begin{subfigure}[t]{0.495\textwidth}
        \includegraphics[width=\linewidth]{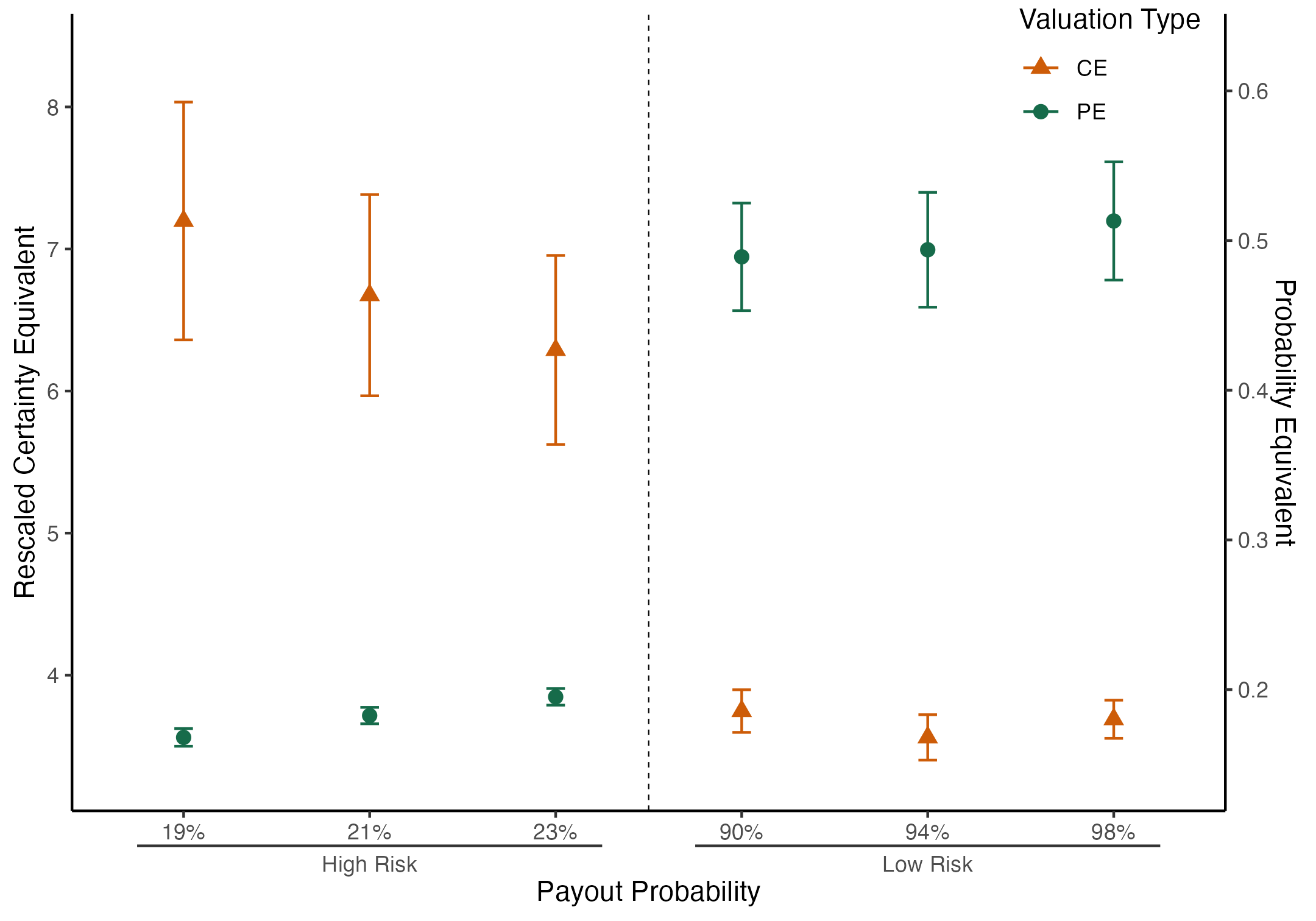}
        \caption{Certainty and Probability Equivalents}
        \label{fig:risk_pr_val}
    \end{subfigure}
    \caption{Preference reversal experiments for simple lotteries, aggregated across scale factor. Panel (a) presents binary choice rates for each high/low risk lottery comparison. Panel (b) presents average rescaled certainty equivalents, computed by dividing each certainty equivalent by the scale factor (scale on left axis), and average probability equivalents (scale on right axis). Error bars reflect 95\% confidence intervals.}
    \label{fig:risk_pr}
\end{figure}
\begin{figure}[t!]
    \centering
    \begin{subfigure}[t]{0.495\textwidth}
        \includegraphics[width=\linewidth]{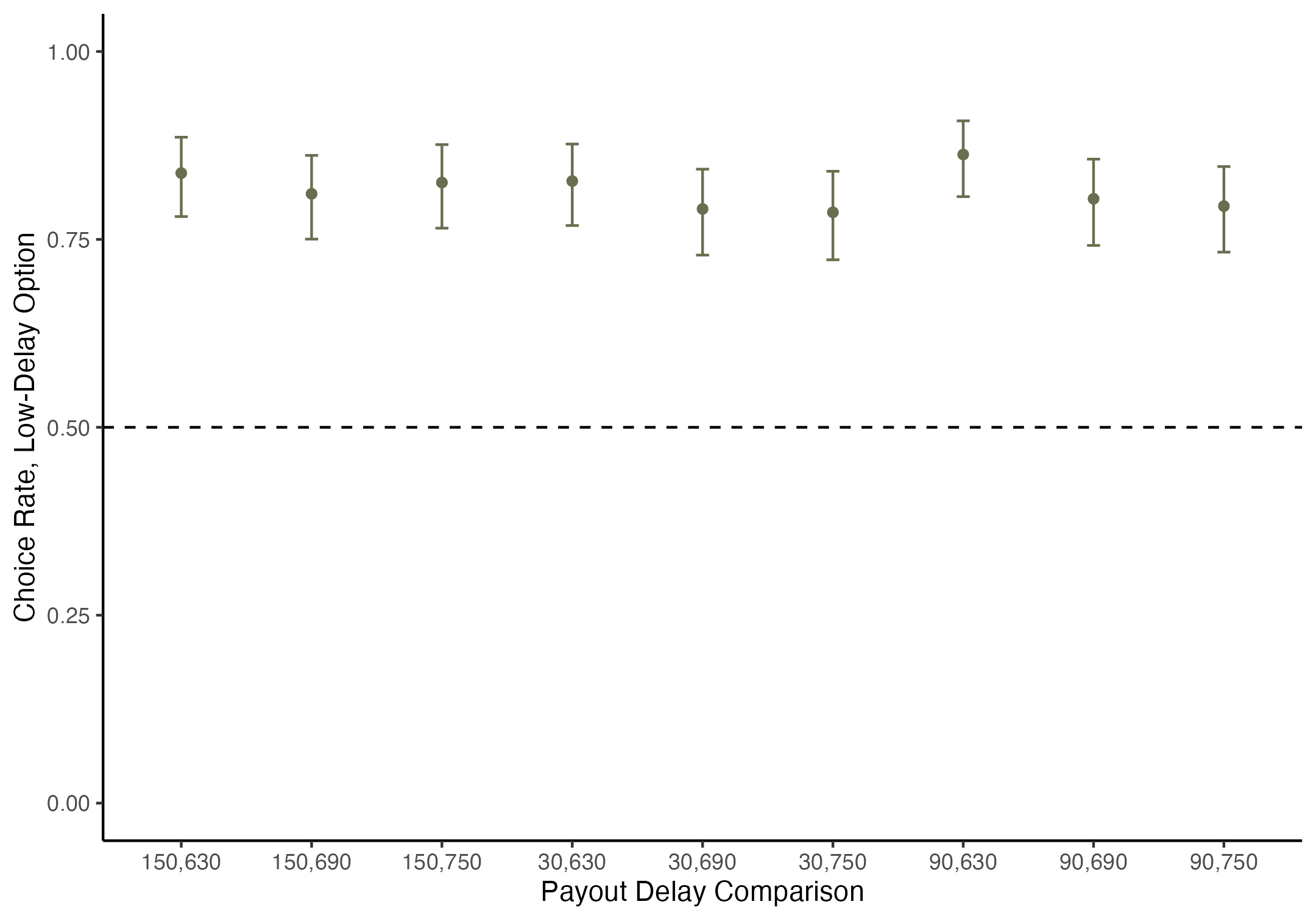}
        \caption{Binary Choice Probabilities}
        \label{fig:time_pr_choice}
    \end{subfigure}
    \begin{subfigure}[t]{0.495\textwidth}
        \includegraphics[width=\linewidth]{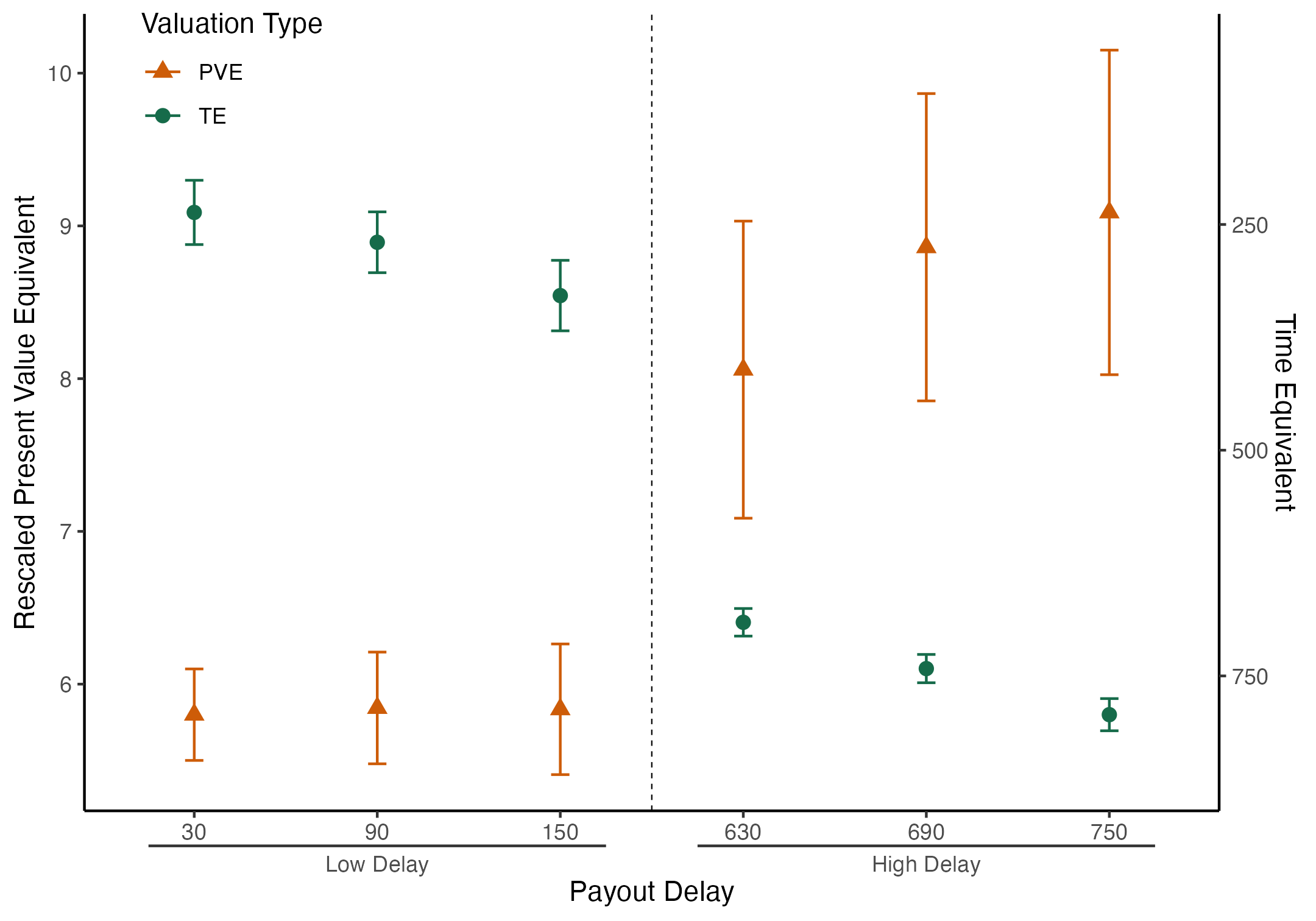}
        \caption{Present Value and Time Equivalents}
        \label{fig:time_pr_val}
    \end{subfigure}   
    \caption{Preference reversal experiments for delayed payments, aggregated across scale factor. Panel (a) presents binary choice rates for each high/low delay comparison. Panel (b) presents average rescaled present value equivalents, computed by dividing each present value equivalent by the scale factor (scale on left axis) and average time equivalents (scale on right axis). The direction of the time equivalent axis is inverted, so that valuations are increasing in the vertical axis. Error bars reflect 95\% confidence intervals.}
    \label{fig:time_pr}
\end{figure}

Figures \ref{fig:risk_pr} and \ref{fig:time_pr} present experimental results, aggregating across scale factors. Figures \ref{fig:risk_pr_choice} and \ref{fig:time_pr_choice} show that binary choice favors the low-risk and low-delay options, with choice rates for those options well above 50\%. In contrast, Figures \ref{fig:risk_pr_val} and \ref{fig:time_pr_val} indicate an apparent reversal in preference when options are valued via certainty equivalents and present value equivalents: subjects assign higher valuations to high-risk and high-delay options on average. In these same figures, however, we see the predicted \textit{flipping} of these valuation patterns when we manipulate the ease of comparing each option to the price list. When the same options are instead valued via probability equivalents and time equivalents, subjects instead assign higher valuations to the low-risk and low-delay options, thus eliminating the aggregate reversal. Therefore, we find that the same patterns of reversals in lottery choice also extend to intertemporal choice, providing out-of-sample evidence for the domain-general mechanism in our framework.



Similar patterns hold when we replace average valuations with medians, as documented in Appendix Tables \ref{tab:median_differences_risk} and~\ref{tab:median_differences_time}. Whereas choices favor the low-risk/delay option, median certainty and present value equivalents favor the high-risk/delay option in almost all comparisons, and changing the numeraire eliminates this reversal: median probability and time equivalents favor the low risk/delay option in all comparisons. These results are inconsistent with alternative noise-based explanations of reversals such as \citet{blavatskyy_preference_2009}, which impose consistency between modal choices and median valuations. \\

\noindent\textbf{\textit{Mechanisms}}. Our model predicts that our numeraire manipulations flip valuation patterns by changing the relative ease of valuation. To shed light on this mechanism, we elicited subjects' cognitive uncertainty (CU) over their valuations, following \citet{enke_cognitive_2023}. This measure of the subjective difficulty of valuations changes in the predicted direction. Compared to low-risk lotteries, average CU for high-risk lotteries is 56\% higher under certainty equivalents (\(p<0.001\)), but 16\% lower under probability equivalents (\(p=0.017\)). Compared to low-delay options, average CU for high-delay options is 10\% higher under present value equivalents (\(p=0.178\)), but 32\% lower under time equivalents (\(p<0.001\)).

Moreover, the reported difficulty of valuing options against money predicts the propensity of preference reversals, consistent with the mechanism in our framework.  Among subjects assigned to value options against money, the rate of within-subject preference reversals---the share of comparisons in which choice favors the low-risk/delay option yet valuations against money favor the high-risk/delay option---is 15.6 p.p. higher for subjects with above-median average CU in lottery choice $(p=0.065)$ and 16.2 p.p. higher in intertemporal choice $(p=0.042)$.


\subsection{Tests of Valuation Predictions}
\label{SEC:valuation_exps}

We experimentally test the prediction, developed in Section \ref{SEC:valuation_biases}, that inverse-S probability weighting and hyperbolic discounting can be reversed by manipulating the units of valuation. Specifically, our model predicts that relative to certainty equivalents, probability equivalents should exhibit underweighting of low probabilities and overweighting of high probabilities. Likewise, relative to the discounting revealed by present value equivalents, time equivalents should exhibit overvaluation of low delays and undervaluation of high delays. 

We recruit subjects through an online survey platform to complete 24 incentivized multiple price lists: 12 for lotteries, and 12 for intertemporal payments. Each subject is randomly assigned to a valuation mode for each domain: certainty or probability equivalents for lotteries, and present value or time equivalents for intertemporal payments. Our sample contains 302 subjects. See Appendix \ref{APP:experiments_val} for additional design details.
\\
\\
%
\textit{\textbf{Lottery Valuations.}} We elicit certainty equivalents $CE(l)$ for a simple lottery $l = (\overline{w},p_l)$, and relate the normalized valuations $CE(l)/\overline{w}$ to payout probabilities $p_l$. We also elicit probability equivalents $PE(c)$ of a certain payment $c=(w_c,1)$ against a yardstick lottery $(\overline{w}, p)$, and relate the normalized payouts $w_c/\overline{w}$ to probability equivalents $PE(c)$.  We draw $\overline{w}$ from $\{\$9, \$18, \$27\}$. For CEs, we draw $p_l$ from $\{0.03, 0.05, 0.10, 0.25, 0.5, 0.75, 0.90, 0.95, 0.97\}$. For PEs, we draw $w_c$ so that $w_c/\overline{w}\in \{0.033, 0.056, 0.11, 0.25, 0.5, 0.75, 0.89, 0.944, 0.967\}$.

\begin{figure}[t!]
    \centering
    \includegraphics[width = 0.7\textwidth]{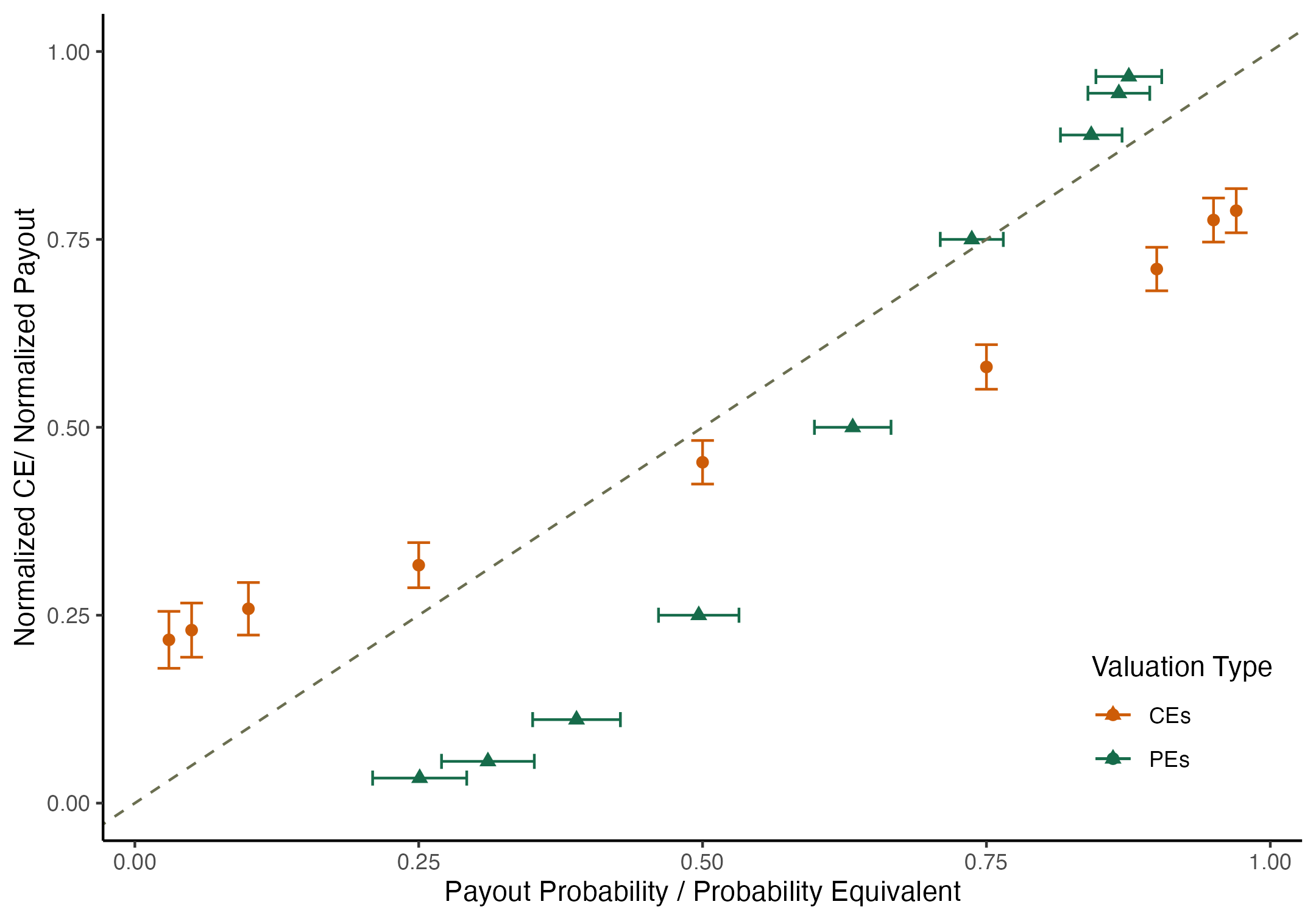}
    \caption{Lottery valuation results. Orange dots plot the payout probability $p_l$ against the average normalized certainty equivalent $CE(l)/\overline{w}$; turquoise triangles plot average probability equivalents $PE(c)$ against the normalized payout $w_c/\overline{w}$. The dashed black line represents linear probability weighting. Error bars reflect 95\% confidence intervals.}
    \label{fig:pwf_experiment}
\end{figure}

Figure \ref{fig:pwf_experiment} plots average valuations. Certainty equivalents follow the standard inverse-S pattern: valuations are risk-seeking over low probabilities and risk-averse over high probabilities. Consistent with model predictions and existing evidence, the pattern \textit{reverses} for probability equivalents: relative to certainty equivalents, probability equivalents are more risk-averse over low probabilities, and more risk-seeking over high probabilities. Having established this result, we turn to assessing a novel cross-domain prediction of our framework: the analogous manipulation in intertemporal choice should \textit{also} reverse hyperbolic discounting. \\

\noindent\textit{\textbf{Intertemporal Valuations}}. We elicit present value equivalents $PVE(\upsilon)$ for a delayed payment $\upsilon = (\overline{m}, t_{\upsilon})$, relating normalized valuations $PVE(\upsilon)/\overline{m}$ to payout delays $t_{\upsilon}$. We also elicit time equivalents $TE(c)$ of an immediate payment $c=(m_c,0)$ against a yardstick delayed payment $(\overline{m}, t)$, relating normalized payouts $m_c/\overline{m}$ to time equivalents $TE(c)$. We draw $\overline{m}$ from $\{25, 30, 35\}$. For PVEs, we draw $t_{\upsilon}$ from $\{7, 30, 60, 120, 240, 360, 480, 720, 1080\}$ (in days). For TEs, we draw $m_c$ so that $m_c/\overline{m}\in \{0.20, 0.35, 0.50, 0.65, 0.75, 0.85, 0.90, 0.95, 0.97\}$.

Results are presented in Figure \ref{fig:disc_experiment}. Focusing first on the discount function implied by present value equivalents, we document the familiar pattern of short-run impatience and long-run patience (i.e. hyperbolicity). We also see that this hyperbolicity is \textit{more extreme} than the hyperbolic discount function estimated from our binary choice experiment (traced in the dotted line; see Appendix \ref{APP:structural_time} for estimation details)---consistent with the interpretation that the hyperbolicity governing subjects' binary choices is exaggerated when measured through present value equivalents. Moreover, when instead we elicit valuations using time equivalents, the pattern of hyperbolicity \textit{reverses}: relative to the discounting implied by their present value equivalents, subjects exhibit short-run patience and long-run impatience. \\

\begin{figure}[t!]
    \centering
    \includegraphics[width = 0.7\textwidth]{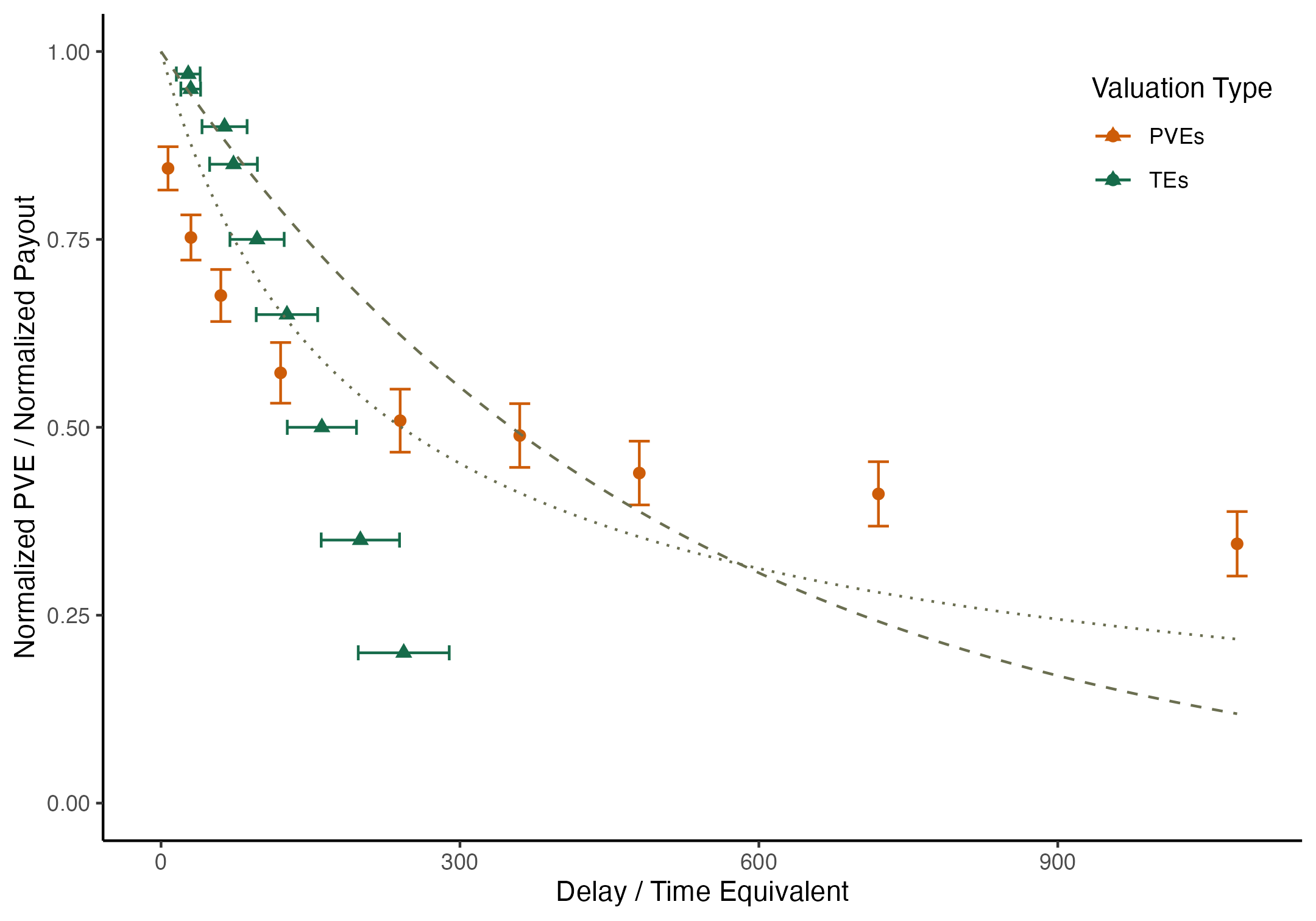}
    \caption{Intertemporal valuation results. Orange dots plot the payout delay $t_{\upsilon}$ against the average normalized present value equivalent $PVE(\upsilon)/\overline{m}$; turquoise triangles plot average time equivalents $TE(c)$ against the normalized payout $m_c/\overline{m}$.  The dashed (dotted) line traces the exponential (hyperbolic) discount function estimated from binary choice data. Error bars reflect 95\% confidence intervals.}
    \label{fig:disc_experiment}
\end{figure}

\noindent \textbf{\textit{Mechanisms}}. That the same manipulation reverses canonical valuation patterns in both risk and time provides novel evidence that they may share a common mechanism. Here, we provide additional evidence for this mechanism. Recall our framework's explanation for these patterns: as the valuation task moves away from a dominance boundary---e.g. a payout probability of 0 or 1 in certainty equivalents or a delay of 0 in present value equivalents---comparisons to prices involve increasingly difficult tradeoffs, producing pull-to-center effects. Under this account, we should expect corresponding variation in the subjective difficulty of valuations. Appendix Figure \ref{fig:valuation_cu} shows that this is indeed the case: in all four valuation paradigms, cognitive uncertainty over valuations monotonically increases away from dominance boundaries. For instance, cognitive uncertainty for certainty equivalents increases as the payout probability moves away from 0\% and 100\%, peaking at an interior payout probability.

\subsection{Benchmarking Model Performance}
\label{SEC:benchmarking}

Whereas the predominant approach in behavioral economics is to model non-standard elements in the DM's value function, this paper models an orthogonal feature of choice: how difficult the DM finds a choice, given her objectives. To study whether our theory indeed explains variation in choice that is not captured by preference-based models, we compare the explanatory power of these modeling approaches in our binary choice data. 

First, we structurally estimate our choice model, which is parameterized by a transformation $G$ that maps the value-dissimilarity ratio into choice probabilities and (in the case of lottery and intertemporal choice) domain-specific preference parameters. Across all domains, we use the same specification of $G$ developed in Section \ref{SEC:theory_params}. We then compare performance against a set of models specifying the DM's value function, which we fit to our data using logit errors. In particular, we estimate a ``standard'' benchmark model: a distortion-free logit model in multiattribute choice, exponential discounted utility in intertemporal choice, and expected utility in lottery choice. We also estimate a set of behavioral models: salience-weighting \citep{bordalo_salience_2013}, focusing \citep{koszegi_model_2013}, and relative thinking \citep{bushong_model_2021} in multiattribute choice; quasi-hyperbolic and hyperbolic discounting in intertemporal choice; and simplicity theory \citep{puri_simplicity_2025} and cumulative prospect theory in lottery choice. Estimates are obtained using maximum likelihood; see Appendix \ref{APP:structural} for details on each specification.

\begin{figure}[b!]
\small
\begin{subfigure}[t]{0.33\textwidth}
\centering
    \caption{\centering Multiattribute}
    \label{fig:model_comp_multi}
    \vspace{0.5em}
    {\includegraphics[width=55mm]{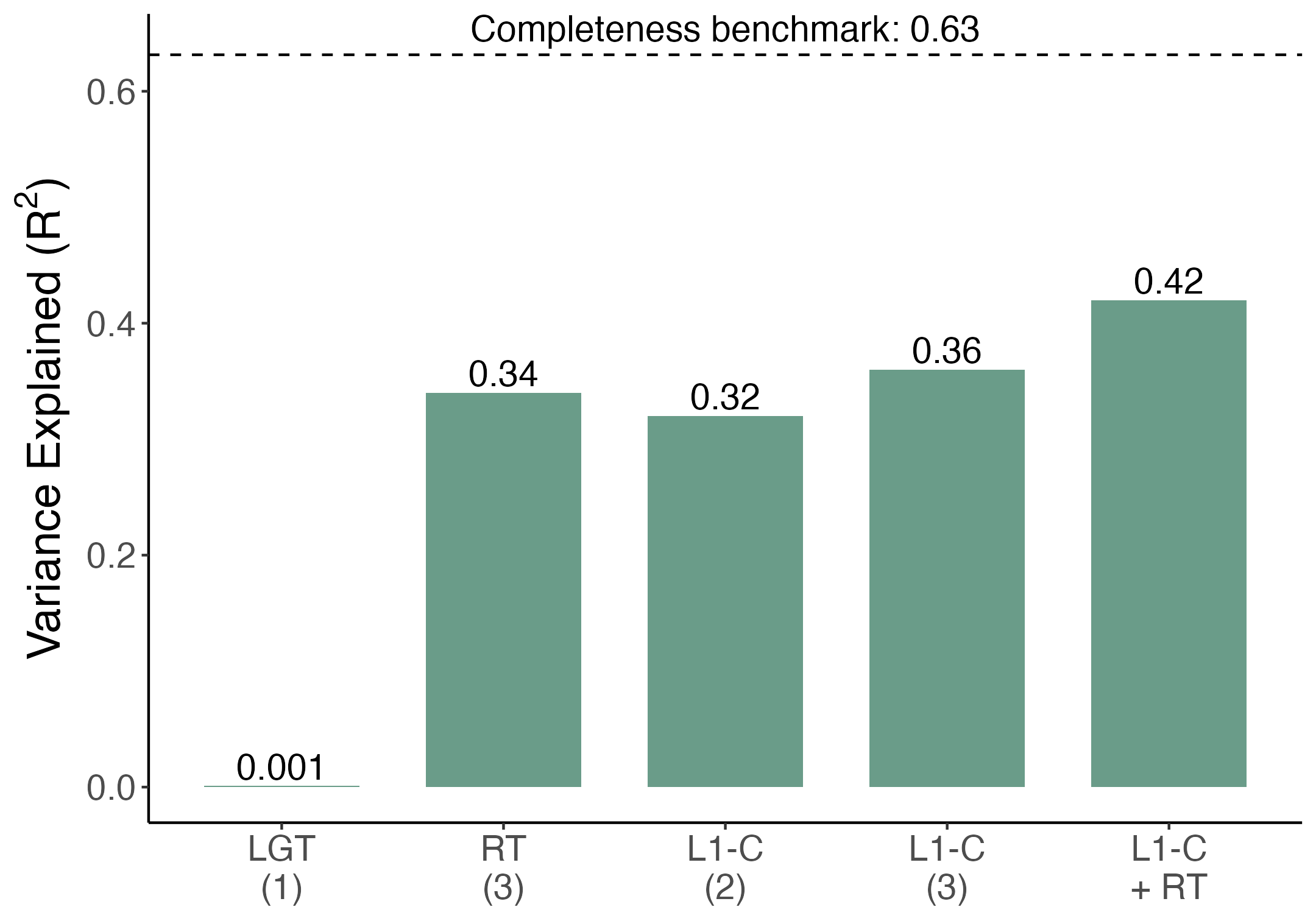}}
\end{subfigure}
\begin{subfigure}[t]{0.33\textwidth}
\centering
\caption{\centering Intertemporal }
    \label{fig:model_comp_temp}
    \vspace{0.5em}
    {\includegraphics[width=55mm]{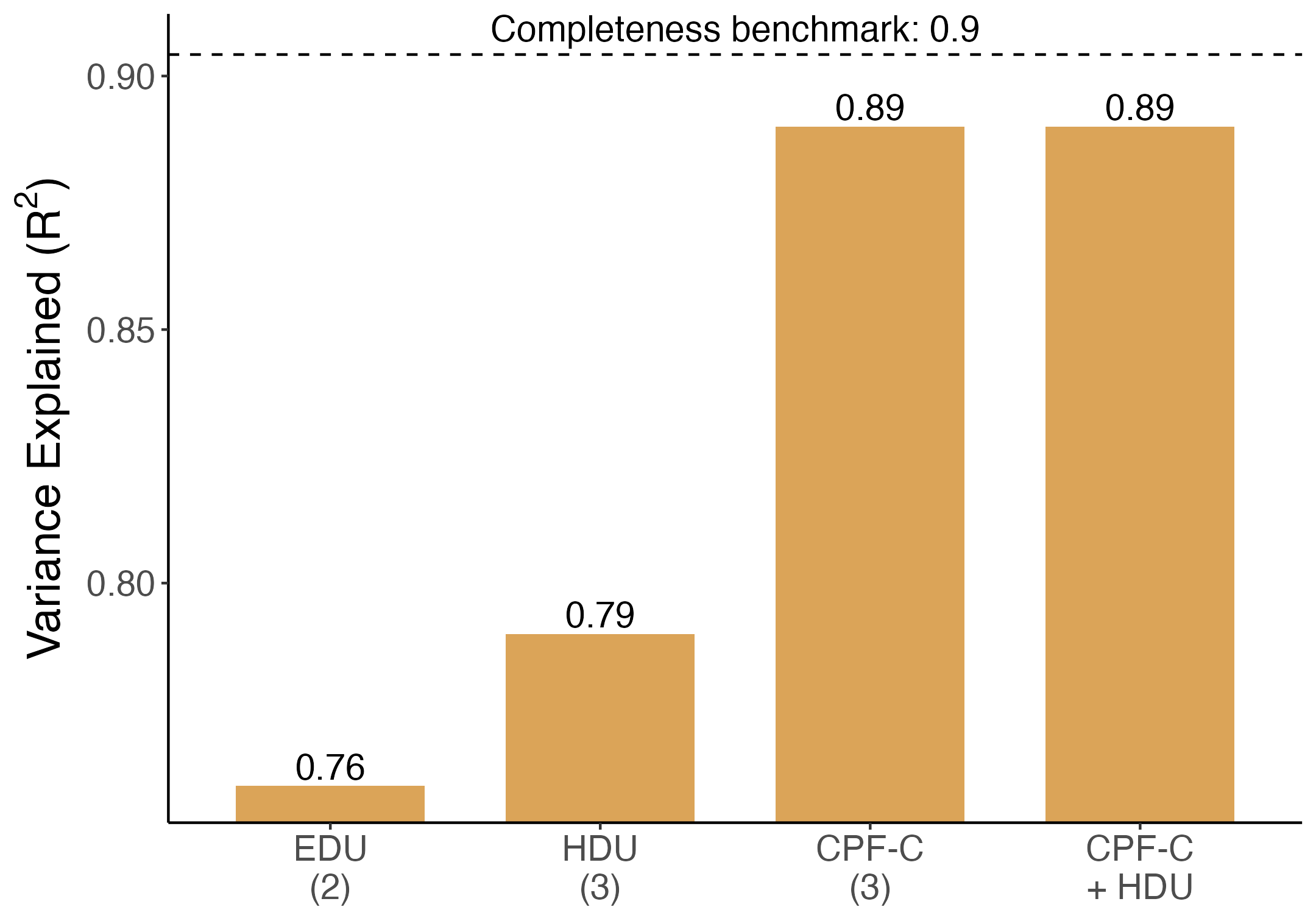}}
\end{subfigure}
\begin{subfigure}[t]{0.33\textwidth}
\centering
\caption{\centering Lottery}
    \label{fig:model_comp_lottery}
    \vspace{0.5em}
    {\includegraphics[width=55mm]{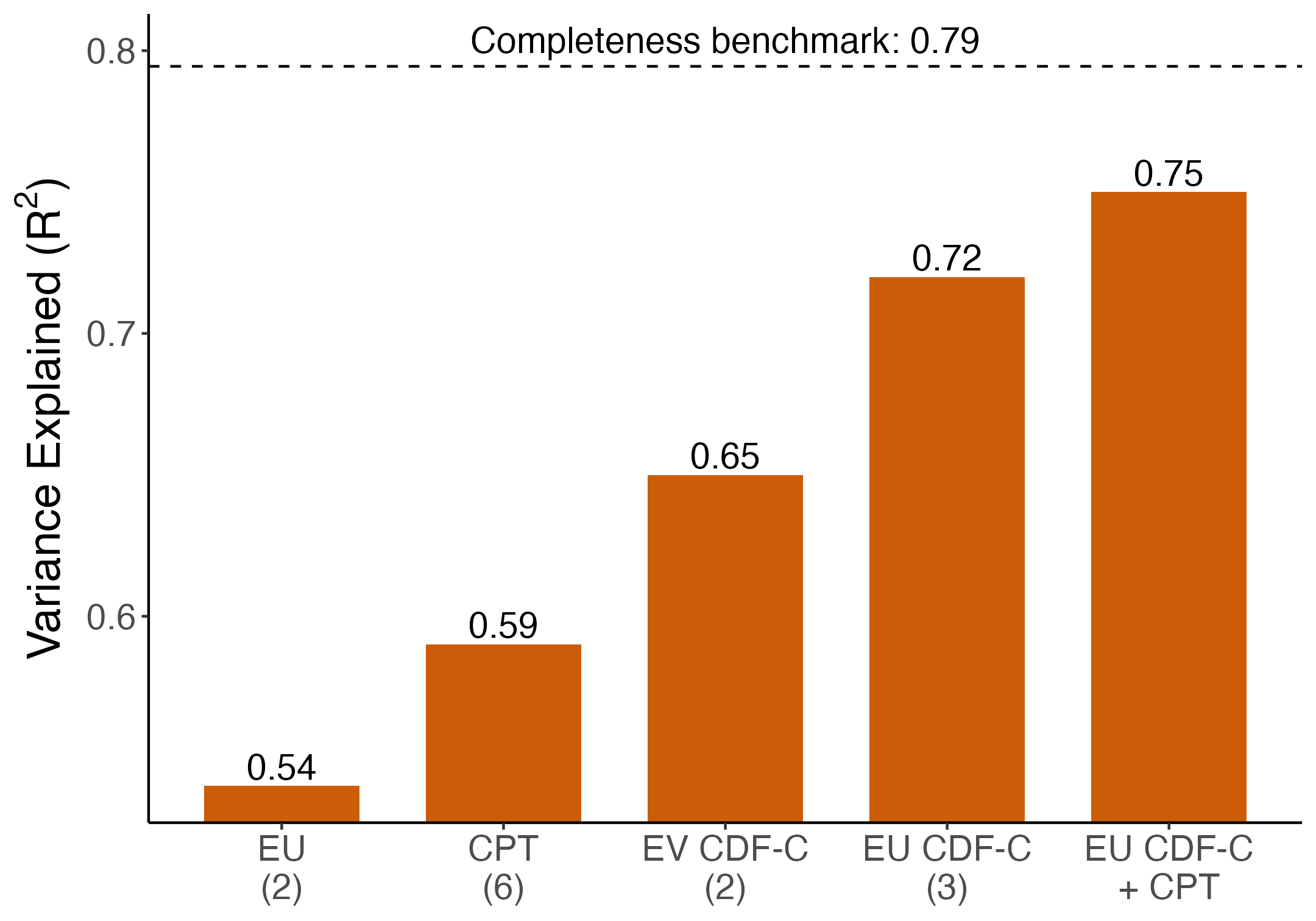}}
\end{subfigure}
\captionsetup{font=small}
\caption{Variance explained of models and completeness benchmarks. $R^2$ values are observation-weighted. Number of free parameters in parentheses. ``LGT,'' ``RT,'' ``L1-C'' refer to the Distortion-Free, Relative Thinking, and $L_1$-complexity models described in Appendix \ref{APP:structural_mac}. ``EDU,'' ``HDU,'' and ``CPF-C'' refer to the Exponential Discounting, Hyperbolic Discounting, and CPF-complexity models described in Appendix \ref{APP:structural_time}. ``EU,'' ``CPT,'' ``EV CDF-C,'' and ``EU CDF-C'' refer to the Expected Utility, Cumulative Prospect Theory, risk-neutral CDF-complexity, and expected utility CDF-complexity models described in Appendix \ref{APP:structural_risk}. Completeness benchmarks are obtained using an ensemble of models and a neural network (see \href{https://jeffreyyang97.github.io/personalwebsite/CC_OA.pdf}{Supplemental Appendix H}).}
\label{fig:model_comp}
\end{figure}

Figure \ref{fig:model_comp} reports the variance explained over problem-level choice rates of the benchmark model, the leading behavioral model, and our comparison complexity model. We emphasize that we do not view this exercise as a horserace per se, as these behavioral models likely capture real insights about preferences. Instead, this exercise highlights  the extent to which our model captures orthogonal variation in choice relative to existing models. We see this clearly in the final column, which reports the variance explained of an ensemble that combines the predictions of the leading behavioral model and our comparison complexity model.\footnote{Specifically, we report the variance explained of an observation-weighted regression of choice rates against the predicted choice rates of the two models.} Appendix Tables \ref{tab:structural_mac}, \ref{tab:structural_time}, and \ref{tab:structural_risk} report full estimation results.

In multiattribute choice, we estimate our $L_1$-complexity model using the two and three parameter specifications of $G$ developed in Section $\ref{SEC:theory_params}$. In terms of fit, both the two and three parameter models ($R^2$ of 0.32 and 0.36, respectively) are comparable to the leading behavioral model of relative thinking ($R^2=0.34$). Importantly, $L_1$-complexity explains a substantial amount of variation in choices not captured by the relative thinking model. The ensemble of three-parameter $L_1$-complexity and relative thinking has an $R^2$ of 0.42, a 24\% increase in variance explained compared to relative thinking alone.


In intertemporal and lottery choice, our model delivers significant performance gains. In intertemporal choice, our CPF-complexity model ($R^2=0.89$) explains 13\% more variation than hyperbolic discounting ($R^2=0.79$), using the same number of free parameters. In lottery choice, we estimate two versions of the CDF-complexity model---one that assumes risk-neutral preferences, and one that allows for utility curvature. Despite having far fewer free parameters, the risk-neutral CDF-complexity model ($R^2=0.65)$ explains 10\% more variation than cumulative prospect theory $(R^2=0.59)$. This is in line with \citet{enke_quantifying_2023}, which finds that allowing complexity to enter the noise term of a logit choice model substantially increases predictive power.\footnote{The ``complexity index'' developed by \citet{enke_quantifying_2023} loads heavily on ``excess dissimilarity,'' which equals the denominator of the CDF ratio minus the numerator, assuming $u(x)=x$.} Adding a parameter to capture utility curvature in our model yields an $R^2$ of 0.72 ---a 22\% improvement over cumulative prospect theory.

These results indicate that there is meaningful variation in choice probabilities which cannot be accounted for by preference models under standard specifications of noise, and that developing descriptive models of noise is important for making sense of behavior.\\

\noindent \textbf{\textit{Completeness and Restrictiveness.}} Following \citet{fudenberg_measuring_2022}, we establish ``completeness benchmarks'' in our three domains by training flexible models to predict choice rates based on problem features; these benchmarks capture the predictable variation of choice rates in our datasets that \textit{any} model could hope to explain. We form these benchmarks using an ensemble described in \href{https://jeffreyyang97.github.io/personalwebsite/CC_OA.pdf}{Supplemental Appendix H}, which combines parametric model predictions with those of a neural network. The dashed lines in Figure \ref{fig:model_comp} report the $R^2$ of these benchmarks, and Appendix Tables \ref{tab:structural_mac}, \ref{tab:structural_time}, and \ref{tab:structural_risk} report the completeness of each model following \citet{fudenberg_measuring_2022}. Our model captures 70\% of the predictable variation in our multiattribute choice data, and over 90\% in lottery and intertemporal choice. Note that this high level of completeness in lottery and intertemporal choice suggests that there is limited scope for systematic heterogeneity \textit{across problems} in the relative performance of our model: if our model fit substantially worse than alternative models on a class of problems, the neural network would deliver large performance gains. However, this exercise cannot speak to heterogeneity \textit{across subjects}; we characterize completeness only with respect to the best representative agent model.

One concern is that the predictive power of our model may come at the cost of parsimony---i.e., it is flexible enough to explain any dataset. To address this, we estimate the restrictiveness of each model: a measure of how well it fits synthetic data. Following  \citet{fudenberg_how_2023}, we simulate synthetic data according to a set of common restrictions imposed by the models we evaluate, and compare the fit of each model on this data (see Appendix \ref{APP:completeness} for details). This yields a restrictiveness measure $r\in[0,1]$, where $r=1$ means the model fits the synthetic data no better than a baseline model on average, and $r=0$ implies a perfect fit. Appendix Tables \ref{tab:structural_mac}, \ref{tab:structural_time}, and \ref{tab:structural_risk} report this measure in each domain. While our model is less restrictive than relative thinking in multiattribute choice ($r=0.45$ vs. 0.46), our model is \textit{more} restrictive than cumulative prospect theory ($r=0.60$ vs. 0.53) and hyperbolic discounting ($r=0.59$ vs. 0.58) in lottery and intertemporal choice, respectively. These results suggest that the performance gains of our model do not come at substantial costs to model parsimony. \\

\noindent\textbf{\textit{Alternative noise specifications.}} The above analysis demonstrates that our model of tradeoff-driven noise offers additional predictive power beyond existing models of \textit{value}, but does not answer how it compares to alternative specifications of \textit{noise}. As the high completeness of our model in lottery and intertemporal choice means that there is limited scope for alternative specifications to improve on fit in those data, we focus on multiattribute choice, where our model's completeness (70\%) leaves scope for improvement. Appendix \ref{APP:alternative_metrics} compares our multiattribute model to a flexible class of moderate utility noise specifications using the $L_p$ distance. We find that these alternative specifications provide essentially the same aggregate goodness of fit as our model. However, they predict dominance violations that are strongly rejected in the data.


\section{Relationship to Existing Models}
\label{SEC:existing_models}

\noindent\textbf{\textit{Relationship to Linear Differentiation.}} \citet{he_random_2023} propose a \textit{Linear Differentiation Model} (LDM) in the multiattribute domain with $X=\mathbb{R}^n$, where choice rates are given by $\rho(x,y)=G\left(\frac{U(x)-U(y)}{\sqrt{(x-y)'\Sigma (x-y)}}\right)$ for linear utility $U(x)=\sum_{k}\beta_kx_k$, an $n\times n$ positive definite matrix $\Sigma$, and a continuous, strictly increasing $G$. This model is identical to the $L_1$-complexity model save for the distance metric in the denominator: in the LDM, dissimilarity is measured by the generalized Euclidean distance, in contrast to the $L_1$ distance. The two models are disjoint (see Appendix \ref{APP:L_2}). Below, we outline two key dimensions along which they differ.

First, as Theorem \ref{THM:representation} implies, the $L_1$-complexity model respects \textit{dominance}: $\rho(x,y)$ is maximized when $x$ attribute-wise dominates $y$ (written $x>_{D} y$). The LDM violates this property. For example, consider a parameterization of the LDM with $n=3$, $\beta=(1,1,1)$, and $\Sigma=I$, i.e., the distance in the ratio is Euclidean: $d_{L2}(x,y)=\sqrt{\sum_{k}(x_k-y_k)^2}$. Consider the choice options $x=(4,0,0), x'=(-1,2,3),$ and $y=(0,0,0)$. In this case, the LDM predicts a dominance violation $\rho(x',y)>\rho(x,y)$: while both comparisons involve the same value difference, we have $d_{L2}(x,y)>d_{L2}(x',y)$. In Appendix \ref{APP:L_2}, we show that such dominance violations arise for any parameterization of the LDM when there are 3 or more attributes. 


Second, the $L_1$-complexity model also satisfies a \textit{monotonicity} property, wherein $x'>_{D}x$ implies $\rho(x',y)\geq \rho(x,y)$: that is, improving a choice option along each attribute cannot decrease its probability of being chosen.\footnote{See Proposition \ref{PROP:L2_monotonicity} in the Appendix. Each of $L_1$, CDF, and CPF representations satisfy monotonicity with respect to the domain-specific dominance notions; see Lemma \ref{LEM:monotonicity} in the Appendix.} The LDM, on the other hand, violates monotonicity. For example, take the LDM parameterization discussed above and consider $x=(5,5,0)$, $x'=(10,5,0)$, and $y=(0,0,0)$: here, the LDM predicts the monotonicity violation $\rho(x',y)<\rho(x,y)$. In Appendix \ref{APP:L_2}, we show that such monotonicity violations arise generically in the LDM. 

This highlights a key difference in how the analyst should interpret the attributes within each model. As the $L_1$-complexity model satisfies dominance and monotonicity, it describes settings where the ranking within each attribute is unambiguous to the DM, and where choice rates principally reflect the difficulty of making tradeoffs across attributes. On the other hand, since the LDM violates dominance and monotonicity, it is better suited to applications where choice rates also reflect difficulty in processing within-attribute differences, or alternatively, reflect disagreement in a population over the valence of attributes.\\
\\
\textbf{\textit{Relationship to Bayesian Probit.}}
\citet{natenzon_random_2019} develops a \textit{Bayesian Probit} model, which like ours, models a DM who responds to imprecise comparisons in a Bayesian fashion. Here, the DM has i.i.d. Gaussian priors over $v_x$, and chooses based on signals $s_x=v_x+\frac{1}{\sqrt{p}}\epsilon_x$ received for each option in the menu, where $\epsilon_x\sim N(0,1)$ are jointly normal across options. The pairwise correlations of $(\epsilon_x,\epsilon_y)$ allow the model to capture the ease of comparison between options, where $x,y$ are more comparable if $(\epsilon_x,\epsilon_y)$ are more highly correlated. 

Recall that in our model, the DM only receives \textit{ordinal} information on value. In Bayesian Probit, the DM learns about the \textit{cardinal} value differences, which rules out certain choice patterns. For instance, consider the simple lotteries $x=(\$10,0.8)$, $y=(\$5,1)$, $z=(\$0,1)$. Since $z$ is dominated by both $x$ and $y$, we might expect that $\rho(x,z)=\rho(y,z)=1$ and also $\rho(x,y)<1$; that is, the DM does not err in the presence of dominance but finds tradeoffs  difficult. Bayesian Probit cannot rationalize this choice data; $\rho(x,z)=\rho(y,z)=1$ implies $\corr(\epsilon_x,\epsilon_z) = \corr(\epsilon_y,\epsilon_z) = 1$, which implies  $\corr(\epsilon_x,\epsilon_y) = 1$; as a result, we have $\rho(x,y)=1$.\footnote{ More generally, when $\rho(x,z),\rho(y,z)$ are close to 1, the Bayesian Probit model places lower bounds on the choice probability $\rho(x,y)$. These bounds imply that there are binary choice rules rationalizable by $L_1$, CDF, and CPF complexity but not by Bayesian Probit; see \href{https://jeffreyyang97.github.io/personalwebsite/CC_OA.pdf}{Supplemental Appendix F} for details.}

Intuitively, since Bayesian Probit features learning over cardinal value differences, the DM cannot learn $v_x-v_z$ and $v_y-v_z$ without also learning $v_x-v_y$. This feature of Bayesian Probit makes it less suitable for modeling the behavioral regularities central to this paper. Our application to valuations, for instance, considers settings where the DM finds a lottery trivial to compare to some prices and difficult to compare to others, and yet perfectly understands the ranking of prices. While Bayesian Probit cannot accommodate these patterns, our model can. 

\section{Conclusion}
\label{SEC:conclusion}

This paper shows that the difficulty of making tradeoffs can serve as an organizing mechanism that rationalizes a range of empirical regularities across multiple domains of behavioral economics---including regularities central to recent work on complexity---and makes novel predictions. Using experimental evidence we demonstrate the predictive power of our proposed complexity measures, and experimentally test the novel implications of our framework. We close by discussing limitations of the paper and potential avenues for future work.\\

\noindent\textbf{\textit{Extending to additional domains.}} Although we formulated our theory in the domains of multiattribute, lottery, and intertemporal choice, the difficulty of tradeoffs is likely a general feature of decision-making. One extension is to adapt our $L_1$ complexity measure to choice under uncertainty by reinterpreting states of the world as attributes. Another would be to apply the model to derive solution concepts in normal-form games that account for tradeoff-induced noise.\\

\noindent\textbf{\textit{Incorporating non-standard preferences.}} To develop a parsimonious model of tradeoff complexity, we abstracted from non-standard preferences in the DM's objective function $v_x$. However, decision-makers with non-standard preferences also presumably face tradeoffs, and so there is value in integrating models of these preferences---which augment $v_x$---with our model of tradeoff-driven complexity, and understanding their behavioral implications. Some models can be easily integrated with our approach: for instance, many non-standard multiattribute choice models reflect distorted attribute weights, which can directly be incorporated into the $L_1$ complexity measure. Others may require more involved adaptations of our measures. \\


\noindent\textbf{\textit{Developing implications for measurement.}} This paper develops a theory that can help explain apparent inconsistencies across workhorse preference elicitation methods, such as the gap between certainty equivalents and binary choice found in lottery preference reversals. However, we do not view our results as an indication that any one elicitation method should be \textit{a priori} favored over others---the suitability of a given measure ultimately depends on the goals of the researcher and how data is interpreted in service of those goals.\footnote{For instance, even if valuations are subject to pull-to-center distortions as in our model, \textit{comparisons} of valuations could still be informative about  preferences, so long as these distortions impact valuations in the same way. Likewise, inference of preference intensities on the basis of binary choice rates may be confounded by the presence of heteroskedastic noise; see \citet{mcgranaghan_distinguishing_2024} for a discussion.} Instead, we view a theory like ours, which models how noise may manifest differently across elicitation methods, as a tool to help facilitate the interpretation of existing elicitation methods and motivate the development of new measures. We view both aims as promising directions for future work.



\newpage

\vspace{1cm}

\bibliographystyle{aer}
\bibliography{Paper/References.bib}

\newpage
\appendix

\begin{center}
\Large{\textbf{APPENDIX}}
\end{center}

\section{Appendix: Characterization of $L_1$ Complexity}
\label{APP:l1_axioms}
The binary choice probabilities induced by $\tau^{L1}$ take the following form: 

\begin{definition}
\label{def:L1_complexity_rho}
    A binary choice rule $\rho$ has an $L_1$-complexity representation if there exists $\beta\in\mathbb{R}^n$ with $\beta_k\neq 0$ for all $k$, such that
    \begin{align*}
     \rho(x,y)=G\left(\frac{U(x)-U(y)}{d_{L1}(x,y)}\right)   
    \end{align*}
    for some continuous, strictly increasing $G$ satisfying $G(r)=1-G(-r)$.
\end{definition}

We provide an axiomatic characterization of this representation. Let $x_{\set{k}}y$ denote the option obtained by replacing the $k$th attribute of $y$ with $x_k$, i.e. $(x_{\set{k}}y)_k=x_k$ and $(x_{\set{k}}y)_j=y_j$ for all $j\neq k$. Say that $x$ \textit{dominates} $y$, written $x>_{D}y$, if $\rho(x_{\set{k}}y,y)\geq 1/2$ for all $k$ with a strict inequality for at least one $k$---that is, if $x$ is revealed better along each attribute considered in isolation. Say that attribute $k$ is \textit{null} if $\rho(x_{\set{k}}z,y_{\set{k}}z)=1/2$ for all $x,y,z\in X$. Consider the following conditions: 
\begin{enumerate}[label={M\arabic*}., itemsep=0.8mm]
    \item \textbf{Continuity:} $\rho(x,y)$ is continuous on its domain.
    \item \textbf{Linearity:} $\rho(x,y)=\rho(\alpha x+(1-\alpha) z,\alpha y+(1-\alpha) z)$.
    \item \textbf{Moderate Transitivity:} If $\rho(x,y)\geq1/2$ and $\rho(y,z)\geq1/2$, then either $\rho(x,z)> \min\set{\rho(x,y),\rho(y,z)}$ or $\rho(x,z)=\rho(x,y)=\rho(y,z)$.
    \item \textbf{Dominance}: If $x>_{D}y$, then $\rho(x,y)\geq \rho(w,z)$ for any $w,z\in X$, where the inequality is strict if $w\not>_{D} z$.
    \item \textbf{Simplification}: If $\rho(x,y)\geq 1/2$: for any $x'\in X$ satisfying
    \begin{enumerate}[label={(\arabic*}), leftmargin=1.5cm, itemsep=0.8mm]
        \item $x'_i=y_i$ for some $i$,
        \item $x'_j\neq x_j$ for at most one $j\neq i$,
    \end{enumerate}
    such that $\rho(x',x)=1/2$, we have $\rho(x',y)\geq \rho(x,y)$.
\end{enumerate}

Continuity is a technical axiom.\footnote{Continuity holds on the domain where $x\neq y$. This allows the model to accommodate dominance-related discontinuities, e.g. a jump between the choice probabilities $\rho(\$5,\$4.99)$ and $\rho(\$5,\$5.01)$.}  Linearity reflects the fact that preferences and the $L_1$ distance are linear in attributes; in Appendix \ref{APP:theory_axioms}, we characterize a generalized representation that allows for non-linear preferences. Moderate Transitivity is due to \citet{he_moderate_2024}, who show that the axiom characterizes the moderate utility class. Dominance and Simplification are exact counterparts of the properties of $\tau^{L1}$ discussed in Section \ref{SEC:theory_mac}. Dominance says that if $x$ dominates $y$, the likelihood of accurate choice is maximal. Simplification says that aggregating value differences across two attributes into one increases the likelihood of correct choice.

Theorem \ref{THM:representation} states that when there are three or more attributes, M1--M5 characterize the behavioral implications of our representation for binary choice data, and that its parameters $(G,\beta)$ can be identified from choice data.

\begin{thm}
\label{THM:representation}
Suppose all attributes are non-null and $n> 2$. $\rho(x,y)$ has an $L_1$-complexity representation $(G,\beta)$ if and only if it satisfies M1--M5. Moreover, if $\rho(x,y)$ also has an $L_1$-complexity representation $(G',\beta')$, then $G'=G$ and $\beta'=C\beta$ for some $C>0$. 
\end{thm}

In \href{https://jeffreyyang97.github.io/personalwebsite/CC_OA.pdf}{Supplemental Appendix F}, we extend Theorem \ref{THM:representation} to the two-attribute case. Here, we also provide guidance on how the analyst should specify the attributes in a given setting, and show how attributes can be identified from choices.

\section{Appendix: Tables and Figures}
\label{APP:figs}

\begin{table}[!h]
\caption{Complexity Responses vs. $L_1$ Ratio}
\begin{center}
\begin{tabular*}{1\linewidth}{@{\extracolsep{\fill}}l c c c c c c c c c c c@{}}
\hline
 & \multicolumn{2}{c}{\makecell{\textit{Dep. Variable}:\\ Error Rate}} & & \multicolumn{2}{c}{\makecell{\textit{Dep. Variable}:\\ Inconsistency Rate}} & & \multicolumn{2}{c}{\makecell{\textit{Dep. Variable}:\\ CU}} & & \multicolumn{2}{c}{\makecell{\textit{Dep. Variable}:\\ Response Time}} \\
\cline{2-3} \cline{5-6} \cline{8-9} \cline{11-12}
 & (1) & (2) &   & (3) & (4) &   & (5) & (6) &   & (7) & (8) \\
\hline
$L_1$ Ratio & $-0.26^{***}$ & $-0.26^{***}$ &  & $-0.19^{***}$ & $-0.19^{***}$ &  & $-0.12^{***}$ & $-0.12^{***}$ &  & $-8.79^{***}$ & $-8.79^{***}$ \\
            & $(0.02)$      & $(0.02)$      &  & $(0.02)$      & $(0.02)$      &  & $(0.01)$      & $(0.01)$      &  & $(0.56)$      & $(0.57)$      \\
$\Delta V$  &               & $0.00$        &  &               & $-0.01$       &  &               & $-0.00$       &  &               & $-0.00$       \\
            &               & $(0.00)$      &  &               & $(0.01)$      &  &               & $(0.00)$      &  &               & $(0.14)$      \\
(Intercept) & $0.32^{***}$  & $0.32^{***}$  &  & $0.26^{***}$  & $0.28^{***}$  &  & $0.25^{***}$  & $0.27^{***}$  &  & $20.45^{***}$ & $20.46^{***}$ \\
            & $(0.01)$      & $(0.02)$      &  & $(0.01)$      & $(0.03)$      &  & $(0.00)$      & $(0.01)$      &  & $(0.43)$      & $(0.69)$      \\
\hline
R$^2$       & $0.32$        & $0.32$        &  & $0.03$        & $0.03$        &  & $0.30$        & $0.31$        &  & $0.31$        & $0.31$        \\
Num. obs.   & $662$         & $662$         &  & $4880$        & $4880$        &  & $662$         & $662$         &  & $662$         & $662$         \\
\hline
\multicolumn{12}{l}{\scriptsize{\makebox[0pt][l]{\hspace{-\tabcolsep}\parbox{1.05\linewidth}{\vspace{4pt} OLS Estimates. Standard errors (in parentheses) are robust.
                                            ``$L_1$ Ratio" and ``$\Delta V$" are the $L_1$ ratio and the monetary
                                            value difference for each choice problem. \\ $^{***}p<0.001$; $^{**}p<0.01$; $^{*}p<0.05$.}}}}
\end{tabular*}
\label{tab:l1regs}
\end{center}
\end{table}

\begin{table}[!ht]
\centering
\caption{Structural Estimates: Multiattribute Choice} 
\label{tab:structural_mac}
\begin{tabular*}{0.95\linewidth}{@{\extracolsep{\fill}}lrcccccc@{}}
  \hline
\hline
 &   & DF & BGS & Focus & RT & $L_1$-C 2P & $L_1$-C, 3P \\ 
  \hline
Parameter Estimates &  &   &   &   &   &   &   \\ 
  \quad $\delta$ &  &   & 1 &   &   &   &   \\ 
  \quad $\theta$ &  &   &   & 0 &   &   &   \\ 
  \quad $\omega$ &  &   &   &   & 0.84 &   &   \\ 
  \quad $\xi$ &  &   &   &   & 1.62 &   &   \\ 
  \quad $\kappa$ &  &   &   &   &   & 0.09 & 0.04 \\ 
  \quad $\gamma$ &  &   &   &   &   & 2.36 & 0.86 \\ 
  \quad $\psi$ &  &   &   &   &   &   & 0.55 \\ 
  \quad $\eta$ &  & 0.37 & 0.37 & 0.37 & 1.54 &   &   \\ 
   \hline
$R^2$ &  & 0.001 & 0.001 & 0.001 & 0.339 & 0.323 & 0.357 \\ 
  Completeness &  & 0 & 0 & 0 & 0.56 & 0.62 & 0.7 \\ 
  Restrictiveness &  & 0.482 & 0.482 & 0.483 & 0.461 & 0.454 & 0.452 \\ 
  \quad &  & (0.001) & (0.001) & (0.001) & (0.001) & (0.001) & (0.001) \\ 
   \hline 
 \multicolumn{8}{l}{\makebox[0pt][l]{\hspace{-\tabcolsep}\parbox{0.95\linewidth}{\scriptsize{``DF'', ``BGS'', ``Focus'', and ``RT''
        refer to the Distortion-Free, Salience, Focusing, and Relative Thinking models described in
        Appendix \ref{APP:structural_mac}. ``$L_1$-C, 2P'' and ``$L_1$-C, 3P'' refer to the
                        2 and 3 parameter $L_1$-Complexity models described in Appendix \ref{APP:structural_mac}.
                        Completeness and Restrictiveness measures are defined in Appendix \ref{APP:completeness}.
                        Standard errors for Restrictiveness estimates in parentheses.}}}} 
\end{tabular*}
\end{table}

\begin{table}[!h]
\caption{Complexity Responses vs. CDF Ratio}
\begin{center}
\begin{tabular*}{1.05\linewidth}{@{\extracolsep{\fill}}l c c c c c c c c c c c@{}}
\hline
 & \multicolumn{2}{c}{\makecell{\textit{Dep. Variable}:\\ Error Rate}} & & \multicolumn{2}{c}{\makecell{\textit{Dep. Variable}:\\ Inconsistency Rate}} & & \multicolumn{2}{c}{\makecell{\textit{Dep. Variable}:\\ CU}} & & \multicolumn{2}{c}{\makecell{\textit{Dep. Variable}:\\Response Time}} \\
\cline{2-3} \cline{5-6} \cline{8-9} \cline{11-12}
 & (1) & (2) &   & (3) & (4) &   & (5) & (6) &   & (7) & (8) \\
\hline
CDF Ratio   & $-0.32^{***}$ & $-0.28^{***}$ &  & $-0.10^{***}$ & $-0.11^{***}$ &  & $-0.19^{***}$ & $-0.21^{***}$ &  & $-3.24^{***}$ & $-3.24^{***}$ \\
            & $(0.00)$      & $(0.00)$      &  & $(0.00)$      & $(0.00)$      &  & $(0.01)$      & $(0.01)$      &  & $(0.36)$      & $(0.38)$      \\
$\Delta V$  &               & $-0.01^{***}$ &  &               & $0.00^{***}$  &  &               & $0.01^{***}$  &  &               & $-0.00$       \\
            &               & $(0.00)$      &  &               & $(0.00)$      &  &               & $(0.00)$      &  &               & $(0.00)$      \\
(Intercept) & $0.48^{***}$  & $0.51^{***}$  &  & $0.28^{***}$  & $0.28^{***}$  &  & $0.25^{***}$  & $0.22^{***}$  &  & $12.50^{***}$ & $12.50^{***}$ \\
            & $(0.00)$      & $(0.00)$      &  & $(0.00)$      & $(0.00)$      &  & $(0.01)$      & $(0.01)$      &  & $(0.22)$      & $(0.24)$      \\
\hline
R$^2$       & $0.45$        & $0.47$        &  & $0.10$        & $0.11$        &  & $0.31$        & $0.35$        &  & $0.13$        & $0.13$        \\
Num. obs.   & $10920$       & $10920$       &  & $10420$       & $10420$       &  & $500$         & $500$         &  & $500$         & $500$         \\
\hline
\multicolumn{12}{l}{\scriptsize{\makebox[0pt][l]{\hspace{-\tabcolsep}\parbox{1.05\linewidth}{\vspace{4pt} OLS estimates. Standard errors (in parentheses) are robust.
                                            ``CDF" Ratio" and ``$\Delta V$" are the representative-agent
                                            CDF ratio and value difference for each choice problem, computed using the value of
                                            $\alpha$ estimated in the Expected Utility model described in Appendix
                                            \ref{APP:structural_risk}. \\ $^{***}p<0.001$; $^{**}p<0.01$; $^{*}p<0.05$.}}}}
\end{tabular*}
\label{tab:cdfregs_global}
\end{center}
\end{table}

\begin{table}[!h]
\caption{Individual-Level Error Rates vs. CDF Ratio}
\begin{center}
\begin{tabular*}{0.72\linewidth}{@{\extracolsep{\fill}}l c c c c@{}}
\hline
 & \multicolumn{4}{c}{\makecell{\textit{Dep. Variable}:\\ Binary Error (Indiv. $\hat\alpha$)}} \\
\cline{2-5}
 & (1) & (2) & (3) & (4) \\
\hline
Global CDF Ratio  & $-0.27^{***}$ & $-0.27^{***}$ &               &               \\
                  & $(0.01)$      & $(0.01)$      &               &               \\
Indiv. CDF Ratio  &               &               & $-0.35^{***}$ & $-0.35^{***}$ \\
                  &               &               & $(0.01)$      & $(0.01)$      \\
Indiv. $\Delta V$ &               & $-0.00$       &               & $0.00$        \\
                  &               & $(0.00)$      &               & $(0.00)$      \\
(Intercept)       & $0.36^{***}$  & $0.36^{***}$  & $0.41^{***}$  & $0.41^{***}$  \\
                  & $(0.01)$      & $(0.01)$      & $(0.01)$      & $(0.01)$      \\
\hline
R$^2$             & $0.04$        & $0.04$        & $0.07$        & $0.07$        \\
Num. obs.         & $12500$       & $12500$       & $12500$       & $12500$       \\
\hline
\multicolumn{5}{l}{\scriptsize{\makebox[0pt][l]{\hspace{-\tabcolsep}\parbox{0.72\linewidth}{\vspace{4pt} OLS estimates. Standard errors (in parentheses) are robust.
                                            ``Global CDF" Ratio" is the representative-agent CDF ratio for each subject-choice problem,
                                            computed using the value of $\alpha$ estimated in the Expected Utility model described in Appendix
                                            \ref{APP:structural_risk}. ``Indiv. CDF" Ratio" and ``Indiv. $\Delta V$" are the individual-level
                                            CPF ratio and value difference for each subject-choice problem, computed using the individual-level
                                            $\alpha_i$ estimates under the same model. \\ $^{***}p<0.001$; $^{**}p<0.01$; $^{*}p<0.05$.}}}}
\end{tabular*}
\label{tab:cdf_error_indiv}
\end{center}
\end{table}

\begin{table}[h!]
\centering
\caption{Structural Estimates: Lottery Choice} 
\label{tab:structural_risk}
\begin{tabular*}{1\linewidth}{@{\extracolsep{\fill}}lrcccccc@{}}
  \hline
\hline
 &   & EU & ST & RDEU & CPT & EV CDF-C  & EU CDF-C \\ 
  \hline
Parameter Estimates &  &   &   &   &   &   &   \\ 
  \quad $\alpha$ &  & 0.85 & 0.83 & 0.83 & 0.75 &   & 0.6 \\ 
  \quad $\beta$ &  &   &   & 0.77 & 0.78 &   &   \\ 
  \quad $\lambda$ &  &   &   & 1.07 & 0.79 &   &   \\ 
  \quad $\chi$ &  &   &   &   & 1.06 &   &   \\ 
  \quad $\nu$ &  &   &   &   & 0.83 &   &   \\ 
  \quad $\phi$ &  &   & -2.6 &   &   &   &   \\ 
  \quad $\upsilon$ &  &   & 0.6 &   &   &   &   \\ 
  \quad $\mu$ &  &   & 0.81 &   &   &   &   \\ 
  \quad $\kappa$ &  &   &   &   &   & 0.15 & 0.15 \\ 
  \quad $\gamma$ &  &   &   &   &   & 0.77 & 0.71 \\ 
  \quad $\eta$ &  & 0.22 & 0.24 & 0.24 & 0.33 &   &   \\ 
   \hline
$R^2$ &  & 0.54 & 0.56 & 0.55 & 0.59 & 0.65 & 0.72 \\ 
  Completeness &  & 0.64 & 0.67 & 0.66 & 0.71 & 0.81 & 0.9 \\ 
  Restrictiveness &  & 0.564 & 0.561 & 0.564 & 0.532 & 0.6 & 0.596 \\ 
  \quad &  & (0.000) & (0.000) & (0.000) & (0.000) & (0.000) & (0.000) \\ 
   \hline 
 \multicolumn{8}{l}{\makebox[0pt][l]{\hspace{-\tabcolsep}\parbox{1\linewidth}{\scriptsize{``EU'', ``ST'', ``RDEU'', and ``CPT''
        refer to the Expected Utility, Simplicity Theory, Reference-Dependent Expected Utility, and Cumulative Prospect Theory
        models described in Appendix \ref{APP:structural_risk}. ``EV CDF-C'' and ``EU CDF-C'' refer to the risk-neutral
                        and expected utility CDF complexity models described in Appendix \ref{APP:structural_risk}.
                        Completeness and Restrictiveness measures are defined in Appendix \ref{APP:completeness}.
                        Standard errors for Restrictiveness estimates in parentheses.}}}} 
\end{tabular*}
\end{table}

\begin{table}[!h]
\caption{Error Rate vs. CDF Ratio in Expected Value Task}
\begin{center}
\begin{tabular*}{0.72\linewidth}{@{\extracolsep{\fill}}l c c@{}}
\hline
 & \multicolumn{2}{c}{\makecell{\textit{Dependent Variable}:\\ Error Rate}} \\
\cline{2-3}
 & (1) & (2) \\
\hline
CDF Ratio             & $-0.31^{***}$ & $-0.30^{***}$ \\
                      & $(0.01)$      & $(0.01)$      \\
Abs. Value Difference &               & $-0.00^{**}$  \\
                      &               & $(0.00)$      \\
(Intercept)           & $0.40^{***}$  & $0.41^{***}$  \\
                      & $(0.01)$      & $(0.01)$      \\
\hline
R$^2$                 & $0.29$        & $0.29$        \\
Num. obs.             & $2118$        & $2118$        \\
\hline
\multicolumn{3}{l}{\scriptsize{\makebox[0pt][l]{\hspace{-\tabcolsep}\parbox{0.72\linewidth}{\vspace{4pt} OLS estimates using expected value task data from \citet{enke_quantifying_2023}.
  Subjects were instructed to choose the lottery which pays out more on average, which means errors are objectively defined.
  Standard errors (in parentheses) are robust.
  CDF Ratio computed using expected value (risk-neutral preferences). $^{***}p<0.001$; $^{**}p<0.01$; $^{*}p<0.05$.}}}}
\end{tabular*}
\label{tab:prediction_error}
\end{center}
\end{table}

\begin{table}[!h]
\caption{Complexity Responses vs. CPF Ratio}
\begin{center}
\begin{tabular*}{1.05\linewidth}{@{\extracolsep{\fill}}l c c c c c c c c c c c@{}}
\hline
 & \multicolumn{2}{c}{\makecell{\textit{Dep. Variable}:\\ Error Rate}} & & \multicolumn{2}{c}{\makecell{\textit{Dep. Variable}:\\ Inconsistency Rate}} & & \multicolumn{2}{c}{\makecell{\textit{Dep. Variable}:\\ CU}} & & \multicolumn{2}{c}{\makecell{\textit{Dep. Variable}:\\ Response Time}} \\
\cline{2-3} \cline{5-6} \cline{8-9} \cline{11-12}
 & (1) & (2) &   & (3) & (4) &   & (5) & (6) &   & (7) & (8) \\
\hline
CPF Ratio   & $-0.51^{***}$ & $-0.51^{***}$ &  & $-0.20^{***}$ & $-0.17^{***}$ &  & $-0.07^{***}$ & $-0.07^{***}$ &  & $-1.54^{***}$ & $-0.80^{***}$ \\
            & $(0.01)$      & $(0.01)$      &  & $(0.01)$      & $(0.02)$      &  & $(0.00)$      & $(0.00)$      &  & $(0.20)$      & $(0.21)$      \\
$\Delta V$  &               & $-0.00$       &  &               & $-0.01^{***}$ &  &               & $-0.00$       &  &               & $-0.19^{***}$ \\
            &               & $(0.00)$      &  &               & $(0.00)$      &  &               & $(0.00)$      &  &               & $(0.02)$      \\
(Intercept) & $0.55^{***}$  & $0.55^{***}$  &  & $0.27^{***}$  & $0.28^{***}$  &  & $0.15^{***}$  & $0.15^{***}$  &  & $10.65^{***}$ & $10.92^{***}$ \\
            & $(0.01)$      & $(0.01)$      &  & $(0.01)$      & $(0.01)$      &  & $(0.00)$      & $(0.00)$      &  & $(0.15)$      & $(0.16)$      \\
\hline
R$^2$       & $0.59$        & $0.59$        &  & $0.02$        & $0.03$        &  & $0.22$        & $0.22$        &  & $0.05$        & $0.10$        \\
Num. obs.   & $1097$        & $1097$        &  & $8290$        & $8290$        &  & $1097$        & $1097$        &  & $1097$        & $1097$        \\
\hline
\multicolumn{12}{l}{\scriptsize{\makebox[0pt][l]{\hspace{-\tabcolsep}\parbox{1.05\linewidth}{\vspace{4pt} OLS estimates. Standard errors (in parentheses) are robust.
                                            ``CPF" Ratio" and ``$\Delta V$" are the representative-agent
                                            CPF ratio and value difference for each choice problem, computed using the value of 
                                            $\delta$ estimated in the Exponential Discounting model described in Appendix 
                                            \ref{APP:structural_time}. \\ $^{***}p<0.001$; $^{**}p<0.01$; $^{*}p<0.05$.}}}}
\end{tabular*}
\label{tab:cpfregs_global}
\end{center}
\end{table}

\begin{table}[!h]
\caption{Individual-Level Error Rates vs. CPF Ratio}
\begin{center}
\begin{tabular*}{0.72\linewidth}{@{\extracolsep{\fill}}l c c c c@{}}
\hline
 & \multicolumn{4}{c}{\makecell{\textit{Dep. Variable}:\\ Binary Error (Indiv. $\hat\delta$)}} \\
\cline{2-5}
 & (1) & (2) & (3) & (4) \\
\hline
Global CPF Ratio  & $-0.18^{***}$ & $-0.17^{***}$ &               &               \\
                  & $(0.01)$      & $(0.01)$      &               &               \\
Indiv. CPF Ratio  &               &               & $-0.37^{***}$ & $-0.32^{***}$ \\
                  &               &               & $(0.01)$      & $(0.01)$      \\
Indiv. $\Delta V$ &               & $-0.02^{***}$ &               & $-0.01^{***}$ \\
                  &               & $(0.00)$      &               & $(0.00)$      \\
(Intercept)       & $0.26^{***}$  & $0.32^{***}$  & $0.40^{***}$  & $0.40^{***}$  \\
                  & $(0.00)$      & $(0.01)$      & $(0.01)$      & $(0.01)$      \\
\hline
R$^2$             & $0.02$        & $0.06$        & $0.10$        & $0.10$        \\
Num. obs.         & $41450$       & $41450$       & $41450$       & $41450$       \\
\hline
\multicolumn{5}{l}{\scriptsize{\makebox[0pt][l]{\hspace{-\tabcolsep}\parbox{0.72\linewidth}{\vspace{4pt} OLS estimates. Standard errors (in parentheses) are robust.
                                            ``Global CPF" Ratio" is the representative-agent CPF ratio for each subject-choice problem,
                                            computed using the value of $\delta$ estimated in the Exponential Discounting model described in Appendix 
                                            \ref{APP:structural_time}. ``Indiv. CPF" Ratio" and ``Indiv. $\Delta V$" are the individual-level
                                            CPF ratio and value difference for each subject-choice problem, computed using the individual-level 
                                            $\delta_i$ estimates under the same model. \\ $^{***}p<0.001$; $^{**}p<0.01$; $^{*}p<0.05$.}}}}
\end{tabular*}
\label{tab:cpf_error_indiv}
\end{center}
\end{table}

\begin{table}[h!]
\centering
\caption{Structural Estimates: Intertemporal Choice} 
\label{tab:structural_time}
\begin{tabular*}{0.7\linewidth}{@{\extracolsep{\fill}}lrcccc@{}}
  \hline
\hline
 &   & EDU & QDU & HDU & CPF-C \\ 
  \hline
Parameter Estimates &  &   &   &   &   \\ 
  \quad $\delta$ &  & 0.95 & 0.96 &   & 0.96 \\ 
  \quad $\beta$ &  &   & 0.84 &   &   \\ 
  \quad $\iota$ &  &   &   & 0.16 &   \\ 
  \quad $\zeta$ &  &   &   & 0.12 &   \\ 
  \quad $\kappa$ &  &   &   &   & 0.03 \\ 
  \quad $\gamma$ &  &   &   &   & 0.85 \\ 
  \quad $\eta$ &  & 0.35 & 0.41 & 0.43 &   \\ 
   \hline
$R^2$ &  & 0.76 & 0.76 & 0.79 & 0.89 \\ 
  Completeness &  & 0.76 & 0.77 & 0.82 & 0.99 \\ 
  Restrictiveness &  & 0.591 & 0.584 & 0.585 & 0.593 \\ 
  \quad &  & (0.001) & (0.001) & (0.001) & (0.001) \\ 
   \hline 
 \multicolumn{6}{l}{\makebox[0pt][l]{\hspace{-\tabcolsep}\parbox{0.7\linewidth}{\scriptsize{``EDU", ``QDU", ``HDU", and ``CPF-C"
        refer to the Exponential Discounting, Quasi-Hyperbolic Discounting, Hyperbolic Discounting, and CPF Complexity models described in 
                        Appendix \ref{APP:structural_time}. For these estimates, each time period is 24 days.
                        Completeness and Restrictiveness measures are defined in Appendix \ref{APP:completeness}. 
                        Standard errors for Restrictiveness estimates in parentheses.}}}} 
\end{tabular*}
\end{table}

\begin{table}[h!]
\centering
\caption{Median Valuation Comparisons: Lottery Choice} 
\label{tab:median_differences_risk}
\begin{tabular}{cccc}
  \hline
\hline
Prob., High-risk & Prob., Low-risk & $\Delta$Median, CEs & $\Delta$Median, PEs \\ 
  \hline
0.19 & 0.90 & $-1.250^{*}$ & $0.290^{***}$ \\ 
  0.19 & 0.94 & $-1.500^{*}$ & $0.280^{***}$ \\ 
  0.19 & 0.98 & $-1.250^{*}$ & $0.320^{***}$ \\ 
  0.21 & 0.90 & $-0.875^{*}$ & $0.280^{***}$ \\ 
  0.21 & 0.94 & $-1.125^{**}$ & $0.270^{***}$ \\ 
  0.21 & 0.98 & $-0.875^{*}$ & $0.310^{***}$ \\ 
  0.23 & 0.90 & $-0.875$ & $0.280^{***}$ \\ 
  0.23 & 0.94 & $-1.125$ & $0.270^{***}$ \\ 
  0.23 & 0.98 & $-0.875$ & $0.310^{***}$ \\ 
   \hline
\multicolumn{4}{p{0.75\linewidth}}{\scriptsize{Notes: $\Delta$Median is the median rescaled valuation of the low-risk option minus that of the high-risk option. Stars indicate significance of median differences based on a two-sided bootstrap test using 10,000 bootstrap samples: $^{***}p<0.001$; $^{**}p<0.01$; $^{*}p<0.05$.}}
\end{tabular}
\end{table}

\begin{table}[h!]
\centering
\caption{Median Valuation Comparisons: Intertemporal Choice} 
\label{tab:median_differences_time}
\begin{tabular}{cccc}
  \hline
\hline
Delay, High-Delay & Delay, Low-Delay & $\Delta$Median, PVEs & $\Delta$Median, TEs \\ 
  \hline
750 & 150 & $-1.875^{**}$ & $-528.5^{***}$ \\ 
  750 & 90 & $-1.375^{*}$ & $-558.5^{***}$ \\ 
  750 & 30 & $-1.625^{*}$ & $-573.5^{***}$ \\ 
  690 & 150 & $-1.625^{*}$ & $-476.0^{***}$ \\ 
  690 & 90 & $-1.125$ & $-506.0^{***}$ \\ 
  690 & 30 & $-1.375^{*}$ & $-521.0^{***}$ \\ 
  630 & 150 & $-0.125$ & $-442.5^{***}$ \\ 
  630 & 90 & $0.375$ & $-472.5^{***}$ \\ 
  630 & 30 & $0.125$ & $-487.5^{***}$ \\ 
   \hline
\multicolumn{4}{p{0.80\linewidth}}{\scriptsize{Notes: $\Delta$Median is the median rescaled valuation of the low-delay option minus that of the high-delay option. Stars indicate significance of median differences based on a two-sided bootstrap test using 10,000 bootstrap samples: $^{***}p<0.001$; $^{**}p<0.01$; $^{*}p<0.05$.}}
\end{tabular}
\end{table}

\begin{figure}[b!]
\small
\centering
\begin{subfigure}[t]{0.45\textwidth}
\centering
    \caption{\centering Certainty Equivalents}
    \vspace{0.5em}
    {\includegraphics[width=65mm]{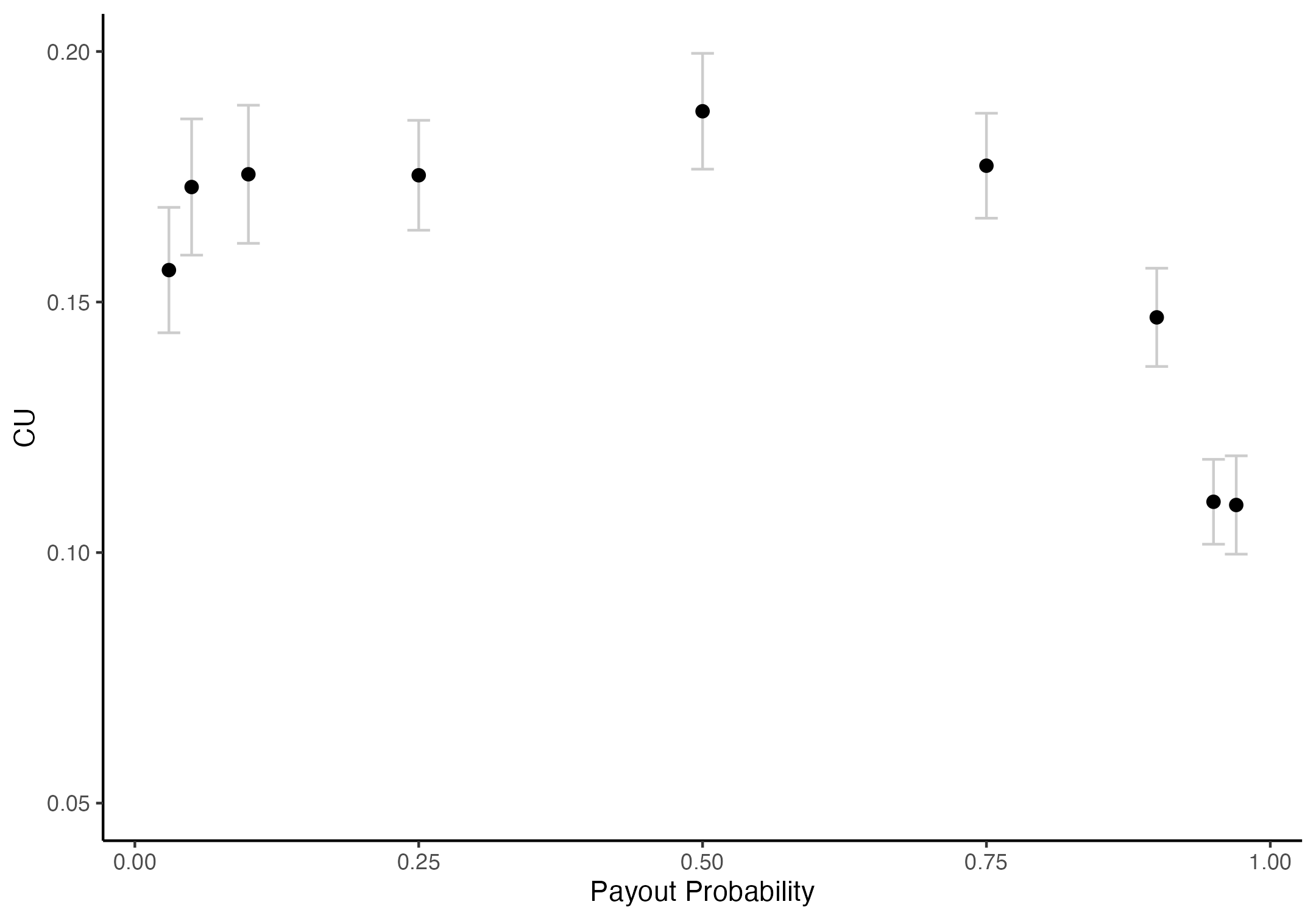}}
\end{subfigure}
\begin{subfigure}[t]{0.45\textwidth}
\centering
\caption{\centering Probability Equivalents }
    \vspace{0.5em}
    {\includegraphics[width=65mm]{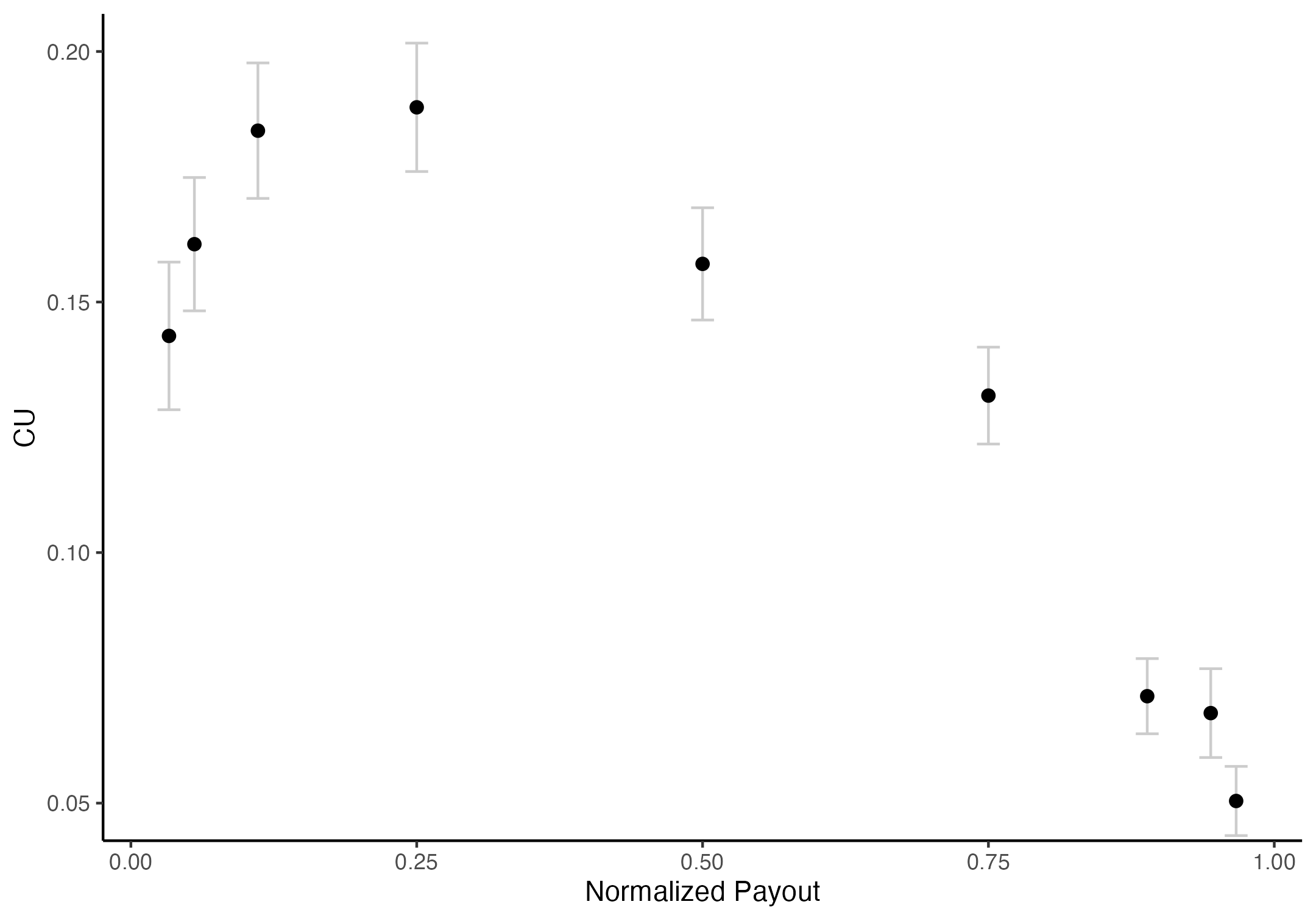}}
\end{subfigure}\\
\begin{subfigure}[t]{0.45\textwidth}
\centering
\caption{\centering Present Value Equivalents}
    \vspace{0.5em}
    {\includegraphics[width=65mm]{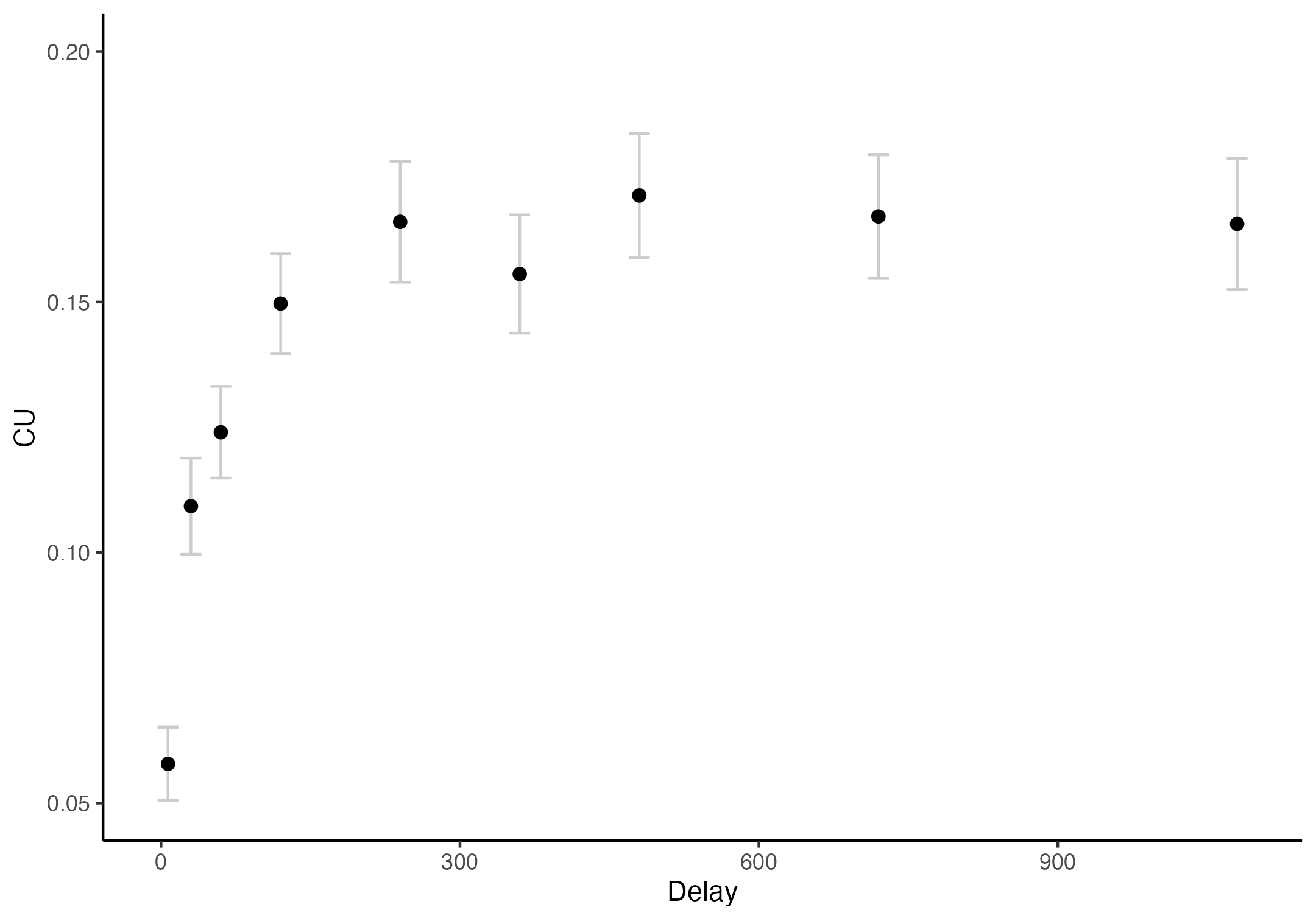}}
\end{subfigure}
\begin{subfigure}[t]{0.45\textwidth}
\centering
\caption{\centering Time Value Equivalents}
    \vspace{0.5em}
    {\includegraphics[width=65mm]{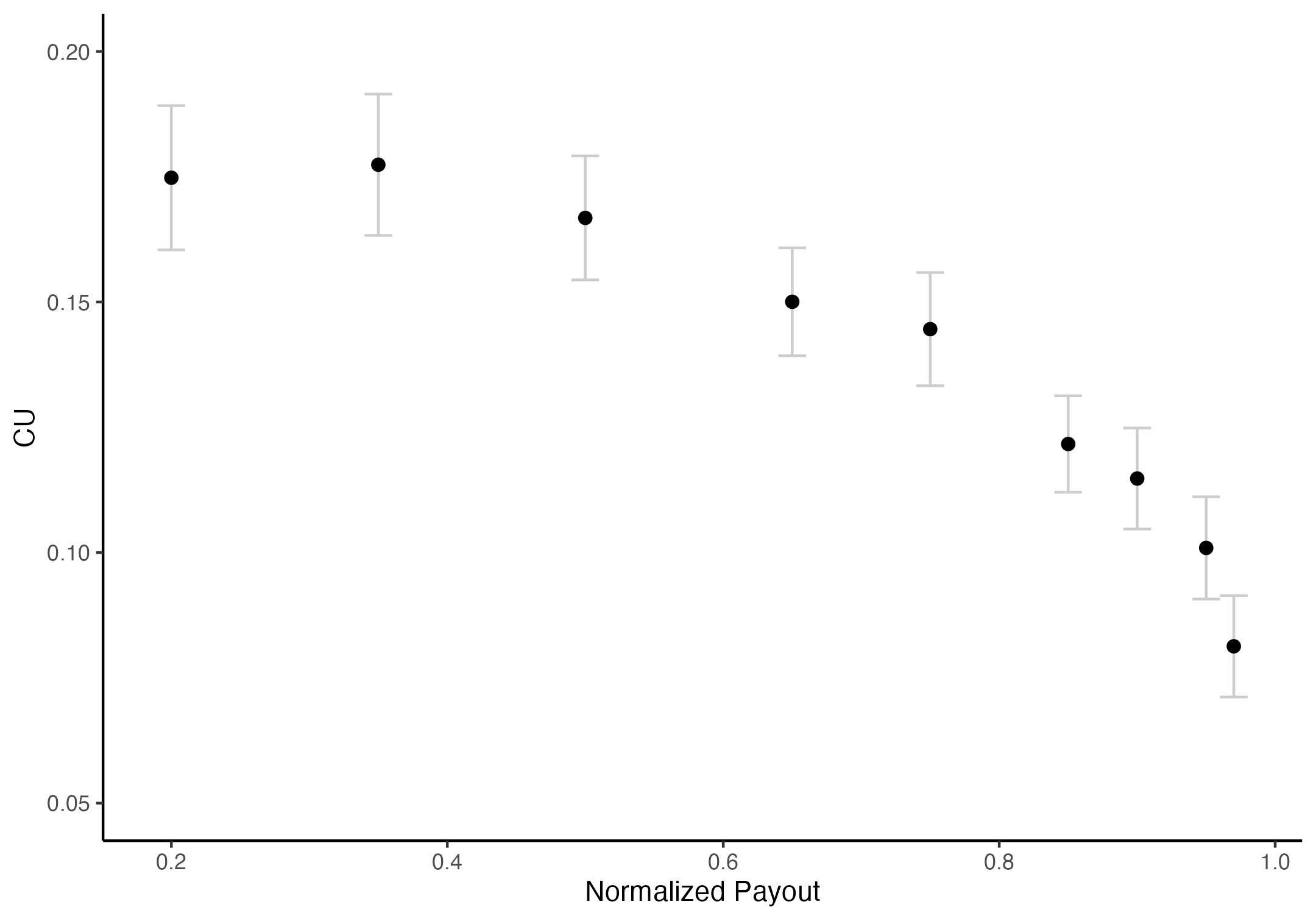}}
\end{subfigure}
\captionsetup{font=small}
\caption{Average CU across valuation paradigms. Error bars reflect 95\% confidence intervals. }
\label{fig:valuation_cu}
\end{figure}

\clearpage

\clearpage

\begin{center}
\Large{\textbf{ONLINE APPENDIX}}
\end{center}

\section{Appendix: Additional Theoretical Results}
\label{APP:theory}
Proofs of all results stated in this Appendix are compiled in \href{https://jeffreyyang97.github.io/personalwebsite/CC_OA.pdf}{Supplemental Appendix G}. 

\subsection{Characterization Results}
\label{APP:theory_axioms}

\subsubsection{Multiattribute Choice: Nonlinear Preferences}

We consider a more general multiattribute domain. Each option in $X\equiv X_1\times X_2\times ...\times X_n$ is defined on $n$ attributes, where each $X_i$ is a connected and separable topological space. Preferences are additively separable in each attribute, where the value of each $x\in X$ is given by $U(x)=\sum_{k} u_k(x_k)$. Say that $u_k$ is \textit{non-trivial} if there exist $x_k,x'_k\in X_k$ such that $u_k(x_k)\neq u_k(x'_k)$. Consider the following representation:

\begin{definition}
    $\tau$ has an additively separable $L_1$-\textit{complexity representation} if there exist continuous, non-trivial $u_i:X_i\to\mathbb{R}$ such that for $U(x)=\sum_{k}u_k(x_k)$ and $d_{L1}(x,y)=\sum_{k}|u_k(x_k)-u_k(y_k)|$, whenever $d_{L1}(x,y)\neq 0$, 
\begin{align*}
    \tau_{xy}=H\left(\frac{|U(x)-U(y)|}{d_{L1}(x,y)}\right)
\end{align*}
for $H$ continuous, increasing with $H(0)=0$, and $\tau_{xy}=0$ otherwise. Similarly, a binary choice rule $\rho$ has an additively separable $L_1$-\textit{complexity representation} if there exist continuous, non-trivial $u_i:X_i\to\mathbb{R}$ such that whenever $d_{L1}(x,y)\neq 0$,
\begin{align*}
    \rho(x,y)=G\left(\frac{U(x)-U(y)}{d_{L1}(x,y)}\right).
\end{align*}
for $G$ continuous, strictly increasing, and $\rho(x,y)=1/2$ otherwise. 
\end{definition} 

This representation mirrors the linear $L_1$-complexity representation discussed in the main text, except each utility-weighted attribute value $\beta_kx_k$ is replaced by its potentially non-linear counterpart $u_k(x_k)$. We provide an axiomatic characterization for this more general representation, in which Linearity is relaxed and replaced with two axioms. 

First some definitions. For $E\subseteq  I$, let $x_Ey$ denote the option that replaces the value of option $y$ along attributes $k\in E$ with $x_k$. Say that comparisons $(x,y),(w,z)\in \mathcal{D}$ are \textit{congruent} if for all $i\in I$, either $\rho(x_{\set{i}}y,y)\geq 1/2$ and $\rho(w_{\set{i}}z,z)\geq 1/2$, or $\rho(x_{\set{i}}y,y)\leq 1/2$ and $ \rho(w_{\set{i}}z,z)\leq 1/2$. That is, if $(x,y)$ and $(w,z)$ are congruent, the advantages and disadvantages in the two comparisons are located in the same attributes. 

\begin{enumerate}[label={M\arabic*}.]
  \setcounter{enumi}{5}
    \item \textbf{Separability:} $\rho(x_{E}z,y_{E}z)=\rho(x_Ez',y_Ez')$ for all $x,y,z,z'\in X$, $E\subseteq  I$. 
    \item \textbf{Tradeoff Congruence}: Suppose that $(x,y)$ is congruent to $(y,z)$, and $\rho(x,y),\rho(y,z)\geq 1/2$. Then $\rho(x,z)\leq \max\set{\rho(x,y),\rho(y,z)}$.
\end{enumerate}

Separability is the stochastic analog of the 
familiar coordinate independence axiom in deterministic choice, which says that $x_{E}z\succeq y_{E}z\implies x_{E}z'\succeq y_{E}z'$ for all $E\subseteq  I$, $x,y,z,z'\in X$. The interpretation of  Tradeoff Congruence is as follows: consider the attribute-wise tradeoffs involved in comparing $z$ to $y$ and $x$ to $y$, where $x$ is in fact better than $y$, and $y$ is better than $z$. The condition says that if replacing $y$ with $x$ in the first comparison and replacing $y$ with $z$ in the second only increases the magnitude of these tradeoffs---i.e., if $(x,y)$ and $(y,z)$ are congruent---then $(x,z)$ cannot be an easier comparison than both of the intermediate comparisons $(x,y)$ and $(y,z)$. Intuitively, both of these replacements only increase the size of the tradeoffs the DM must contend with, and so as revealed by choice probabilities, the DM cannot find the comparison $(x,z)$ easier than both $(x,y)$ and $(y,z)$.  

The following result states that Continuity, Moderate Transitivity, Dominance, Simplification, Separability, and Tradeoff Congruence characterize the additively separable representation, and that its primitives are identified from choice data. 

\begin{thm}
\label{THM:representation_additive} 
Suppose that $n>2$ and that all attributes are non-null. Then a binary choice rule $\rho$ satisfies M1, M3--M7 if and only if it has an additively separable $L_1$-complexity representation. Moreover, suppose that at least two attributes are non-null. If $\rho$ has  additively separable $L_1$ complexity representations $((u_i)_{i=1}^n,G)$ and $((u'_i)_{i=1}^n,G')$, then there exists $C>0$, $b_i\in \mathbb{R}$ such that $u_i'=Cu_i+b_i$ for all $i$, and $G'=G$.
\end{thm}

\subsubsection{Lottery Choice}
Consider the lottery choice domain, where $X$ is the set of finite state lotteries over $\mathbb{R}$. The CDF-complexity representation for $\tau$ implies the following binary choice representation:
\begin{definition}\label{def:CDF_complexity_rho}
A binary choice rule $\rho$ has a \textit{CDF}-Complexity representation if there exists $u:\mathbb{R}\to\mathbb{R}$ strictly increasing such that 
\begin{align*}
    \rho(x,y)=G\left(\frac{EU(x)-EU(y)}{d_{CDF}(x,y)}\right)
\end{align*}
for $G$ continuous, strictly increasing. 
\end{definition}

Let $\geq$ denote the partial order corresponding to first-order stochastic dominance. Let $S_x=\set{w\in \mathbb{R}: f_x(w)>0}$ denote the support of $x$. Consider the following axioms: 

\begin{enumerate}[label={L\arabic*}.]
    \item \textbf{Continuity:} $\rho(x,y)$ is continuous on its domain.
    \item \textbf{Independence:} $\rho(x,y)=\rho(\lambda x+(1-\lambda)z,\lambda y+(1-\lambda)z)$ for $\lambda\in (0,1)$.
    \item \textbf{Moderate Transitivity:} If $\rho(x,y)\geq1/2$ and $\rho(y,z)\geq1/2$, then either $\rho(x,z)> \min\set{\rho(x,y),\rho(y,z)}$ or $\rho(x,z)=\rho(x,y)=\rho(y,z)$.
    \item \textbf{Dominance:} If $x\geq y$, then $\rho(x,y)\geq \rho(w,z)$ for any $w,z\in X$, where the inequality is strict if $w\not\geq z$.
    
    \item \textbf{Simplification}: If $\rho(x,y)\geq 1/2$: for any $x'\in X$ with support in $S_x\cup S_y$ satisfying 
    \begin{enumerate}[label={(\arabic*}), leftmargin=1.5cm, itemsep=1mm]
        \item $F_{x'}(w^*)=F_{y}(w^*)$ for some $w^*\in S_x\cup S_y$,
        \item $F_{x'}(w)\neq F_x(w)$ for at most one $w\in S_x\cup S_y/\set{w^*}$,
    \end{enumerate}
    such that $\rho(x',x)=1/2$, we have $\rho(x',y)\geq\rho(x,y)$. 
\end{enumerate}

Axioms L1--L4 are direct analogs of M1--M4 in the characterization of $L_1$ complexity. Axiom L5 says that concentrating value differences in the same region of the payoff distribution makes lotteries easier to compare, and is an analog of the Simplification property (Axiom M5) for $L_1$ complexity. L1--L5 exhaust the behavioral content of CDF complexity. 

\begin{thm}
    \label{THM:representation_risk}
    A binary choice rule $\rho$ satisfies L1-L5 if and only if it has a CDF-Complexity representation $(G,u)$. Moreover, if $(G',u')$ also represents $\rho$, then $G'=G$ and there exists $C>0,b\in\mathbb{R}$ such that $u'=Cu+b$.
\end{thm}

In \href{https://jeffreyyang97.github.io/personalwebsite/CC_OA.pdf}{Supplemental Appendix F.3}, we develop and axiomatize a generalization of the CDF-complexity representation that allows for rank-dependent probability weighting. The characterization involves a suitable weakening of Independence and the addition of a lottery analogue of Tradeoff Congruence (M7), paralleling the characterization of additively separable complexity in multiattribute choice.

\subsubsection{Lottery Choice: Relationship to \citet{fishburn_probabilistic_1978}}
\citet{fishburn_probabilistic_1978} axiomatizes the \textit{incremental EU advantage} model, in which 
\begin{align*}
    \rho(x,y)=G\left(\frac{EU(x)-EU(y)}{d_{CDF}(x,y)}\right)
\end{align*}
for some strictly increasing $u$ and $G$ strictly increasing, with $G(1)=1, G(-1)=0$. It is identical to our Definition \ref{def:CDF_complexity_rho} except here, $G$ need not be continuous, and a FOSD-dominant lottery is chosen with probability 1. Below, we state Fishburn's axioms and discuss their relationship to ours.

\citet{fishburn_probabilistic_1978} characterizes this model using 8 axioms. The first four are similar or identical to L1--L4: 1) a continuity axiom, 2) Moderate Transitivity, 3) Independence, and 4) a strengthening of Dominance that requires $\rho(x,y)=1$ if and only if $x>y$. The remaining four take the place of our Simplification axiom. To state these axioms, write $x>_0y$ if and only if $x$ payoff-dominates $y$---that is, all possible outcomes in $x$ are greater than all possible outcomes in $y$, and write $(w_1,w_2,w_3)C(x,y)$ if and only if $w_1>w_2>w_3$, $f_x(w_i)+f_y(w_i)>0$ for each $i$, $f_x(w_2)=0$, $f_x(w)+f_y(w)=0$ for all $w\in(w_3,w_1)\setminus\set{w_2}$, and either $F_x(w_1)\geq F_y(w_1)$ or $F_y(w_2)\geq F_x(w_1)$. These axioms hold for $x,y,x',y'\in X$, $w,w',w''\in\mathbb{R}$, and $\lambda\in (0,1)$: 

\begin{enumerate}[label={\arabic*})]
\setcounter{enumi}{4}
    \item If $x>_0 w>_0y$ and $\rho(w,\frac{1}{2}x+\frac{1}{2}y)=1/2$, then $\rho(\lambda x+(1-\lambda)y,w)=\rho(w,\lambda y+(1-\lambda)x)$.
    \item If $x>_0w>_0y$, $x'>_0w'>_0y'$, and $\rho(w,\frac{1}{2}x+\frac{1}{2}y)=\rho(w',\frac{1}{2}x'+\frac{1}{2}y')$, then $\rho(w,\lambda x+(1-\lambda)y)=\rho(w',\lambda x'+(1-\lambda)y')$.
    \item If $\frac{1}{2}x+\frac{1}{2}y>_0\frac{1}{2}x'+\frac{1}{2}y'$, $\rho(x,y)=\rho(x',y')$, then $\rho(\frac{1}{2}x+\frac{1}{2}x',\frac{1}{2}y+\frac{1}{2}y')=\rho(x,y)$. 
    \item If $(w,w',w'')C(x,y)$ and $\rho(w',\lambda w+(1-\lambda)w'')=\frac{1}{2}$, then with $y=(1-f_y(w'))y'+f_y(w')w'$, $\rho(x,y)=\rho(x,(1-f_y(w'))y'+f_y(w')[\lambda w+(1-\lambda)w''])$. 
\end{enumerate}
See \citet{fishburn_probabilistic_1978} for an interpretation of these axioms. 

One important difference between this axiom system and ours is that Fishburn's additional axioms do not admit straightforward translations to multiattribute choice: 5), 6), and 7) involve a payoff dominance notion for lotteries that does not have an analog in multiattribute choice, and 8) involves the replacement of lottery outcomes, which also does not have a multiattribute analog. Fishburn's result therefore cannot be directly extended to multiattribute choice, and so our characterization for $L_1$ complexity is to our knowledge novel to the literature. 

Another distinction is that Fishburn's additional axioms involve mixture operations, whereas our Simplification axiom does not. This means that Fishburn's axiom system cannot as easily be adapted to characterize more general models that weaken Independence/Linearity. On the other hand, by weakening Independence/Linearity, our axiom system can be adapted to characterize a non-linear representation in multiattribute choice (Theorem \ref{THM:representation_additive}) and a rank-
dependent representation in lottery choice (Theorem 7 in \href{https://jeffreyyang97.github.io/personalwebsite/CC_OA.pdf}{Supplemental Appendix F.3}).

\subsubsection{Intertemporal Choice}
Consider intertemporal choice, where $X$ is the set of finite payoff streams. For a payoff flow $x\in X$, let $T_x=\set{t:m_x(t)\neq 0}$ denote the \textit{support} of $x$, and for $x,y\in X$ let $T_{xy}=T_x\cup T_y\cup\set{0,\infty}$ denote the \textit{joint support} of $x$ and $y$. We consider the following extension of Definition \ref{def:CPF_complexity} to general time discounting. Call $d:\mathbb{R}^{+}\cup \set{+\infty}\to\mathbb{R}^{+}$ a \textit{discount function} if $d$ is strictly decreasing and $d(\infty)=0$. We will consider discounted utility preferences of the form $DU(x)=\sum_{t}d(t)m_x(t)$. Note that $d$ need not be continuous, and so can capture discontinuous time preferences such as quasi-hyperbolic discounting. 

\begin{definition}
\label{def:CPF_complexity_general}
    $\tau$ has a generalized CPF-complexity representation if there exists a discount function $d$ such that 
    \begin{align*}
        \tau_{x,y}=H\left(\frac{|DU(x)-DU(y)|}{d_{CPF}(x,y)}\right)
    \end{align*}
for $H$ continuous, strictly increasing with $H(0)=0$, where $d_{CPF}(x,y)=\sum_{k=0}^{n-1}|M_x(t_k)-M_y(t_k)|\cdot(d(t_k)-d(t_{k+1}))$, for $t_0<t_1<...<t_{n}$ enumerating $T_{xy}$. Similarly, a binary choice rule $\rho$ has a generalized CPF complexity representation if 
\begin{align*}
        \rho(x,y)=G\left(\frac{DU(x)-DU(y)}{d_{CPF}(x,y)}\right)
    \end{align*}
for some continuous, strictly increasing $G$.
\end{definition}

Note that if $d$ is differentiable, $d_{CPF}$ takes the form $d_{CPF}(x,y)=\int_{0}^{\infty}|M_x(t)-M_y(t)|\cdot (-d'(t))\,dt$. In the case where $d(t)=\delta^t$, generalized CPF complexity reduces to Definition \ref{def:CPF_complexity}. Let $\geq$ denote the partial order on $X$ corresponding to temporal dominance (i.e., $x \geq y$ iff at every time $t \in \mathbb{R}^+\cup\{+\infty\}$, $M_x(t) \geq M_y(t))$.
Consider the following axioms: 

\begin{enumerate}[label={T\arabic*}.]
    \item \textbf{Continuity:} $\rho(x,y)$ is continuous on its domain.
    \item \textbf{Linearity:} $\rho(x,y)=\rho(\lambda x+(1-\lambda)z,\lambda y+(1-\lambda)z)$ for $\lambda\in (0,1)$.
    \item \textbf{Moderate Transitivity:} If $\rho(x,y)\geq1/2$ and $\rho(y,z)\geq1/2$, then either $\rho(x,z)> \min\set{\rho(x,y),\rho(y,z)}$ or $\rho(x,z)=\rho(x,y)=\rho(y,z)$.
    \item \textbf{Dominance:} If $x\geq y$, then $\rho(x,y)\geq \rho(w,z)$ for any $w,z\in X$, where the inequality is strict if $w\not\geq z$. 
    \item \textbf{Simplification}: If $\rho(x,y)\geq 1/2$, for any $x'\in X$ with support in $T_x\cup T_y$ satisfying 
    \begin{enumerate}[label={(\arabic*}), leftmargin=1.5cm, itemsep=1mm]
        \item $M_{x'}(t^*)=M_{y}(t^*)$ for some $t^*\in T_x\cup T_y$,
        \item $M_{x'}(t)\neq M_x(t)$ for at most one $t\in T_x\cup T_y/\set{t^*}$,
    \end{enumerate}
    such that $\rho(x',x)=1/2$, we have $\rho(x',y)\geq\rho(x,y)$. 
\end{enumerate}

Axioms T1--T4 are direct analogs of M1--M4 in the characterization of $L_1$ complexity. Axiom T5 says that concentrating value differences between payoff flows in the same time period makes them easier to compare, and is an analog of the Simplification property (Axiom M5) for $L_1$ complexity. T1--T5 exhaust the behavioral content of CPF complexity. 

\begin{thm}
\label{THM:representation_time}
A binary choice rule $\rho$ satisfies T1---T5 iff it has a generalized CPF-Complexity Representation $(G,d)$. Moreover, if $\rho$ is also represented by $(G',d')$, then $G'=G$, and there exists $C>0$ such that $d'=Cd$. 
\end{thm}

To characterize CPF complexity with exponential discounting preferences, an additional standard stationarity axiom is needed.

\begin{enumerate}[label={T\arabic*}.]
  \setcounter{enumi}{5}
    \item \textbf{Stationarity}. If $\rho(x,y)>1/2$: for $x',y',k>0$ s.t. $m_{x'}(t)=m_{x}(t-k)$, $m_{y'}(t)=m_{y}(t-k)$ for all $t\geq k$ and $m_{x'}(t)=m_{y'}(t)=0$ for all $t<k$, $\rho(x',y')\geq 1/2$. 
\end{enumerate}

\subsection{Menu Sequences and Tiebreaking in Multinomial Choice}
\label{APP:theory_tiebreaking}

\textbf{\textit{Menu Sequence Extension}}. Consider a \textit{menu sequence} $A^1,A^2,...,A^n\in \mathcal{A}$ in a choice context $C$. Here, the DM generates signals $s$ for each pairwise comparison in $A^1\cup A^2\cup...A^n\cup C$, and chooses the option from each menu with the highest posterior expected value (randomizing in the case of ties; see below), yielding joint choice frequencies
\begin{align*}
    \rho((x^1,...x^n),(A^1,...,A^n)|C)=\mathbb{P}\left(\bigcap_{i=1}^n\set{s: \mathbb{E}[v_{x^i}|s]>\mathbb{E}[v_y|s]\,\forall\, y\in A^i/\set{x^i} }\,\big|\,v\right).
\end{align*}
Here, $\rho((x^1,...x^n),(A^1,...,A^n)|C)$ records the frequency of choosing $x^i\in A^i$ for $i=1,...,n$. Given an option $x$ and a price list $Z$, a valuation task $(x,Z)$ is simply the binary menu sequence $A^1,...,A^n=\set{x,z^1},...,\set{x,z^n}$.\\

\noindent\textbf{\textit{Tiebreaking}}. Fix any choice problem $(A,C)$. If $s$ induces a tie among options that maximize posterior expected value, we assume a symmetric tiebreaking rule in which the DM randomizes between the maximal options. In particular, for any option $x\in A$ and signal realization $s$, let $\mathcal{N}(x,s)\equiv |\set{y\in A: \mathbb{E}[v_y|s]=\mathbb{E}[v_x|s]}|$ denote the number of options in $A$ with the same posterior expected value as $x$, and define the random variable
\begin{align*}
    c(x,s)\equiv 
    \begin{cases}
        1/\mathcal{N}(x,s) & \mathbb{E}[v_x|s]\geq \mathbb{E}[v_y|s]\,\forall\, y\in A\\
        0 & \text{otherwise}
    \end{cases}
\end{align*}
Choice probabilities are given by $\rho(x,A|C)=\mathbb{E}[c(x,s)\,|\, v]$. 

We assume the same tiebreaking rule in our extension to menu sequences wherein the DM independently randomizes between the maximal options in each menu. In particular, fix a menu sequence $((A^1,...,A^n),C)$. For any option $x\in A^i$ and signal realization $s$, let $\mathcal{N}^i(x,s)\equiv |\set{y\in A^i: \mathbb{E}[v_y|s]=\mathbb{E}[v_{x}|s]}|$ denote the number of options in $A^i$ with the same posterior expected value as $x$, and define 
\begin{align*}
    c^i(x,s)\equiv 
    \begin{cases}
        1/\mathcal{N}^i(x,s) & \mathbb{E}[v_x|s]\geq \mathbb{E}[v_y|s]\,\forall\, y\in A^i\\
        0 & \text{otherwise}
    \end{cases}
\end{align*}
Choice probabilities are given by $\rho((x^1,...,x^n),(A^1,...,A^n)|C)=\mathbb{E}\left[\prod_{i=1}^n c^i(x^i,s)\,\bigg|\, v\right].$

\subsection{Identification in Multinomial Choice}
\label{APP:theory_multinomial_id}
Call $\rho:X\times \mathcal{A}\times\mathcal{C}\to[0,1]$ a multinomial choice rule if $\sum_{x\in A}\rho(x,A|C)=1$ for all $(A,C)\in \mathcal{A}\times\mathcal{C}$. Our multinomial choice model is parameterized by the prior distribution $Q$, the value function $v:X\to \mathbb{R}$, and the signal precisions $\tau:\mathcal{D}\to \mathbb{R}^{+}$, where we make the additional assumption that $\tau(x,y)=0$ if and only if $v(x)=v(y)$. The following result states that $v$ is ordinally identified and $\tau$ is exactly identified. 

\begin{prop}
\label{PROP:multinomial_id}
    Suppose that a multinomial choice rule $\rho$ is represented by $(Q,v,\tau)$ and $(Q',v',\tau')$. Then $\tau'=\tau$ and there exists $\phi:\mathbb{R}\to\mathbb{R}$ strictly increasing such that $v'=\phi\circ v$. 
\end{prop}


\subsection{Decoy Effects}
\label{APP:decoy_effects}

We first state an implication of Proposition \ref{PROP:context} for multiattribute choice. 

\begin{customcor}
\label{COR:context}
Consider options from $X = \mathbb{R}^n$, with $v_x=\sum \beta_kx_k$, and suppose $\tau$ has an $L_1$-complexity representation. Let $v_x,v_y>v_z$.
\begin{enumerate}[label = (\roman*)]
    \item If $v_x=v_y$, then $d_{L1}(x,z)>d_{L1}(y,z)$ implies $\rho(y,x|\set{z})>1/2$.
    \item For any value difference $\Delta = |v_x-v_y|$, there exists $\underline{d} \in \mathbb{R}^+$ such that if $d_{L1}(x,y)>\underline{d}$, there exists $z\in X$ with $d_{L1}(x,z)>d_{L1}(y,z)$ such that $\rho(y,x|\set{z})>1/2$.
\end{enumerate}
\end{customcor} 

(i) says that if $x$ and $y$ are indifferent, then introducing an inferior phantom option $z$ that is more $L_1$-similar to $y$ than $x$ distorts choice in favor of $y$. (ii) says that if $x$ and $y$ are sufficiently $L_1$-dissimilar, there exists a decoy that distorts choice in favor of $y$. Importantly, the $L_1$ complexity measure predicts the comparability of options on the basis of their features, and thus the resulting context effects. As a result, our framework not only accommodates documented decoy and asymmetric dominance effects, as in related models \citep{natenzon_random_2019}; it \textit{predicts} precisely these patterns, as discussed below. 

\begin{ex}\normalfont (Classic decoy effects). Consider a setting where options have two attributes, where $\beta=(1,1)$, and where $\tau$ has an $L_1$ complexity representation. Consider two indifferent choice options $x=(1,2)$, $y=(2,1)$, and consider the effect of including a phantom option on choice shares between $x$ and $y$. 
\\
\\
\textit{Case 1}: $z=(1.8,0.8)$. Since $d_{L1}(x,z)>d_{L1}(y,z)$, we have $\rho(y,x|\set{z})>0.5$. We recover the classic asymmetric dominance effect: the addition of an option that is dominated by the target option $y$ but not by the competitor $x$ distorts choice in favor of $y$.  \\
\\
\textit{Case 2}: $z'=(1.5,1.1)$. Since $d_{L1}(x,z')>d_{L1}(y,z')$, we have $\rho(y,x|\set{z'})>0.5$. Here, the model predicts a ``good deal'' effect---$z'$ is not dominated by either $x$ or $y$, but its proximity to $y$ makes the target option seem like a ``good deal'' relative to $z'$, whereas its distance to $x$ prevents the DM from drawing the same inference about $x$.
\\
\\
\textit{Case 3}: $z''=(0.8,0.5)$. Here, $d_{L1}(x,z'')=d_{L1}(y,z'')$, and so Corollary \ref{COR:context} implies that $\rho(y,x|\set{z''})=0.5$. That is, the model predicts that the addition of a mutually dominated option does not affect choice shares.
\end{ex}
\noindent\textbf{\textit{Comparison to other context-dependent models}}.
Though some of the choice patterns above can be explained by existing context-dependent models, our model is distinct in simultaneously explaining all three.
The salience \citep{bordalo_salience_2013} and focusing models \citep{koszegi_model_2013} cannot rationalize the decoy effects in Cases 1 and 2. 
The relative thinking model \citep{bushong_model_2021}, in which the DM weighs a given change along an attribute by less when there is a larger range of values along that attribute, can rationalize the decoy effect in Case 1 as a result of option $z$ extending the range of attribute 2 more than attribute 1, but not Case 2, where $z'$ has no effect on attribute ranges. The pairwise normalization model \citep{landry_pairwise_2021} predicts that $z$ increases the attractiveness of $y$ relative to $x$ whenever $z_1/z_2$ is closer to $y_1/y_2$ than it is to $x_1/x_2$, and so can rationalize the decoy effects in both Cases 1 and 2, but also delivers the prediction that the addition of a mutually dominated option $z''$ will also distort choice in favor of $x$. Furthermore, all of these models are formulated in multiattribute choice, and so cannot easily explain documented asymmetric dominance/decoy effects in lottery \citep{soltani_range-normalization_2012} or intertemporal choice \citep{marini_decoy_2019}. Because our theory of comparison complexity extends to lottery and intertemporal choice, our model can also be applied to study decoy effects in these domains. 


We also make a conceptual distinction from these models. In our model, as in \cite{natenzon_random_2019}, context effects do not arise from a mechanical bias, but instead as a response to imperfect comparability: decoy options distort choice between $x$ and $y$ only when $x$ and $y$ are hard to compare. This is consistent with the fact that the attraction effect is muted when consumers face familiar choice contexts or have clear prior preferences \citep{huber_lets_2014}.

\subsection{Intertemporal Preference Reversals: Details}
\label{APP:theory_time_reversals}
 Consider the intertemporal domain, where $v_x = \sum_t \delta^t m_x(t)$ and $\tau_{xy} = \tau_{xy}^{CPF} = H\left(\frac{|PV(x) - PV(y)|}{d_{CPF}(x,y)}\right)$, with $H(1) = \infty$.
\begin{ex}
\normalfont
\label{EX:pve_reversals}
(Intertemporal reversals). Consider a DM with a monthly $\delta\leq 0.95$.
\begin{align*}
    &x:\quad \$27 \text{ in 750 days}\\
    &y:\quad \$8.25 \text{ in 30 days}
\end{align*}
Since $v_y\geq v_x$, the DM is more likely to choose $y$ over $x$ in direct choice.  However, when the DM values these options, the differential ease of comparing $x$ and $y$ to money today can cause $x$ to be valued higher than $y$, following the same logic as in Example \ref{EX:cequiv_reversals}.

Formally, the DM faces a valuation task $(\upsilon,Z)$, where $\upsilon=(m_{\upsilon},t_{\upsilon})$ is a delayed payment that pays out $m_{\upsilon}>0$ at time $t_{\upsilon}>0$, valued against a price list of immediate payments $Z=\set{z^1,...,z^n}$, where $z^k=(m_k,0)$.  Call $Z$ \textit{adapted} to a delayed payment $\upsilon$ if $m_k-m_{k+1}$ is constant in $k$ and $m_1=m_{\upsilon}$, $m_n=0$; we restrict attention to adapted price lists.  Let $PVE(\upsilon,Z)=1/2[m_{R(\upsilon,Z)-1}+m_{R(\upsilon,Z)}]$ denote the distribution over the DM's present value equivalents obtained from assigning each switching point to a valuation at the midpoint of the adjacent prices. Figure \ref{fig:pve_sims_pr} plots the present value equivalents simulated from our model for delayed payments $\upsilon$ with the same present value as $x$ and $y$, assuming  $\delta = 0.95$. Here, the high-delay option $x$ has a higher valuation than $y$, even though $\rho(y,x)\geq 1/2$.\footnote{Using choice vignettes, \citet{tversky_causes_1990} document similar intertemporal preference reversals.}

 \begin{figure}[t!]
    \centering
    \small
    \begin{subfigure}[t]{0.48\textwidth}
        \includegraphics[width=\linewidth]{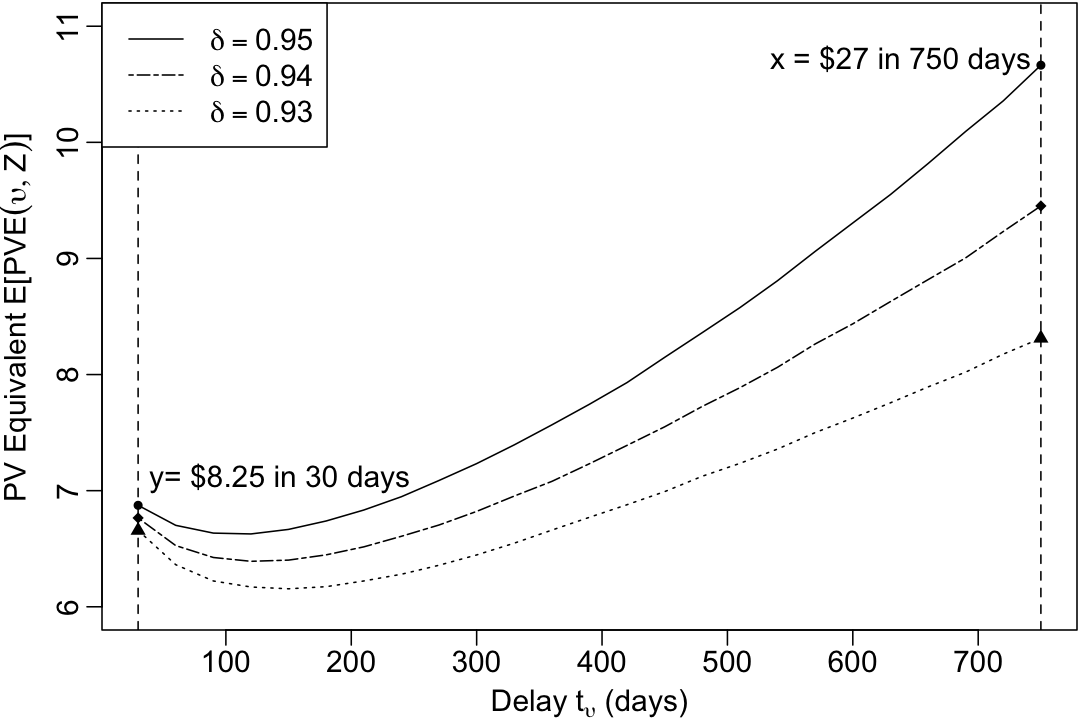}
        \caption{Present Value Equivalents $\mathbb{E}[PVE(\upsilon,Z)]$}
        \label{fig:pve_sims_pr}
    \end{subfigure}
    \hspace{1em}
    \begin{subfigure}[t]{0.48\textwidth}
        \includegraphics[width=\linewidth]{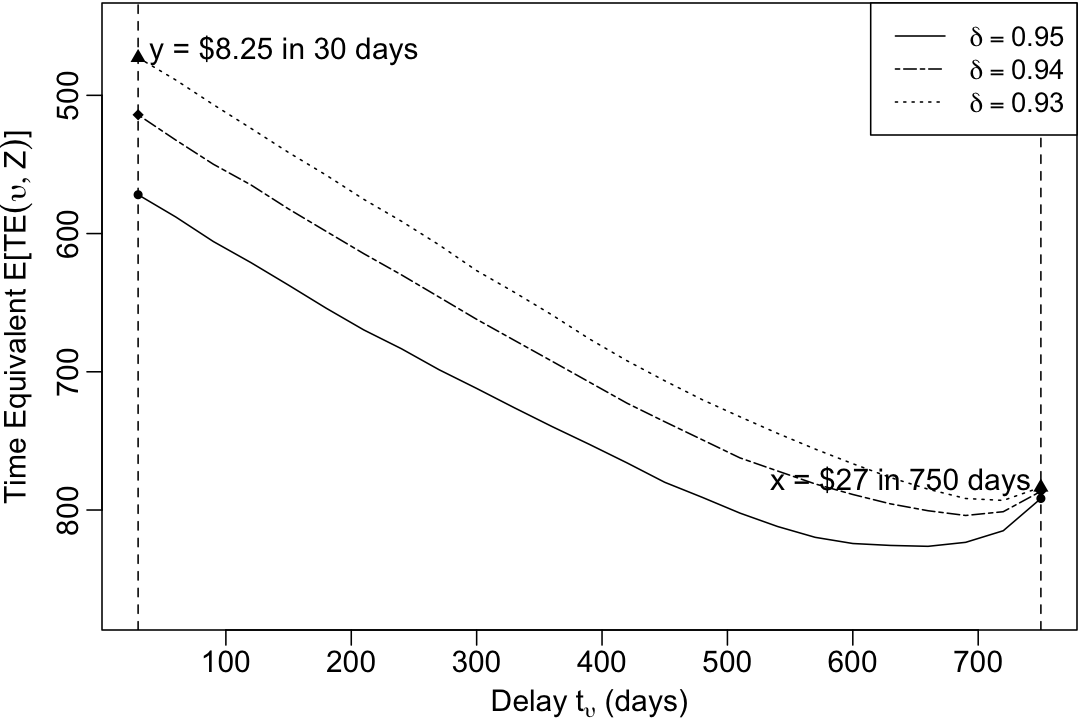}
        \caption{Time Equivalents $\mathbb{E}[TE(\upsilon,Z)]$}
        \label{fig:te_sims_pr}
    \end{subfigure}
    \caption{Simulated average present value equivalents and time equivalents for delayed payments $\upsilon=(m_\upsilon, t_\upsilon)$ with present value equal to that of $y=(8.25, 30)$ for $\delta = 0.95$. For PVEs, $Z$ is adapted to $\upsilon$ with $|Z|=15$. For TEs, $Z=\set{z^1,...,z^n}$, where $z^k=(27.5,t_{\upsilon}+t_k)$, for $(t_1,...,t_n)=(0, 7, 30, 60, 90, 120, 180, 240, 300, 360, 480, 600, 720, 900, 1080)$ days. $\tau$ has a CPF-complexity representation with $H(r)=(\Phi^{-1}(G(r)))^2$, for $G$ given by \eqref{eq:G_param} with $\kappa=0,\gamma=0.5$. Priors are $Q\sim U[0,1]$.}
    \label{fig:time_reversal}
\end{figure}

As we saw in lotteries, the direction of these distortions can be reversed by changing the units of valuation. Suppose the DM values options in terms of \textit{time equivalents}: the time $t$ that makes the delayed payment $(\$27.50, t)$ indifferent to $x$ or $y$.  Formally, the DM faces a valuation task $(\upsilon,Z)$, where $\upsilon=(m_{\upsilon},t_{\upsilon})$ is a delayed payment to be valued against a price list $Z=\set{z^1,...,z^n}$, where now each $z^k=(27.5,t_{\upsilon}+t_k)$ is a delayed payment; we restrict attention to time lists with $t_1=0$. We let 
\begin{align*}
 TE(\upsilon,Z) & = 
 \begin{cases}
  t_{\upsilon}+1/2[t_{R(\upsilon,Z)-1}+t_{R(\upsilon,Z)}] & R(\upsilon,Z)<n+1\\
  t_{\upsilon}+t_{n}+1/2 (t_n-t_{n-1}) & R(\upsilon,Z)=n+1
 \end{cases}
\end{align*}
denote the distribution over the DM's time equivalents obtained from assigning each switching point to a valuation at the midpoint of the adjacent delays. This change in numeraire flips the relative ease of comparing each option to prices, thereby reversing the direction of the valuation distortions, as Figure \ref{fig:te_sims_pr} illustrates. Specifically, we have $\mathbb{E}[TE
(y,Z)] < \mathbb{E}[TE(x,Z)]$ and so the ``reversal'' relative to direct choice disappears. 
\end{ex}

\subsection{Robustness of Simulated Predictions}
\label{APP:sim_robustness}

The simulations in Sections \ref{SEC:reversals} and \ref{SEC:valuation_biases} show that our model can rationalize several valuation anomalies: preference reversals, apparent probability weighting and hyperbolic discounting, and reversals of these patterns across valuation formats. Here, we report on additional simulation results to examine whether these patterns are directionally robust across a comprehensive range of parameters. In particular, we will show that while our model rules in these patterns, it rules out the opposite patterns across the parameter range.

These simulations use the same specifications of $\tau$ and uniform priors as in the main text. Recall that in both domains, $\tau$ is parameterized by a preference parameter that determines the value-dissimilarity ratio---a curvature parameter $\alpha>0$ that governs the Bernoulli utility function $u(w)=w^{\alpha}$ in lottery choice and a discount factor $\delta<1$ in intertemporal choice---as well as a global precision parameter $\gamma>0$ that governs $H(r)$, which maps the value-dissimilarity ratio into a value of $\tau$.\footnote{Specifically, in lottery choice, $\tau$ has a CDF-complexity representation with $u(w)=w^{\alpha}$ and $H(r)=(\Phi^{-1}(G(r)))^2$, for $G$ parameterized  by $\gamma$ according to \eqref{eq:G_param} with tremble parameter $\kappa=0$. In intertemporal choice, $\tau$ has a CPF-complexity representation with discount factor $\delta$ and $H(r)$ similarly parameterized by $\gamma$.} In these simulations, we consider the parameter ranges $\alpha\in\{0.2,\allowbreak 0.3,\allowbreak...,0.9,\allowbreak 1,\allowbreak 1/0.9,\allowbreak ...,\allowbreak 1/0.3,\allowbreak1/0.2\}$, $\delta\in\set{0.98,0.975,...,0.9}$ (where $\delta$ is a 30-day discount factor), and $\gamma\in\set{0.1,0.2,...,2}$. Higher values of $\gamma$ increase the global precision of the DM's comparison signals. The range of $\gamma$ covers comparisons ranging from nearly uninformative to highly precise. For example, at a value-dissimilarity ratio of $r=0.5$, the DM's median posterior belief in the correct ranking rises from $50.4\%$ at $\gamma=0.1$ to $93.4\%$ at $\gamma=2$.\\

\noindent\textbf{\textit{Preference Reversals.}} Section \ref{SEC:reversals} shows that our model can rationalize preference reversals in risk and time:  direct choice favors the low-risk/delay option, yet monetary valuations, i.e. certainty/present value equivalents, instead favor the high-risk/delay option. We assess the directional robustness of this prediction on the options from our preference reversal experiment, which include the low-risk lotteries $(\$4.50,0.98)$, $(\$4.75,0.94)$, $(\$5,0.90)$; the high-risk lotteries $(\$19.50,0.23)$, $(\$21.25,0.21)$, $(\$23.50,0.19)$; the low-delay options $(\$8.25,30)$, $(\$9.50,90)$, $(\$11,150)$; and the high-delay options $(\$24,630)$, $(\$25.50,690)$, $(\$27,750)$.\footnote{Following main text simulations, CE, PE and PVE price lists are adapted with $|Z|=15$. TE lists take the form $z^k=(27.5,t_{\upsilon}+t_k)$, for $(t_1,...t_n)=(0, 7, 30, 60, 90, 120, 180, 240, 300, 360, 480, 600, 720, 900, 1080)$ days, where $t_{\upsilon}$ is the delay of the option being valued.  } For each pairing of a low- and high-risk option, and analogously of a low- and high-delay option, the predicted reversals occur over a nontrivial range of parameters, which include the parameter values\footnote{Restricting $\kappa=0$, the best-fit parameters for our binary choice data (in terms of mean-squared error) are $(\alpha = 0.54, \gamma=0.38)$ for lottery choice and $(\delta=0.945 ,\gamma = 0.74)$ for intertemporal choice.} estimated from our binary choice experiments: for every preference parameter in the range under which a reversal can occur for a given pairing, i.e. binary choice favors the low-risk/delay option, there is some level of $\gamma$ in the range at which the predicted reversal occurs, i.e. average CEs (PVEs) instead favor the high risk/delay option. At the same time, reversals opposite to the predicted direction never occur: it is never the case that direct choice favors the high-risk (high-delay) option while average CEs (PVEs) favor the low-risk (low-delay) option. 

The model also predicts inconsistencies across valuation formats: monetary valuations favor the high-risk/delay option, yet valuations using an alternative numeraire---probability equivalents using a yardstick lottery paying \$24 with probability $p$ in lottery choice, and time equivalents using a yardstick option paying \$27.50 after $t$ days---favor the low-risk/delay option. This prediction is also directionally robust for the options above. The predicted inconsistencies occur over a nontrivial range of parameters, which again include the parameters estimated from our binary choice experiments: for every preference parameter and pairing of options, there is some level of $\gamma$ in which the predicted reversal occurs. The opposite inconsistency never occurs: it is never the case that average certainty equivalents (present value equivalents) favor the low-risk (low-delay) option while average probability equivalents (time equivalents) favor the high-risk (high-delay) option.\\

\noindent\textbf{\textit{Apparent Probability Weighting and Hyperbolic Discounting.}} Section \ref{SEC:valuation_biases} shows that the model rationalizes apparent probability weighting and its reversal across valuation formats: CEs of binary lotteries $(\overline{w},p)$ overvalue low probabilities and undervalue high probabilities, and  PEs of certain payments $(w,1)$ elicited against a yardstick lottery $(\overline{w},q)$ overvalue low normalized payments $w/\overline{w}$ and undervalue high normalized payments, thereby producing the opposite distortion. We show that this pattern is directionally robust across the parameter range. We simulate CEs for $p\in\set{0.02,0.06,\ldots,0.98}$ and PEs for $w/\overline{w}\in\set{0.02,0.06,\ldots,0.98}$.\footnote{Following our main-text simulations, all CE and PE price lists are adapted with $|Z|=15$.}  Across the full parameter range, both CE and PE distortions---defined as the average simulated valuations minus the true valuation---cross zero exactly once and in the predicted direction: they are positive for low probabilities/normalized payoffs and negative for high probabilities/normalized payoffs. Thus, the model robustly generates apparent probability weighting in CEs and its reversal in PEs, and never generates the opposite patterns.

Section \ref{SEC:valuation_biases} also shows that the model rationalizes apparent hyperbolic discounting and its reversal across valuation formats: PVEs of delayed payments $(\overline{m},t)$ undervalue low delays and overvalue high delays, while TEs of immediate payments $(m,0)$, elicited against a yardstick option $(\overline{m},t)$, undervalue low normalized payments $m/\overline{m}$ and overvalue high normalized payments, thereby producing the opposite discounting pattern. We show that this pattern is directionally robust across the parameter range. We simulate PVEs for $t\in\set{30,90,\ldots,1440}$ days and TEs for $m/\overline{m}\in\set{0.05,0.10,\ldots,0.95}$.\footnote{Following our main text simulations, all PVE price lists are adapted with $|Z|=15$, all TE price lists take the form $z^k=(\overline{m},t_k)$, for $(t_1,...t_n)=(0, 7, 30, 60, 90, 120, 180, 240, 300, 360, 480, 600, 720, 900, 1080)$ days.  }  Across the full parameter range, both PVE and TE distortions---defined as the average simulated valuations minus the true valuation---have the predicted endpoint distortions: they are negative at low delays/normalized payoffs and positive at high delays/normalized payoffs. PVE distortions cross zero exactly once across the full parameter range, whereas TE distortions cross exactly once for 300 of 340 parameterizations (88\%); in the remaining 40 cases, there are multiple crossings. Measured by the peak-to-trough range of the distortion within the recrossing region, the maximum oscillation is $10$ days, less than 1\% of the price list range. Thus, the model robustly generates apparent hyperbolic discounting in PVEs and its reversal in TEs, and never generates the opposite patterns.

\subsection{Alternative Explanations for Behavioral Regularities}
\label{APP:alt_models}

Here, we contrast our model with existing formal explanations of the behavioral regularities discussed in the paper. \\

\noindent \textbf{\textit{Decoy and Asymmetric Dominance Effects}}. See the discussion in Appendix \ref{APP:decoy_effects}. \\

\noindent \textbf{\textit{Preference Reversals}}. A classic explanation for lottery preference reversals is preference intransitivities implied by regret theory \citep{loomes_rationale_1983}. In this model, pairwise lottery choices are driven by anticipated regret from comparing the chosen outcome to what would have occurred under the rejected alternative. Under restrictions on the model, this can lead to preference cycles between lotteries and sure amounts, generating reversals. However, under these same restrictions, the model predicts that the same reversals persist when lotteries are valued using probability equivalents as in Example \ref{EX:pequiv_reversals} \citep{butler_imprecision_2007}. As such, the model cannot explain why preference reversals are eliminated by such changes in the numeraire. Moreover, since this model is developed in risky choice, it does not provide an explanation for why similar reversals occur in intertemporal choice. 

Salience theory \citep{bordalo_salience_2012} offers another explanation for lottery preference reversals. In this model, the choice context determines which lottery states are salient, and salient states are overweighted when assessing value. \citet{bordalo_salience_2012} assume that when pricing a lottery in isolation, the choice context is the lottery itself and a sure outcome of 0; this leads to a pricing bias in favor of the riskier lottery which can rationalize reversals between pricing and choice. If we are committed to the paper's assumption that the difference between pricing and choice is that a sure outcome of 0 enters the choice context in pricing, the same pricing bias in favor of the riskier lottery persists regardless of the unit of valuation. Therefore, the model cannot rationalize why preference reversals are eliminated under probability equivalents, and as with regret theory, does not provide a ready explanation for intertemporal reversals. 

\citet{blavatskyy_preference_2009} offers an explanation of preference reversals that, like ours, is rooted in noise in the comparison of lotteries to prices. In this model, the primitive is a binary choice rule $\rho(x,y)$; the certainty equivalent of a lottery $x$ follows the cumulative distribution function $F(w)=\rho(w,x)$, where $w$ denotes a degenerate lottery paying $w$ for sure. \citet{blavatskyy_preference_2009} shows that for a range of specifications of $\rho$, skewness in the resulting certainty equivalents can produce reversals between average valuations and choice. However, the model imposes a strong form of consistency between choice and valuation: the median valuation of each lottery coincides with its binary-choice indifference point (i.e. if $\rho(x,y)>1/2$, the median valuation of $x$ must be higher than $y$). It therefore necessarily rules out any inconsistencies between median valuations and binary choice. Our experiments document precisely this excluded pattern: binary choices favor the low-risk lotteries even when median certainty equivalents favor the high-risk lotteries, with similar reversals in intertemporal choice (see Section \ref{SEC:reversal_exps}). In contrast, because our approach to modeling valuations through binary comparison difficulty predicts systematic valuation distortions, it can rationalize these patterns. \\

\noindent\textbf{\textit{Probability Weighting, Hyperbolic Discounting, and Instability}}. There are many models of probability weighting, including standard ``preference-based'' accounts that fix these distortions as part of the DM's value function, most notably cumulative prospect theory \citep{kahneman_prospect_1979}, and more recent accounts based on cognitive frictions \citep{steiner_perceiving_2016,vieider_decisions_2024,frydman_source_2025,enke_cognitive_2023}. Neither explanation offers a ready explanation for the reversal of probability weighting in probability equivalents vs. certainty equivalents.

Within the class of prospect-theoretic explanations, \citet{sprenger_endowment_2015} shows that a model of stochastic reference points and loss aversion predicts higher risk aversion in probability vs. certainty equivalents. This account cannot explain why probability equivalents are simultaneously more risk seeking (for large probabilities) and more risk averse (for small probabilities) than certainty equivalents. \citet{feldman_certain_2024} show that by allowing utility curvature to differ across valuation procedures, a model of stochastic reference points can generate a difference between probability vs. certainty equivalents similar to our model predictions. Instead, our framework explains these patterns as expressions of the \textit{same} underlying preference. Importantly, since these accounts are formulated in risky choice, they cannot speak to why we observe similar reversals of hyperbolic discounting in intertemporal choice. 

Accounts based on cognitive frictions, which model how valuations of binary lotteries are distorted as a Bayesian response to imprecise perception of lottery features, also struggle to explain the sensitivity of behavior to the elicitation context. Assuming that cognitive imprecision distorts the valuation of risky lotteries and certain payments in the same way across elicitation formats, it is irrelevant from the perspective of these models which format is used. In order to accommodate documented reversal of probability weighting, one would need to model cognitive imprecision in a context-dependent way---precisely the approach taken in this paper.\footnote{This literature has studied a different form of context-dependence: dependence on the history of observed lotteries, which influences the DM's prior \citep{frydman_source_2025}. Our approach to modeling context-dependent cognitive noise is  complementary.} One recent paper that takes this approach is \cite{bouchouicha_cognitive_2026}, which develops a cognitive noise model to explain why the fourfold pattern in lottery valuations tends to disappear in binary choice. In addition to predicting this pattern, our model can simultaneously rationalize inconsistencies between different valuation formats. 

Similarly, preference-based accounts of hyperbolic discounting, which appeal to non-stationary time preferences, struggle to rationalize why valuations of delayed payments are invariant to front-end delays \citep{cohen_measuring_2020} and why hyperbolic discounting should reverse under alternative elicitation methods.

\subsection{Relationship to Linear Differentiation Model}
\label{APP:L_2}

In this section, we contrast the $L_1$-Complexity model (Definition \ref{def:L1_complexity_rho}) against the Linear Differentiation Model (LDM) proposed in \citet{he_random_2023}. We state the LDM representation below, and discuss two axes along which the models differ. 

\begin{definition} \citep{he_random_2023}.
    \label{def:L2_complexity}
    A binary choice rule $\rho$ has a linear differentiation representation if there exists $\beta\in \mathbb{R}^n$, an $n\times n$ symmetric positive-definite matrix $\Sigma$, and a continuous, strictly increasing function $G$, such that 
    \begin{align*}
        \rho(x,y)=G\left(\frac{\beta'(x-y)}{\sqrt{(x-y)'\Sigma(x-y)}}\right).
    \end{align*}
\end{definition}

\noindent \textit{\textbf{Dominance}}. As Theorem \ref{THM:representation} implies, the $L_1$-Complexity model satisfies a dominance property (M4), wherein if $x$ attribute-wise dominates $y$ (written $x>_{D} y$), the choice probability $\rho(x,y)$ is maximal. Consider a weaker dominance notion, which requires only that if $x>_{D} y$ and $w\not>_{D} z$, then $\rho(x,y)\geq \rho(w,z)$. When there are three or more attributes, the LDM violates this dominance notion. 

\begin{prop}
\label{PROP:L2_dom}
Suppose $\rho$ has a linear differentiation representation and at least 3 attributes are non-null. There exists $x,y,w,z\in\mathbb{R}^n$, such that $x>_{D}y$ and $w\not>_{D}z$, such that $\rho(x,y)<\rho(w,z)$.
\end{prop}

\noindent\textit{\textbf{Monotonicity}}. Say that a binary choice rule $\rho$ is \textit{weakly monotonic} if $x'>_{D}x$ implies $\rho(x',y)\geq \rho(x,y)$ for any $y\in\mathbb{R}^n$, i.e. improving a choice option along each attribute cannot decrease its probability of being chosen over some other choice option. The $L_1$-complexity model satisfies such a monotonicity property, whereas in general, the LDM violates monotonicity.  

\begin{prop}
\label{PROP:L2_monotonicity}
If $\rho$ has an $L_1$-complexity representation, $\rho$ is weakly monotonic. If instead $\rho$ has a linear differentiation representation, and at least 2 attributes are non-null, then there exists $x,x',y\in\mathbb{R}^n$ with $x'>_{D}x$ such that $\rho(x,y)>\rho(x',y)$.
\end{prop}

The intuition for the monotonicity violations produced by the LDM stems from a property of the model formalized in Proposition 1 from \citet{he_random_2023}, which states that any $x$ that solves $\max_{x}\rho(x,y)$ lies along the same one-dimensional subset of $\mathbb{R}^n$ containing $y$. An immediate implication of this property is that for any $x'$ that dominates $x$, we will have $\rho(x',y)<\rho(x,y)$ so long as $x'$ does not also lie in that subset. Finally, notice that this property implies that LDM and $L_1$ complexity are disjoint, as under the latter, $\arg\max_{x}\rho(x,y)$ is a multidimensional subset of $\mathbb{R}^n$.

\section{Appendix: Proofs of Main Text Results}
\label{APP:proofs}

\subsection{Characterization of $L_1$-Complexity}
Begin with some basic observations. Let $X$ be a space of options, and let $\mathcal{D}=\set{(x,y)\in X\times X:x\neq y}$. Say  $\rho:\mathcal{D}\to [0,1]$ is a \textit{binary choice rule} on $X$ if $\rho(x,y)=1-\rho(y,x)$. 

Call a (complete) binary relation $\succeq$ on $X$ the \textit{stochastic order} induced by a binary choice rule $\rho$ if for all $x\neq y$, $x\succeq y$ if $\rho(x,y)\geq 1/2$, and for all $x\in X$, $x\succeq x$. Say that a binary choice rule $\rho$ satisfies \textit{moderate transitivity} if for $\rho(x,y),\rho(y,z)\geq 1/2$, then $\rho(x,z)> \min\set{\rho(x,y),\rho(y,z)}$ or $\rho(x,z)=\rho(x,y)=\rho(y,z)$. Say that a binary choice rule $\rho$ satisfies \textit{weak transitivity} if for $\rho(x,y),\rho(y,z)\geq 1/2$, $\rho(x,z)\geq 1/2$. Consider a partial order $\geq_X$ on $X$. Say that $\rho$ satisfies \textit{monotonicity} with respect to $\geq_X$ if $x'\geq_X x$ implies $\rho(x',y)\geq \rho(x,y)$, strict whenever $x\not\geq_X x'$, $x\not\geq_X y$ and $y\not\geq_X x$. Say that $\rho$ satisfies \textit{dominance} with respect to $\geq_X$ if $x\geq_X y$ implies $\rho(x,y)\geq \rho(w,z)$ for all $w,z\in X$, where the inequality is strict if $w\not\geq_X z$. 

\begin{lemma}
    \label{LEM:monotonicity}
   If $\rho$ defined on $X$ satisfies moderate transitivity and dominance with respect to a partial order $\geq_X$, then it satisfies monotonicity with respect to $\geq_X$.
\end{lemma}

\begin{proof}
Take any options $x,y$, and suppose $x'\geq_X x$. If $x\geq_X x'$, then $x'=x$ since $\geq_X$ is a partial order and is therefore antisymmetric, and we are done. Now consider the case where $x\not\geq_X x'$. Note that if $x\geq_X y$, since $\geq_X$ is transitive we also have $x'\geq_X y$, and so Dominance implies that $\rho(x',y)\geq \rho(x,y)$ and we are done. 

Now consider the case where $x\not\geq_X y$. Let $\succeq$ denote the stochastic order induced by $\rho$; since $\rho$ satisfies MST, $\succeq$ is complete and transitive. By dominance, we have $\rho(x',x)>\rho(x,x')\implies \rho(x',x)>1/2$ and so $x'\succ x$. There are three cases: 

\textit{Case 1}: $x'\succeq x\succeq y$. By moderate transitivity, $\rho(x',y)> \min\set{\rho(x',x),\rho(x,y)}$ or $\rho(x',y)=\rho(x',x)=\rho(x,y)$. But since $\rho(x',x)> \rho(x,y)$ by dominance, it must be the case that $\rho(x',y)> \rho(x,y)$. \textit{Case 2}: $x'\succeq y\succeq x$. By definition of $\succeq$, $\rho(x',y)\geq 1/2\geq \rho(x,y)$. Also, since $x'\succ x$, we must have one of $x'\succ y$ or $y\succ x$, and so by definition of $\succeq$ we must have one of $\rho(x',y)> 1/2$ or $1/2> \rho(x,y)$, which implies $\rho(x',y)>\rho(x,y)$. \textit{Case 3}: $y\succeq x'\succeq x$. Toward a contradiction, suppose $\rho(y,x')> \rho(y,x)$. By moderate transitivity, $\rho(y,x)>\min\set{\rho(y,x'),\rho(x',x)}$ which implies $\rho(y,x)>\rho(x',x)$, which contradicts dominance and so $\rho(y,x')\leq \rho(y,x)\implies \rho(x',y)\geq \rho(x,y)$.

All that remains is to show that $\rho(x',y)> \rho(x,y)$ when $x\not\geq_X y$ and $y\not\geq_X x$. We have already shown this in Cases 1 and 2; all that remains is to show that the inequality is strict in Case 3. Suppose $y\succeq x'\succeq x$. Toward a contradiction, suppose that $\rho(y,x')\geq \rho(y,x)$. Moderate transitivity then implies that either (i) $\rho(y,x)>\min\set{\rho(y,x'),\rho(x',x)}$ or (ii) $\rho(y,x)=\rho(y,x')=\rho(x',x)$. As we saw above, it cannot be the case that (i) holds. If (ii) holds, then dominance implies that $y\geq x$, a contradiction. 
\end{proof}

For the following result, we consider the case where $X$ is a convex set. Say $\rho$ is \textit{linear} if $\rho(x,y)=\rho(\lambda x+(1-\lambda)z,\lambda y+(1-\lambda)z)$ for all $x,y,z\in X$, $\lambda\in (0,1)$. Say $\rho$ is \textit{superadditive} if for any $x,y,x',y'$ with $\rho(x,y),\rho(x',y')\geq 1/2$, for any $\lambda\in [0,1]$ we have $\rho(\lambda x+(1-\lambda)x',\lambda y+(1-\lambda)y')\geq \min\set{\rho(x,y),\rho(x',y')}$. 

\begin{lemma}\label{LEM:superadditivity}
Let $X$ be a vector space. If $\rho$ defined on $X$ satisfies moderate transitivity and linearity, then $\rho$ is superadditive.
\end{lemma}
\begin{proof}
Since $X$ is a vector space, linearity implies $\rho(x,y)=\rho(C x,C y)$ and $\rho(x,y)=\rho(x-z,y-z)$ for any $C>0,x,y,z\in X$. Consider $x,y,x',y'$ with $\rho(x,y),\rho(x',y')\geq 1/2$, and $\lambda\in [0,1]$. By linearity, $\rho(\lambda(x-y),0)=\rho(x,y)\geq 1/2$, $\rho(0,-(1-\lambda)(x'-y'))=\rho(x',y')\geq 1/2$, and $\rho(\lambda x+(1-\lambda)x',\lambda y +(1-\lambda)y')=\rho(\lambda(x-y),-(1-\lambda)(x'-y'))$. This, in conjunction with moderate transitivity, implies that
\begin{align*}
\rho(\lambda x+(1-\lambda)x',\lambda y +(1-\lambda)y')&=\rho(\lambda(x-y),-(1-\lambda)(x'-y'))\\
&\geq \min\set{\rho(\lambda(x-y),0),\rho(0,-(1-\lambda)(x'-y'))}\\
&=\min\set{\rho(x,y),\rho(x',y')}
\end{align*}
\end{proof}
\subsubsection*{Proof of Theorem \ref{THM:representation}.}
Necessity of the axioms is immediate from the definition. We now show sufficiency. 

Assume that M1--M5 holds. Let $\succeq$ denote the stochastic order on $\mathbb{R}^n$  induced by $\rho$. By weak transitivity, $\succeq$ is transitive. 
Since $\rho$ satisfies Continuity and Linearity, $\succeq$ satisfies axioms D1--D3 of Theorem 9.1  of \citet{gilboa_theory_2009}. Invoking an intermediate step in the proof of this theorem, we conclude that there exist weights $\beta\in \mathbb{R}^n$ such that $U(x)=\sum_k\beta_kx_k$ represents $\succeq$. Since all attributes are non-null, we have that $\beta_k\neq 0$ for all $k$. For the remainder of the proof, we henceforth identify each option $x$ with its weighted attribute values, so that $U(x)=\sum_k x_k$. Since $\rho$ satisfies Dominance and MST, Lemma \ref{LEM:monotonicity} implies that $\rho$ satisfies monotonicity with respect to the component-wise dominance relation on $\mathbb{R}^n$.  

For $z\in \mathbb{R}^n$, Let $d^{+}(z)=\sum_{k:z_k\geq 0}z_k$ and $d^{-}(z)=\sum_{k:z_k< 0}|z_k|$ denote the summed advantages and disadvantages in the comparison between $z$ and $0$. Say that $z$ has \textit{no dominance relationship} if $d^{+}(z),d^{-}(z)>0$.\\
\\
\textbf{Claim 1.} For any $z\in\mathbb{R}^n$ satisfying $\sum_k z_k\geq 0$, $\rho(z,0)=\rho(d^+(z)e_1-d^-(z)e_2,0)$.

\begin{proof}
For $i,j\in \set{1,...,n}$, $i\neq j$, define $z^{ij}\in\mathbb{R}^n$ satisfying 
\begin{align*}
    z^{ij}_k=
    \begin{cases}
    d^{+}(z) & k=i\\
    -d^{-}(z) & k=j\\
    0 & \text{otherwise}
    \end{cases}
\end{align*}

Note that because we have normalized utility weights 1, for all $i\neq j$, $l\neq m$, we have $U(z^{ij})=d^{+}(z)-d^{-}(z)=U(z^{lm})$, and so $z^{ij}\sim z^{lm}$. We will first show that $\rho(z^{ij},0)=\rho(z^{lm},0)$ for all $i\neq j,l\neq m$. It is sufficient to show that for all $i,j$ $\rho(z^{ij},0)=\rho(z^{12},0)$. There are two cases to consider:\\ 
\\
\textbf{Case 1}: $j>i$. Since $z^{1j}\sim z^{ij}$, and since $z^{1j}_i=0$, $z^{1j}_k=z^{ij}_k$ for all $k\neq i,1$, Simplification implies that $\rho(z^{1j},0)\geq \rho(z^{ij},0)$. Also, since $z^{ij}_1=0$, and $z^{ij}_k=z^{1j}_k$ for all $k\neq 1,i$, Simplification implies $\rho(z^{1j},0)\leq \rho(z^{ij},0)$, and so $\rho(z^{1j},0)=\rho(z^{ij},0)$. An analogous argument yields $\rho(z^{12},0)=\rho(z^{1j},0)$, and so $\rho(z^{ij},0)=\rho(z^{12},0)$. \\
\\
\textbf{Case 2}: $j<i$. By analogous arguments as above, we have $\rho(z^{ij},0)=\rho(z^{nj},0)$, $\rho(z^{nj},0)=\rho(z^{n2},0)$ and $\rho(z^{n2},0)=\rho(z^{12},0)$, and so $\rho(z^{ij},0)=\rho(z^{12},0)$ as desired.  \\

Let $K^{+}=\set{i\in \set{1,2,...,n}: z_i\geq 0}$ and $K^{-}=\set{i\in \set{1,2,...,n}: z_i< 0}$. Defining 
$\lambda_i=\frac{z_i}{\sum_{k\in K^{+}}{z_k}}$ for $i\in K^{+}$, and $\gamma_j=\frac{z_j}{\sum_{k\in K^{-}}z_k}$ for $j\in K^{-}$, note that $z=\sum_{i\in K^{+}}\sum_{j\in K^{-}}\lambda_i\gamma_jz^{ij}$, and so $z$ can be expressed as a mixture of $z^{ij}$'s. Since $\rho$ satisfies superadditivity by Lemma \ref{LEM:superadditivity}, by inductive application of superadditivity, we have $\rho(z,0)\geq \rho(z^{ij},0)$ for all $i\neq j$, which in turn implies $\rho(z,0)\geq \rho(z^{12},0)$. 

Note that by repeated application of Simplification, we have $\rho(z,0)\leq \rho(z^{ij},0)$, for some $i$ where $z_i\geq 0$, and some $j$ where $z_j\leq 0$. Since $\rho(z^{ij},0)=\rho(z^{12},0)$, we have $\rho(z,0)\leq \rho(z^{12},0)$, and so $\rho(z,0)=\rho(z^{12},0)$ as desired. 
\end{proof}
\noindent\textbf{Claim 2}. For $z$ with $\sum_{k}z_k\geq 0$, $\rho(z,0)=\tilde{G}\left(\frac{d^+(z)-d^-(z)}{d^+(z)+d^-(z)}\right)$ for some strictly 
increasing, continuous $\tilde{G}:[0,1]\to \mathbb{R}$. 

\begin{proof}
Fix $z$ with $\sum_{k}z_k\geq 0$, and suppose $z$ has no dominance relationship, that is $d^{+}(z)>0$, $d^{-}(z)>0$. Claim 1 implies $\rho(z,0)=\rho(d^{+}(z)e_1-d^{-}e_2,0)$. Define $F:[1,\infty)\to [1/2,1)$ by $F(t)=\rho(te_1-e_2,0)$; by monotonicity of $\rho$, $F$ is strictly increasing. Linearity implies $\rho(d^{+}(z)e_1-d^{-}e_2,0)=\rho((d^{+}(z)/d^{-}(z))e_1-e_2,0)=F(d^{+}(z)/d^{-}(z))$.

Let $\varphi(z)=\frac{z-1}{z+1}$;  and define $\tilde{G}:[0,1)\to\mathbb{R}$ where $\tilde{G}(z)=F(\varphi^{-1}(z))$; since $\varphi$ and $F$ are strictly increasing, $\tilde{G}$ is strictly increasing. By construction, we have $F(z)=\tilde{G}\left(\frac{z-1}{z+1}\right)$, and so $\rho(z,0)=\rho(d^{+}(z)e_1-d^{-}e_2,0)=\tilde{G}\left(\frac{d^+(z)-d^-(z)}{d^+(z)+d^-(z)}\right)$. Since $\rho$ is continuous, $\tilde{G}$ is continuous on its domain $[0,1)$, and in particular is uniformly continuous since it is increasing and bounded. Take the continuous extension of $\tilde{G}$ to $[0,1]$.  

Now consider the case where $z$ has a dominance relationship; that is $d^{+}(z)>0$, $d^{-}(z)=0$. By Dominance, $\rho(z,0)=\rho(d^{+}(z)e_1-d^{-}e_2,0)$ takes on some constant value $q$ such that $q>\rho(z',0)$ for all $z'$ without a dominance relationship, which implies that $q>\tilde{G}(t)$ for all $t\in[0,1)$. Since $\rho$ is continuous, it must be the case that $q=\tilde{G}(1)$.
\end{proof}

Now, let $G:[-1,1]\to \mathbb{R}$ be the symmetric extension of $\tilde{G}$ satisfying $G(z)=\tilde{G}(z)$ for $z\geq 0$, and $G(z)=1-\tilde{G}(-z)$ otherwise.\\

\noindent\textbf{Claim 3}. For any $z$, $\rho(z,0)=G\left(\frac{d^{+}(z)-d^{-}(z)}{d^{+}(z)+d^{-}(z)}\right)$.
\begin{proof}
Claim 2 implies that $\rho(z,0)=G\left(\frac{d^{+}(z)-d^{-}(z)}{d^{+}(z)+d^{-}(z)}\right)$ whenever $\sum_{k}z_k\geq 0$. Now consider the case where  $\sum_{k}z_k< 0$. By symmetry and linearity, $\rho(z,0)=1-\rho(-z,0)$, and Claim 2 implies $1-\rho(-z,0)=1-\tilde{G}\left(\frac{d^{-}(z)-d^{+}(z)}{d^{+}(z)+d^{-}(z)}\right)$. $1-\tilde{G}\left(\frac{d^{-}(z)-d^{+}(z)}{d^{+}(z)+d^{-}(z)}\right)=G\left(\frac{d^{+}(z)-d^{-}(z)}{d^{+}(z)+d^{-}(z)}\right)$ by construction, and so the desired result obtains. 
\end{proof}

Take any $x,y$, and let $z=x-y$. By linearity and Claim 3, we have $\rho(x,y)=\rho(z,0)=G\left(\frac{d^{+}(z)-d^{-}(z)}{d^{+}(z)+d^{-}(z)}\right)$. By construction, $G\left(\frac{d^{+}(z)-d^{-}(z)}{d^{+}(z)+d^{-}(z)}\right)=G\left(\frac{\sum_{k}z_k}{\sum_{k}|z_k|}\right)$, and so $\rho(x,y)=G\left(\frac{U(x)-U(y)}{d_{L1}(x,y)}\right)$ as desired.

To show uniqueness, suppose $(G,\beta)$ and $(G',\beta')$ both represent $\rho$. Define the stochastic preference relation $\succeq$ as before. Since $G$ and $G'$ are both strictly increasing and symmetric around 0, $U(x)=\sum_{k}\beta_kx_k$ and $U'(x)=\sum_{k}\beta'_kx_k$ both represent $\succeq$, and so there exists $C>0$ such that $\beta'_k=C\beta_k$. This in turn implies that for all $z\in\mathbb{R}^n$, we have $G\left(\frac{\sum_{k}\beta_kz_k}{\sum_{k}|\beta_kz_k|}\right)=G'\left(\frac{\sum_{k}\beta'_kz_k}{\sum_{k}|\beta'_kz_k|}\right)
    =G'\left(\frac{\sum_{k}\beta_kz_k}{\sum_{k}|\beta_kz_k|}\right)$. Let $z=\alpha/\beta_1 e_1-\gamma/\beta_2e_2$. Note that for any $r\in [-1,1]$, there exists $\alpha,\gamma$ such that $\frac{\sum_{k}\beta_kz_k}{\sum_{k}|\beta_kz_k|}=\frac{\alpha-\gamma}{|\alpha|+|\gamma|}=r$, and so $G'(r)=G(r)$ for all $r\in[-1,1]$.\\
\rightline{$\square$}

\subsection{Multinomial Choice Results}
\label{APP:proofs_multinomial}
We prove our results for a more general signal structure, where $s_{xy}=\text{sgn}(v_x-v_y)+\frac{1}{\sqrt{\tau_{xy}}}e_{xy}$ where the $e_{xy}$ are distributed according to a continuous density $g$ that is symmetric around 0 and satisfies the monotone likelihood ratio property: that is $\frac{\partial }{\partial x}\frac{g(x-t)}{g(x)}>0$ for all $t>0$. 

\begin{lemma}
\label{LEM:MLRP_prop}
Let $g$ be a continuous density that is symmetric around $0$ and satisfies the monotone likelihood ratio property.
\begin{enumerate}
    \item $g'(x-t)g(x)-g(x-t)g'(x)> 0$ for all $t>0$, $x$
    \item $g$ is unimodal; that is $g'(x)=-g'(-x)\leq 0$ for all $x> 0$
    \item $g(t-x)> g(-t-x)$ for any $t,x> 0$
\end{enumerate}
\end{lemma}
\begin{proof}
    1) follows directly from the definition of MLRP. To see 2), towards a contradiction suppose $g'(x)> 0$ for some $x> 0$. Then for any $t>0$, 1) implies $g'(x-t)\geq \frac{g(x-t)g'(x)}{g(x)}\geq 0$. Taking $t=2x$ gives $g'(-x)>0$, contradicting symmetry.  To see 3), note that by symmetry, $\frac{g(t-x)}{g(-t-x)}=1$ for $x=0$. MLRP of $g$ implies that $\frac{g(t-x)}{g(-t-x)}> 1$ for all $x> 0$ as desired. 
\end{proof}

The following observations pertain to a finite set of options $A$. Enumerate $A$ by $1,2,...,N$ and let $s=(s_{ij})_{i<j}$ collect all pairwise signals in $A$. Let $X_{(k)}^N$ denote the $k$th-highest order statistic among $N$ draws from the prior distribution $Q$. Let $V_{(k)}^N=\mathbb{E}[X_{(N+1-k)}^N]$, so that $V_{(k)}^N$ is the expected value of an option ranked $k$th from the top. Let $\pi:A\to A$ denote a permutation function; let $\Pi$ denote the set of permuation functions on $A$. With some abuse of notation, associate each $\pi$ with the event that the $v_i$'s are ordered according to $\pi$: that is $\pi(i)=n$ means that option $i$ is ranked $n$th in the ordering. The posterior expected value of an option $i$ given signal $s$ is then given by $\mathbb{E}[v_i|s]=\sum_{n=1}^NV_{(n)}^N \cdot Pr(\pi(i)=n|s)$, where $Pr(\pi(i)=n|s)\propto \sum_{\pi\in \Pi:\pi(i)=n} \prod_{k =1}^N\prod_{j<k} g\left(\sqrt{\tau_{jk}}(s_{jk}-\text{sgn}(\pi(k)-\pi(j))\right)$.

\begin{lemma}
\label{LEM:sigorder_monotonicity}
Take any permutation $\pi$ satisfying $\pi(i)<\pi(j)$. Then $\frac{\partial}{\partial e_{ij}}Pr(\pi|s)> 0$.
\end{lemma}
\begin{proof}
    We have
    \begin{align*}
        Pr(\pi|s)&=\frac{\prod_{k<l} g\left(\sqrt{\tau_{kl}}(s_{kl}-\text{sgn}(\pi(l)-\pi(k)))\right)}{\sum_{\pi'\in\Pi}\prod_{k<l} g\left(\sqrt{\tau_{kl}}(s_{kl}-\text{sgn}(\pi'(l)-\pi'(k)))\right)}\\
        &=\frac{\lambda g\left(\sqrt{\tau_{ij}}s_{ij}-\sqrt{\tau_{ij}}\right)}{\alpha g\left(\sqrt{\tau_{ij}}s_{ij}-\sqrt{\tau_{ij}}\right)+\beta g\left(\sqrt{\tau_{ij}}s_{ij}+\sqrt{\tau_{ij}}\right)}\\
        &=\frac{\lambda g\left(e_{ij}+\eta-\sqrt{\tau_{ij}}\right)}{\alpha g\left(e_{ij}+\eta-\sqrt{\tau_{ij}}\right)+\beta g\left(e_{ij}+\eta+\sqrt{\tau_{ij}}\right)}
    \end{align*}
where the $\lambda,\alpha,\beta,\eta$ do not depend on $e_{ij}$, and $\lambda,\alpha,\beta$ are non-negative. MLRP implies
\begin{align*}
    \frac{\partial }{\partial e_{ij}}Pr(\pi|s)&=\frac{\partial }{\partial e_{ij}}\left(\frac{\lambda}{\alpha+\beta \frac{g\left(e_{ij}+\eta+\sqrt{\tau_{ij}}\right)}{g\left(e_{ij}+\eta-\sqrt{\tau_{ij}}\right)}}\right)>0
\end{align*}
\end{proof}

\begin{lemma}
\label{LEM:sigcomp_monotonicity}
$1\set{\mathbb{E}[v_i|s]>\mathbb{E}[v_j|s]}$ is increasing in $e_{ij}$. 
\end{lemma}
\begin{proof}
We show the stronger result that $\mathbb{E}[v_i|s]-\mathbb{E}[v_j|s]$ is increasing in $s_{ij}$. Note that 
\begin{align*}
    \mathbb{E}[v_i|s]-\mathbb{E}[v_j|s]&=\sum_{\pi\in \Pi}\left(V^N_{(\pi(i))}-V^N_{(\pi(j))}\right)Pr(\pi|s)\\
    &=\sum_{\pi\in \Pi: \pi(i)<\pi(j)}\left(V^N_{(\pi(i))}-V^N_{(\pi(j))}\right)Pr(\pi|s)+\sum_{\pi\in \Pi: \pi(i)>\pi(j)}\left(V^N_{(\pi(i))}-V^N_{(\pi(j))}\right)Pr(\pi|s)
\end{align*}
Since $V^N_{(\pi(i))}-V^N_{(\pi(j))}> 0$ if $\pi(i)<\pi(j)$ and $V^N_{(\pi(i))}-V^N_{(\pi(j))}<0$ otherwise, Lemma \ref{LEM:sigorder_monotonicity} implies that $\frac{\partial}{\partial e_{ij}}\left[\mathbb{E}[v_i|s]-\mathbb{E}[v_j|s]\right]>0$. 
\end{proof}

We now observe a result that will be useful for the proof of Proposition \ref{PROP:context}. 

\begin{lemma}
    \label{LEM:context}
    Consider any two options $x,y$ with $\tau_{xy}=0$. Then, for any $z$ with $v_z>\max\set{v_y,v_x}$  $\rho(y,x|\set{z})$ is decreasing in $\tau_{yz}$ and increasing in $\tau_{xz}$. Likewise, if $v_z<\min\set{v_y,v_x}$  $\rho(y,x|\set{z})$ is increasing in $\tau_{yz}$ and decreasing in $\tau_{xz}$.
\end{lemma}

\begin{proof}

    Suppose that $v_z>\max\set{v_y,v_x}$; the proof for the case where $v_z<\min\set{v_y,v_x}$ is identical. For options $i,j,k$ let $\pi_{ijk}$ denote the permutation that ranks $i$ first, $j$ second, and $k$ last. We have $\mathbb{E}[v_y|s]-\mathbb{E}[v_x|s]=\left(V_{(1)}^3-V_{(3)}^3\right)(Pr(\pi_{yzx}|s)-Pr(\pi_{xzy}|s))$. Lemma \ref{LEM:sigorder_monotonicity} implies that $\mathbb{E}[v_y|s]-\mathbb{E}[v_x|s]$ is strictly increasing in $e_{yz}$ and decreasing in $e_{xz}$. 
    
    This implies that for $e^*_{yz}(\tau_{xz},e_{xz})$ defined implicitly by $g(e_{xz}-2\sqrt{\tau_{xz}})g(e^*_{yz}(\tau_{xz},e_{xz}))=g(e_{xz})g(e^*_{yz}(\tau_{xz},e_{xz})-2\sqrt{\tau_{yz}})$, for any realization of $e_{xz}$, we have $\mathbb{E}[v_y|s]-\mathbb{E}[v_x|s]=0$ when $e_{yz}=e^*_{yz}(\tau_{xz},e_{xz})$, and so $\mathbb{E}[v_y|s]-\mathbb{E}[v_x|s]>0$ whenever $e_{yz}>e^*_{yz}(\tau_{xz},e_{xz})$, and $\mathbb{E}[v_y|s]-\mathbb{E}[v_x|s]\leq 0$ otherwise. Here we note three properties of $e^*_{yz}(\tau_{xz},e_{xz})$: 1) $e^*_{yz}(\tau_{xz},e_{xz})$ is strictly increasing in $e_{xz}$; 2) $e^*_{yz}(\tau_{xz},e_{xz})$ is decreasing in $\tau_{xz}$ whenever $e_{xz}\leq \sqrt{\tau_{xz}}$; 3) $e^*_{yz}\left(\tau_{xz},\sqrt{\tau_{xz}}\right)=\sqrt{\tau_{yz}}$, and $e^*_{yz}(\tau_{xz},e_{xz})\leq  \sqrt{\tau_{yz}}$ whenever $e_{xz}\leq \sqrt{\tau_{xz}}$.
    
    Property 1 follows by implicitly differentiating the equality $\frac{g(e_{xz}-2\sqrt{\tau_{xz}})}{g(e_{xz})}=\frac{g(e^*_{yz}(\tau_{xz},e_{xz})-2\sqrt{\tau_{yz}})}{g(e^*_{yz}(\tau_{xz},e_{xz}))}$ and MLRP. Property 2 follows from differentiating the same equality, MLRP, and part 2 of Lemma \ref{LEM:MLRP_prop}.  Property 3 follows from symmetry of $g$ and Property 1. We have 
\begin{align*}
    \rho(y;x|z)=&\int_{e_{xz}=-\infty}^{\sqrt{\tau_{xz}}}\int_{e_{yz}=-\infty}^{\infty}1\left\{\mathbb{E}[v_y|s]-\mathbb{E}[v_x|s]\geq 0\right\}dg(e_{xz})dg(e_{yz})\\
    &+\int_{e_{xz}=\sqrt{\tau_{xz}}}^{\infty}\int_{e_{yz}=-\infty}^{\infty}1\left\{\mathbb{E}[v_y|s]-\mathbb{E}[v_x|s]\geq 0\right\}dg(e_{xz})dg(e_{yz})
\end{align*}
Note that
\footnotesize
\begin{align*}
    &\int_{e_{xz}=\sqrt{\tau_{xz}}}^{\infty}\int_{e_{yz}=-\infty}^{\infty}1\left\{\mathbb{E}[v_y|s]-\mathbb{E}[v_x|s]\geq 0\right\}dg(e_{xz})dg(e_{yz})\\
    &=\int_{e_{xz}=\sqrt{\tau_{xz}}}^{\infty}\int_{e_{yz}=-\infty}^{\infty}1\left\{g(e_{xz}-2\sqrt{\tau_{xz}})g(e_{yz})-g(e_{xz})g(e_{yz}-2\sqrt{\tau_{yz}})\geq  0\right\}dg(e_{xz})dg(e_{yz})\\
    &=\int_{e'_{xz}=-\infty}^{\sqrt{\tau_{xz}}}\int_{e'_{yz}=-\infty}^{\infty}1\left\{g(e'_{xz})g(e'_{yz}-2\sqrt{\tau_{yz}})-g(e'_{xz}-2\sqrt{\tau_{xz}})g(e'_{yz})\geq  0\right\}dg(e'_{xz}-2\sqrt{\tau_{xz}})dg(e'_{yz}-2\sqrt{\tau_{yz}})\\
    &=\int_{e'_{xz}=-\infty}^{\sqrt{\tau_{xz}}}\int_{e'_{yz}=-\infty}^{\infty}1\left\{\mathbb{E}[v_y|s]-\mathbb{E}[v_x|s]\leq 0\right\}dg(e'_{xz}-2\sqrt{\tau_{xz}})dg(e'_{yz}-2\sqrt{\tau_{yz}})
\end{align*}
\normalsize
where the third line uses the change of variables $e'_{xz}=2\sqrt{\tau_{xz}}-e_{xz}$, $e'_{yz}=2\sqrt{\tau_{yz}}-e_{yz}$. This implies that 
\footnotesize
\begin{align*}
    \rho(y;x|z)=&\int_{e_{xz}=-\infty}^{\sqrt{\tau_{xz}}}\int_{e_{yz}=e^*_{yz}(\tau_{xz},e_{xz})}^{\infty}g(e_{xz})g(e_{yz})\,de_{yz}de_{xz}+\int_{e'_{xz}=-\infty}^{\sqrt{\tau_{xz}}}\int_{e_{yz}=-\infty}^{e^*_{yz}(\tau_{xz},e_{xz})}g(e_{xz}-2\sqrt{\tau_{xz}})g(e_{yz}-2\sqrt{\tau_{yz}})\,de_{yz}de_{xz}\\
    =&\int_{e_{xz}=-\infty}^{\sqrt{\tau_{xz}}}\int_{e_{yz}=e^*_{yz}(\tau_{xz},e_{xz})}^{\infty}g(e_{xz})g(e_{yz})\,de_{yz}de_{xz}+\int_{e_{xz}=-\infty}^{-\sqrt{\tau_{xz}}}\int_{e_{yz}=-\infty}^{e^*_{yz}(\tau_{xz},e_{xz}+2\sqrt{\tau_{xz}})}g(e_{xz})g(e_{yz}-2\sqrt{\tau_{yz}})\,de_{yz}de_{xz}
\end{align*}
\normalsize
and so $\frac{\partial}{\partial\tau_{xz}}\rho(y;x|z)=A+B+C$, where
\small
\begin{align*}
    A&=\frac{g(\sqrt{\tau_{xz}})}{2\sqrt{\tau_{xz}}}\left[G\left(-\sqrt{\tau_{yz}}\right)-G\left(e_{yz}^*(\tau_{xz},\sqrt{\tau_{xz}})-2\sqrt{\tau_{yz}}\right)\right]\\
    B&=\int_{e_{xz}=-\infty}^{\sqrt{\tau_{xz}}}-\frac{\partial}{\partial\tau_{xz}}e^*_{yz}(\tau_{xz},e_{xz})\left[g(e_{xz})g(e^*_{yz}(\tau_{xz},e_{xz}))-g(e_{xz}-2\sqrt{\tau_{xz}})g(e^*_{yz}(\tau_{xz},e_{xz})-2\sqrt{\tau_{yz}})\right]\,de_{xz}\\
    C&=\int_{e_{xz}=-\infty}^{-\sqrt{\tau_{xz}}}\frac{1}{\sqrt{\tau_{xz}}}\frac{\partial}{\partial e_{xz}}e^*_{yz}(\tau_{xz},e_{xz}+2\sqrt{\tau_{xz}})g(e_{xz})g(e^*_{yz}(\tau_{xz},e_{xz}+2\sqrt{\tau_{xz}})-2\sqrt{\tau_{yz}})\,de_{xz}
\end{align*}
\normalsize

$A=0$ since $e^*(\tau_{xz},\sqrt{\tau_{xz}})=\sqrt{\tau_{yz}}$. To see that $B\geq0$, note that on the domain of integration, $\frac{\partial}{\partial\tau_{xz}}e^*_{yz}(\tau_{xz},e_{xz})\leq 0$ (Property 2), and $e^*_{yz}(\tau_{xz},e_{xz})\leq \sqrt{\tau_{yz}}$ (Property 3) and so applying part 3) of Lemma \ref{LEM:MLRP_prop}, $g(e^*_{yz}(\tau_{xz},e_{xz}))\geq g(e^*_{yz}(\tau_{xz},e_{xz})-2\sqrt{\tau_{yz}})$ and $g(e_{xz})>g(e_{xz}-2\sqrt{\tau_{xz}})$. To see that $C>0$, note that $\frac{\partial}{\partial e_{xz}}e^*_{yz}(\tau_{xz},e_{xz})> 0$ (Property 1).  We therefore have $\frac{\partial}{\partial\tau_{xz}}\rho(y,x|\set{z})> 0$. A symmetric argument shows that $\frac{\partial}{\partial\tau_{yz}}\rho(y,x|\set{z})< 0$. 
\end{proof}

\subsubsection*{Proof of Proposition \ref{PROP:context}}
Suppose $v_x,v_y>v_z$, and $\tau_{yz}>\tau_{xz}$. 
Lemma \ref{LEM:context} implies that if $\tau_{xy}=0$, $\rho(y,x|\set{z})>1/2$. The desired result follows from the fact that $\rho(y,x|\set{z})$ is continuous in $\tau_{xy}$. \qed

\subsubsection*{Proof of Proposition \ref{PROP:valuation}}
Let $\pi_k$ denote the ordering over $x,z^1,...,z^n$ in which $x$ is ranked $k$th and the $z^j$ are ordered correctly, and let $p_k(s)$ denote the DM's posterior belief over $\pi_k$ given signal $s$, where $p(s)=(p_1(s),...,p_{n+1}(s))$. Note that 
\begin{align*}
    \mathbb{E}[v_x|s]&=\sum_{k=1}^{n+1} V_{(k)}^N p_k(s)\\
    \mathbb{E}[v_j|s]&=\left(\sum_{k=1}^{j}p_k(s)\right) V_{(j+1)}^N +\left(\sum_{k=j+1}^{n+1}p_k(s) \right) V_{(j)}^N \quad\forall j=1,...,n
\end{align*}
where \(N=n+1\), and \(V_{(k)}^N\) is defined above.

To show (i), it suffices to show that when $\tau=0$, $\mathbb{E}[R(x,Z)]=\frac{n+2}{2}$. Suppose $\tau=0$. We have that with probability 1, $p_k(s)=1/{n+1}$ for all $k\in \set{1,...,n+1}$. Let $\mu$ denote the expectation of $Q$. By symmetry of $Q$, we have $V^N_{(k)}=2\mu - V^{N}_{(N+1-k)}$ for all $k=1,...,N$, and so $\mathbb{E}[v_x|s]=\mu$ with probability 1. 

First consider the case where $n$ is odd, and let $j^*=\frac{n+1}{2}$. We have $\mathbb{E}[v_{j^*}|s]=\frac{1}{2} V^N_{j^*}+\frac{1}{2} V^N_{j^*+1}=\mu$ with probability 1, and so $\mathbb{E}[v_x|s]=\mathbb{E}[v_{j^*}|s]$ and $\mathbb{E}[v_x|s]\neq \mathbb{E}[v_{k}|s]$ for any $k\neq j^*$ with probability 1. This implies that $\mathbb{P}(R(x,Z)=j^*)=\mathbb{P}(R(x,Z)=j^*+1)=1/2$, and so $\mathbb{E}[R(x,Z)]=\frac{n+2}{2}$ as desired. Now consider the case where $n$ is even. Let $j^*=n/2$, $k^*=n/2+1$. Since $V^N_{(k)}=2\mu - V^{N}_{(N+1-k)}$, we have $V_{(j^*)}^N>V_{(j^*+1)}^N=\mu=V_{(k^*)}^N>V_{(k^*+1)}^N$. This implies that with probability 1, $\mathbb{E}[v_{j^*}|s]=\frac{n/2
}{n+1}V_{(j^*)+1}^N+\frac{n/2+1}{n+1}V_{(j^*)}^N>\mu$, and $\mathbb{E}[v_{k^*}|s]=\frac{n/2+1
}{n+1}V_{(k^*)+1}^N+\frac{n/2}{n+1}V_{(k^*)}^N<\mu$, which in turn implies that $\mathbb{E}[v_{k^*}|s]<\mathbb{E}[v_x|s]<\mathbb{E}[v_{j^*}|s]$, and so we have $R(x,Z)=k^*=\frac{n+2}{2}$ with probability 1. This implies that $\mathbb{E}[R(x,Z)]=\frac{n+2}{2}$.

To show (2), Let $R(s)$ denote the DM's switching point given the signal $s$: that is, $R(s)=R$ if $\mathbb{E}[v_x|s]> \mathbb{E}[v_k|s]$ for all $k\geq R$ and $\mathbb{E}[v_x|s]<\mathbb{E}[v_k|s]$ for all $k<R$. Note that $R(s)$ is well defined for any $\tau>0$ since ties in posterior expected values occur with probability 0 if $\tau>0$. Note that there exists $\epsilon>0$ such that whenever $p_k(s)>1-\epsilon$, $R(s)=k$. Since $p_{R^*(x,Z)}(s)\to_p 1$ as $\tau\to\infty$, $R(s)\to_p R^*(x,Z)$ as $\tau\to\infty$.\qed

\section{Appendix: Experiments}
\label{APP:experiments}
Here we provide details on the design of our choice and valuation experiments. Screenshots of  experimental instructions, comprehension checks, and sample choice interfaces for these experiments are compiled in \href{https://jeffreyyang97.github.io/personalwebsite/CC_OA.pdf}{Supplemental Appendix J}.

\subsection{Multi-Attribute Binary Choice}

In our multiattribute choice experiments, we collected data on 662 choice problems in total: 582 problems in the \textit{main} problem sample, and 80 problems in a \textit{robustness} problem sample.

The main sample consists of 46 two-attribute problems, 432 three-attribute problems, and 104 four-attribute problems. The three-attribute choice options consist of a monthly fee, a per-GB usage rate (where the fictional consumer has a monthly usage of 6 GB), and an annual device cost; the two-attribute choice options consist of a monthly fee and usage rate, and the four-attribute choice options additionally contain a quarterly wi-fi charge. The two-attribute problems are generated by drawing a value difference (in bonus payment terms) from $\set{\$3.84,\$5.76}$ and an $L_1$-ratio from  $\set{1.00, 0.94, 0.89, 0.84, 0.80, 0.76, 0.70, 0.59, 0.48, 0.39, 0.30, 0.20}$.\footnote{Due to rounding in the attribute values, the actual $L_1$ ratios deviate slightly from these values.} The three- and four-attribute problems are generated by drawing a value difference and $L_1$ ratio value, which determines the summed attribute-wise advantages and disadvantages in the comparison, and randomizing how advantages and disadvantages are split across attributes. 

The robustness sample consists of 10 two-attribute problems, 60 three-attribute problems, and 10 four-attribute problems that are identical in structure to those in the main sample except for the attribute weights: in the robustness sample, the fictional consumer has a monthly usage of 12 GB. Each problem in the robustness sample is constructed to match the utility-weighted attribute values of a corresponding problem in the main sample. 

We collect data from the two problem samples in separate experiments. In the main experiment, each subject completes 50 choice problems in total: 30 randomly drawn unique three-attribute problems, 10 repeat problems drawn from these 30 unique problems, and 10 randomly drawn unique two- or four-attribute problems. Subjects first complete the 40 three-attribute problems; for their last 10 problems, they will see either two- or four-attribute problems, with 30\% of subjects randomly assigned to the two-attribute problems and the remaining subjects assigned to the four-attribute problems. The robustness experiment follows an identical structure, except that 50\% of subjects are randomly assigned  to the two-attribute problems with the remaining subjects assigned to the four-attribute problems. 

Subjects for both the main and robustness experiments were recruited from Prolific, screening for subjects based in the U.S. with a Prolific approval rating greater than or equal to 98\% and with 500 or more completes. Subjects who failed a comprehension check were screened out of the study. As pre-registered, data for both the main and robustness experiment were collected in waves to reach a pre-specified number of subjects who did not report using a calculator in the experiment: 350 for the main experiment and 48 in the robustness experiment. In total, 428 subjects were recruited for the main experiment (357 non-calculator users) and 65 subjects were recruited for the robustness experiment (50 non-calculator users). The pre-registration for these experiments can be accessed at \url{https://aspredicted.org/TNQ_XBQ}.

\subsection{Intertemporal Binary Choice}
In our intertemporal choice experiments, we collected data on 1097 choice problems in total: 900 problems in the \textit{broad} problem sample, and 197 problems in a \textit{targeted} problem sample.

In the broad problem sample, choice options contain either one or two payouts; in total, there are 300 1-payout vs. 1-payout choice problems, 300 1-payout vs. 2-payout choice problems, and 300 2-payout vs. 2-payout choice problems. For each choice problem, the options are generated by sampling payout amounts and payout delays. The delays of each payout (in days) are drawn from \{0, 12, 24, 48, 72, 108, 144, 180, 216, 264, 312, 360, 420, 480, 540, 600, 660, 720\}, and the monetary amount of each payout is drawn from $\set{\$0,\$0.50,...,\$20}$ for two-payout options and $\set{\$0,\$0.50,...,\$40}$ for one-payout options. Rather than uniformly sampling from these ranges, we employ a sampling procedure that 1) undersamples dominance problems, 2) excludes problems involving very large value differences and problems near indifference, and 3) stratifies by CPF ratio and value difference (computed using a benchmark discount factor). 
 
In the selected problem sample, problems are generated from sampling the same payout amounts and delays as for the broad problem sample, but are generated using a sampling procedure that holds fixed the threshold discount rate that makes the two options in the choice problem indifferent for a DM with exponential time preferences. In particular, 100 problems in the selected sample involve a threshold monthly discount rate of 1 (meaning that \textit{any} individual with exponential time preferences should prefer the option that pays off earlier), and 97 problems involve a threshold monthly discount rate of 0.747. Within each of these subsamples of problems, approximately half involve 1-payout vs 2-payout options, and the other half involves 2-payout vs. 2-payout options. The sampling procedure for the selected problem sample was additionally designed to stratify by CPF ratio and to reduce variation in the value difference. 

In the main experiment, each subject completes 50 choice problems in total: 40 unique problems randomly drawn from the combined sample of 1097 problems, and 10 repeat problems randomly drawn from these 40 unique problems. Subjects were recruited from Prolific, screening for subjects based in the U.S. with a Prolific approval rating greater than or equal to 98\% and with 500 or more completes. Subjects who failed a comprehension check were screened out of the study. 829 subjects in total were recruited. The pre-registration for this experiment can be accessed at  
\url{https://aspredicted.org/QCJ_S81}.

\subsection{Preference Reversal and Valuation Experiments}
\label{APP:experiments_val}
\subsubsection{Lottery Preference Reversals: Experimental Details}
The experiment concerns 12 lotteries: 6 ``base'' options consisting of 3 high-risk and 3 low-risk options as described in Table \ref{tab:reversals}, and 6 ``scaled-up'' options constructed by multiplying the base option payments by a factor of 1.6. The experiment contains two parts, the order of which is randomized: \textit{Binary choice} and \textit{Valuation}. 

In \textit{Binary Choice}, subjects make 16 binary choices between lotteries, 12 of which are drawn from the 18 possible high/low risk lottery comparisons within the base and scaled-up options, and 4 of which are filler problems included to limit repetition. Problem order is randomized, subject to the constraints that 1) no single choice option appears in consecutive choice problems, 2) no consecutive choice problems contain the same set of payoff probabilities, and 3) no consecutive choice problems are both filler problems. 

In \textit{Valuation}, subjects value all 12 options using one of two randomly assigned valuation modes: certainty equivalents or probability equivalents. For certainty equivalents,  each lottery $(\$\lambda w,p)$ is valued against a price list $Z=(z^1,..,z^n)$ of certain payments with evenly-spaced payoffs, i.e., $z^k=(\$\lambda [w-(k-1)d],1)$ and $n=\lfloor w/d\rfloor$, where  $\lambda\in\set{1,1.6}$ is the scale factor. We use $d=0.25$ for the low-risk lotteries and $d=1$ for the lottery $(\$\lambda\cdot 19.5,0.23)$, and $d=1.25$ for the remaining high-risk lotteries. For probability equivalents, $Z=(z^1,..,z^n)$ contains lotteries that pay off $\$\lambda\cdot 24$ with evenly-spaced payoff probabilities, i.e., $z^k=(\$\lambda\cdot 24, p-(k-1)d$) and $n=\lfloor p/d\rfloor$. We use $d=0.05$ for the low-risk lotteries and $d=0.01$ for the high-risk lotteries. Task order is randomized subject to the constraint that no consecutive valuation tasks involve lotteries with the same payoff probability or scale factor. 

If a participant is selected for a bonus (1 in 5 chance), one part of the study (\textit{Valuation} or \textit{Binary Choice}) is selected at random. If \textit{Binary Choice} is selected, subjects receive the option they chose in a randomly selected decision; otherwise subjects receive the option they chose in a randomly selected decision within a randomly selected price list. 

Subjects were recruited from Prolific, screening for subjects based in the U.S. with a Prolific approval rating greater than or equal to 98\% and with 500 or more completes. Subjects who failed a comprehension check were screened out of the study. 151 subjects in total were recruited. The median completion time was 20 minutes. The pre-registration for this experiment can be accessed at  
\url{https://aspredicted.org/C62_GRK}.

\subsubsection{Intertemporal Preference Reversals: Experimental Details}
The experiment concerns 12 delayed payments: 6 ``base'' options consisting of 3 high-delay and 3 low-delay options as described in Table \ref{tab:reversals}, and 6 ``scaled-up'' options constructed by multiplying the base options payments by a scale factor of 1.6. The study consists of two parts, the order of which is randomized: \textit{Binary choice} and \textit{Valuation}. 

In \textit{Binary Choice}, subjects make 16 binary choices between delayed payments, 12 of which are drawn from the 18 high/low delay comparisons within the base and scaled-up options, and 4 of which are filler problems included to limit repetition. Problem order is randomized, subject to the constraints that 1) no single choice option appears in consecutive choice problems, 2) no consecutive choice problems contain the same set of payoff delays, and 3) no consecutive choice problems are both filler problems. 

In \textit{Valuation}, subjects value all 12 options using one of two randomly assigned valuation modes: present value equivalents or time equivalents. For present value equivalents,  each option $(\$\lambda m,t)$ is valued against a price list $Z=(z^1,..,z^n)$ of evenly-spaced immediate payments, i.e., $z^k=(\$\lambda [m-(k-1)d],0)$ and $n=\lfloor m/d\rfloor$, where $\lambda\in\set{1,1.6}$ is the scale factor. We use $d=0.5$ for the low-delay options and $d=1.25$ for the high-delay options. For time equivalents, $Z=(z^1,..,z^n)$ contains delayed payments $z^k=(\$\lambda\cdot 27.5, \tau_k$) where $\tau_k=t+d_k$, for $\allowdisplaybreaks (d_1,...,d_n) = (0,7, 15, 30, 45, 60, 90, 120, 180, 240, 300, 360, 420, 480, 540, 600, 660, 720, 840, 960, 1080)$. 
Task order is randomized subject to the constraint that no consecutive valuation tasks involve options with the same payoff delay or scale factor. 

If a subject  is selected for a bonus (1 in 10 chance), one part of the study (\textit{Valuation} or \textit{Binary Choice}) is selected at random. If \textit{Binary Choice} is selected, subjects receive the option they chose in a randomly selected decision; otherwise, subjects receive the option they chose in a randomly selected decision within a randomly selected price list. 

Subjects were recruited from Prolific, screening for subjects based in the U.S. with a Prolific approval rating greater than or equal to 98\% and with 500 or more completes. Subjects who failed a comprehension check were screened out of the study. 152 subjects in total were recruited. The median completion time was 19 minutes. The pre-registration for this experiment can be accessed at  
\url{https://aspredicted.org/C62_GRK}.

\subsubsection{Valuation Experiments: Experimental Details}

Subjects complete two parts of the experiment, in random order: \textit{Risk} and \textit{Time}. In \textit{Risk} (\textit{Time}), subjects complete 12 multiple price list valuation tasks using one of two randomly selected valuation modes: certainty equivalents and probability equivalents (present value equivalents and time equivalents). 

{
    \def\OldComma{,}
    \catcode`\,=13
    \def,{%
      \ifmmode%
        \OldComma\discretionary{}{}{}%
      \else%
        \OldComma%
      \fi%
    }%
For certainty equivalents, subjects value a simple lottery $l = (\overline{w},p_l)$, against a price list $Z=\set{z^1,...,z^n}$ of certain payments adapted to $l$, with $n=19$. For probability equivalents, subjects value a certain payment $c=(w_c,1)$ against a probability list $Z=\set{z^1,...,z^n}$ of yardstick lotteries $z^k=(\overline{w},p_k)$ adapted to $c$, with $n=21$. We draw $\overline{w}$ from $\{\$9, \$18, \$27\}$. For certainty equivalents, we draw $p_l$ from $\{0.03, 0.05, 0.10, 0.25, 0.5, 0.75, 0.90, 0.95, 0.97\}$. For probability equivalents, we draw $w_c$ so that $w_c/\overline{w}\in \{0.033, 0.056, 0.11, 0.25, 0.5, 0.75, 0.89, 0.944, 0.967\}$. Subjects complete 12 price lists randomly selected from the 27 possible price lists; the order of the price lists is randomized, subject to the constraint that no consecutive price list contains the same payoff probability $p_l$ (normalized payment $w_c/\overline{w}$) for certainty equivalents (probability equivalents). 

For present value equivalents, subjects value a delayed payment $\upsilon = (\$\overline{m},t_{\upsilon})$, against a price list $Z=\set{z^1,...,z^n}$ of immediate payments adapted to $\upsilon$, with $n=21$. For time equivalents, subjects value an immediate payment $c=(m_c,0)$ against a time list $Z=\set{z^1,...,z^n}$ of yardstick delayed payments $z^k=(\overline{m},t_k)$, with $(t_1,...,t_n)=(0, 7, 15, 30, 45, 60, 90, 120, 150, 180, 240, 300, 360, 420, 480, 540, 600, 720, 840, 960, 1080)$. We draw $\overline{m}$ from $\{25, 30, 35\}$. For present value equivalents, we draw $t_{\upsilon}$ from $\{7, 30, 60, 120, 240, 360, 480, 720, 1080\}$ (in days). For time equivalents, we draw $m_c$ so that $m_c/\overline{m}\in \{0.20, 0.35, 0.50, 0.65, 0.75, 0.85, 0.90, 0.95, 0.97\}$. Subjects complete 12 price lists randomly selected from the 27 possible price lists; the order of the price lists is randomized, subject to the constraint that no consecutive price list contains the same payoff delay $t_{\upsilon}$ (normalized payment $m_c/\overline{m}$) for present value equivalents (time equivalents).

If a participant is selected to win a bonus (1 in 8 chance), one part of the study (\textit{Risk} or \textit{Time}) is selected at random, and a price list within that part is randomly selected; subjects receive the option they chose in a randomly selected decision in that  list.

Subjects were recruited from Prolific, screening for subjects based in the U.S. with a Prolific approval rating greater than or equal to 98\% and with 500 or more completes. Subjects who failed a comprehension check were screened out of the study. 302 subjects in total were recruited.  The median completion time was 24 minutes. The pre-registration for this experiment can be accessed at  
\url{https://aspredicted.org/D8R_552}.

\section{Appendix: Additional Analyses}
\subsection{Axiom Tests}

\label{APP:axiom_tests}

We conduct tests of the axioms developed in Section \ref{SEC:theory}  and Appendix \ref{APP:theory_axioms} on our binary choice data. As these choice experiments were not designed with the explicit goal of testing these axioms, they do not permit tests of some of our axioms, specifically Moderate Transitivity. However, some of our datasets do allow for tests of Dominance and Simplification, and Linearity, axioms that characterize our complexity measures. 

\subsubsection{Tests of Dominance}
Dominance can be decomposed into two implications: 1) if $x>_Dy$ and $w\not >_D z$, $\rho(x,y)\geq \max\set{\rho(w,z),\rho(z,w)}$, and 2) if $x>_Dy$ and $w >_D z$, then $\rho(x,y) = \rho(w,z)$.\\

\noindent \textbf{\textit{Test of Implication 1}}. In our multiattribute, intertemporal, and lottery datasets, there are 38, 51, and 1780 choice problems that involve dominance, and 624, 1046, and 9140 choice problems that do not, respectively. This makes for 23712, 53346, and 16269200 comparisons between choice rates of dominance problems (denoted $\rho^D$) and the maximal choice rates of non-dominance problems (denoted $\rho^{nD}$) in our three domains. 

We observe directional dominance violations ($\rho^D<\rho^{nD}$) in 15.62\%, 13.95\%, and 7.66\% of the multiattribute, intertemporal, and lottery comparisons, respectively. We quantify how many of these directional violations achieve statistical significance using the following procedure. For each dominance problem $i$: we conduct one-sided Fisher's exact tests of the null hypothesis $\rho_i^D\geq\rho_j^{nD}$ for each non-dominance problem $j$, adjusting $p$-values for multiple testing using the Benjamini-Hochberg procedure\footnote{This procedure has been shown to be less conservative than the standard Bonferroni correction and is robust under positive dependence across tests \citep{benjamini_control_2001}.}, and report the proportion that are significant at the 5\% level. No dominance problems in our multiattribute and intertemporal datasets exhibit any significant dominance violations with respect to non-dominance problems. In lottery choice, 10\% of dominance problems exhibit any significant dominance violations, with less than 2\% of dominance problems exhibiting significant dominance violations at a proportion higher than 5\%.\\ 

\noindent\textbf{\textit{Tests of Implication 2}}. 
Figure \ref{fig:dom_rates} plots the distribution of choice rates for dominance problems across our domains. Choice rates are tightly concentrated around the median in our multiattribute and intertemporal datasets, with more dispersion in lottery choice, although even here over 75\% of choice rates fall within a 10 percentage point range. 

\begin{figure}[b!]
    \small
    \begin{subfigure}[t]{0.325\textwidth}
        \caption{\centering Multiattribute}
        \vspace{0.5em}
        {\includegraphics[width=\linewidth]{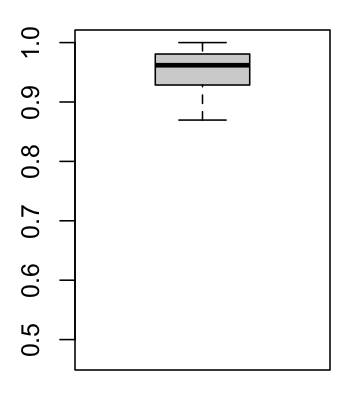}}
    \end{subfigure}
    \begin{subfigure}[t]{0.325\textwidth}
        \caption{\centering Intertemporal}
        \vspace{0.5em}
        {\includegraphics[width=\linewidth]{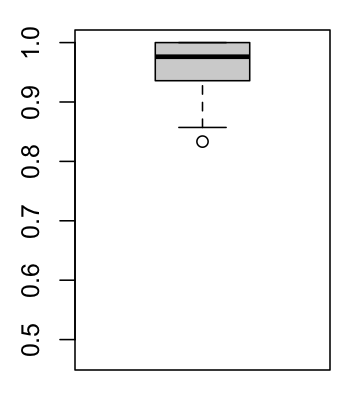}}
    \end{subfigure}
    \begin{subfigure}[t]{0.325\textwidth}
        \caption{\centering Lottery}
        \vspace{0.5em}
        {\includegraphics[width=\linewidth]{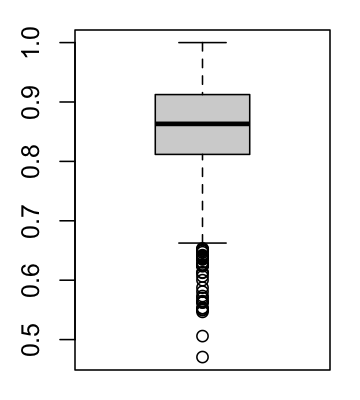}}
    \end{subfigure}
    \captionsetup{font=small}\caption{Choice rates for dominance problems. Boxes indicate the IQR of the distribution ($Q_3-Q_1$), and whiskers indicate $Q_1-1.5\times IQR$ and $Q_3+1.5\times IQR$.}
    \label{fig:dom_rates}
\end{figure}

We quantify the proportion of dominance problems that statistically significantly deviate from the average choice rate as follows. For each dominance problem $i$, we conduct a two-sided Fisher's exact test for the null hypothesis $\rho_i^D= \bar{\rho}^D_{-i}$, where $\bar{\rho}^D_{-i}=\frac{1}{\sum_{j\neq i}n_j}\sum_{j\neq i}n_j\rho_{j}^D$ is the leave-out estimate of the average choice rate among dominance problems, adjusting $p$-values for multiple testing using the Benjamini-Hochberg procedure. In our multiattribute and intertemporal datasets, we cannot reject the null at the 5\% level for any dominance problem, whereas in our lottery dataset, we can reject the null at the 5\% level for 18.65\% of the dominance problems. \\

In sum, we cannot reject Dominance in our multiattribute and intertemporal data, but we do find evidence of dominance violations in lottery choice. These violations are driven by heterogeneity in choice rates among dominance problems, specifically the left tail of dominance problems with lower choice rates than the average. A natural question is what drives  this heterogeneity in choice rates---that is, whether there is predictable and systematic variation in choice rates among dominance problems.

One way to quantify the extent of systematic variation is to compute the completeness index (see Appendix \ref{APP:completeness}) of our estimated CDF-complexity model on the subset of dominance problems. Intuitively, if there is a large degree of systematic variation in choice rates among dominance problems, this variation should be learned by our completeness benchmark; this leads to lower completeness of the CDF-complexity model, which constrains choice rates to be the same for all dominance problems. The completeness of the CDF-complexity model is virtually 1 on the subset of dominance problems, suggesting that in our lottery dataset, there is little predictable variation in the choice rates of dominance problems. That said, in other datasets, there is evidence that first-order stochastic dominance is systematically easier to recognize in some lottery comparisons than others \citep[e.g.][]{birnbaum_new_2008}. In light of this evidence, we leave developing suitable relaxations of Dominance to future work.  

\subsubsection{Tests of Simplification}
We test Simplification in our induced-values multiattribute data under the identifying assumption that the stochastic preference relation coincides with subjects' induced preferences: that $\rho(x,x')=1/2$ whenever $x$ and $x'$ have the same induced preference. Under this assumption, that utilities $U$ are observable, Simplification reduces to the first condition (S-a) below. We will also test an analogous implication of the $L_1$ model (S-b) that bears the same structure and interpretation as the Simplification axiom.\footnote{Condition S-b is jointly implied by Simplification and Linearity, assuming that utilities are observable.}\\

\noindent\textit{Condition S-a)}: If $U(x)\geq U(y)$, for any $x'$ with $U(x')=U(x)$ satisfying 1) $x'_k=y_k$ for some $k$, and 2) $x'_j\neq x_j$ for at most one $j\neq k$, we have $\rho(x',y)\geq \rho(x,y)$.\\

\noindent\textit{Condition S-b)}: If $U(x)\geq U(y)$, for any $y'$ with $U(y')=U(y)$ satisfying 1) $y'_k=x_k$ for some $k$, and 2) $y'_j\neq y_j$ for at most one $j\neq k$, we have $\rho(x,y')\geq \rho(x,y)$.\\

Our multiattribute choice dataset contains 99 non-trivial tests of Conditions S-a and S-b, all of which involve comparisons with the same structure, illustrated in Figure \ref{fig:simpl_test_ex}: there is a \textit{baseline} comparison where the superior option has an advantage (or a disadvantage) along 2 attributes and a disadvantage (or an advantage) along the remaining attribute, and a \textit{merged} comparison is formed by editing one of the options in the baseline comparison so as to merge both advantages (or both disadvantages) into a single attribute. Let $\rho^b$ and $\rho^{m}$ denote the probability of choosing the superior option in the baseline and merged comparisons respectively; Conditions S-a and S-b imply $\rho^m\geq \rho^b$.


\begin{figure}[htbp!]
\begin{align*}
{\small
\begin{array}{ll}
\multicolumn{2}{c}{\text{Baseline Comparison}}\\
    x&=(\$17.58\text{/mo},\,\$2.14\text{/GB},\,\$179.28\text{/yr})\\
    y&=(\$16.32\text{/mo},\,\$2.92\text{/GB},\,\$207.36\text{/yr})
\end{array}
\qquad
\begin{array}{ll}
\multicolumn{2}{c}{\text{Merged Comparison}}\\
    x'&=(\$17.58\text{/mo},\,\$1.75\text{/GB},\,\$207.36\text{/yr})\\
    y'&=(\$16.32\text{/mo},\,\$2.92\text{/GB},\,\$207.36\text{/yr})
\end{array}
}
\end{align*}
\captionsetup{font=small}
\caption{Example of choice problems that test Simplification. Recall that the value of each plan is their annual cost, assuming data usage of 6 GB/month.}
\label{fig:simpl_test_ex}
\end{figure}

We conduct one-sided Fisher's exact tests of the null hypothesis $\rho^m\geq \rho^b$ for all 99 comparisons, adjusting $p$-values for multiple testing using the Benjamini-Hochberg procedure. We fail to reject the null at the 5\% level in all but 1 comparison.

While we fail to statistically reject Simplification in our data, there is suggestive aggregate evidence of a specific violation of the axiom. The leftmost panel of Figure \ref{fig:simpl_rates} plots the distribution of $\rho^m-\rho^b$ across all 99 tests of simplification. This distribution is centered at 0, which is consistent with the prediction made by $L_1$ complexity that $\rho^b=\rho^m$. However, comparing the distributions of $\rho^m-\rho^b$ for tests that involve merging advantages of the superior option (middle panel) versus tests that involve merging disadvantages of the superior option (rightmost panel) reveals heterogeneity: $\rho^m-\rho^b$ is on average negative when merging advantages, and positive when merging disadvantages. This is consistent with an attribute-counting heuristic \citep{bushong_model_2021}, wherein all else equal, an option is perceived as more attractive if it is advantaged along more attributes and disadvantaged along fewer attributes. This heuristic can produce violations of Simplification, and we view the development of generalizations of $L_1$ complexity that allow for this heuristic as an interesting avenue for future research.

\begin{figure}[t!]
    \small
    \centering
    {\includegraphics[width=0.7\linewidth]{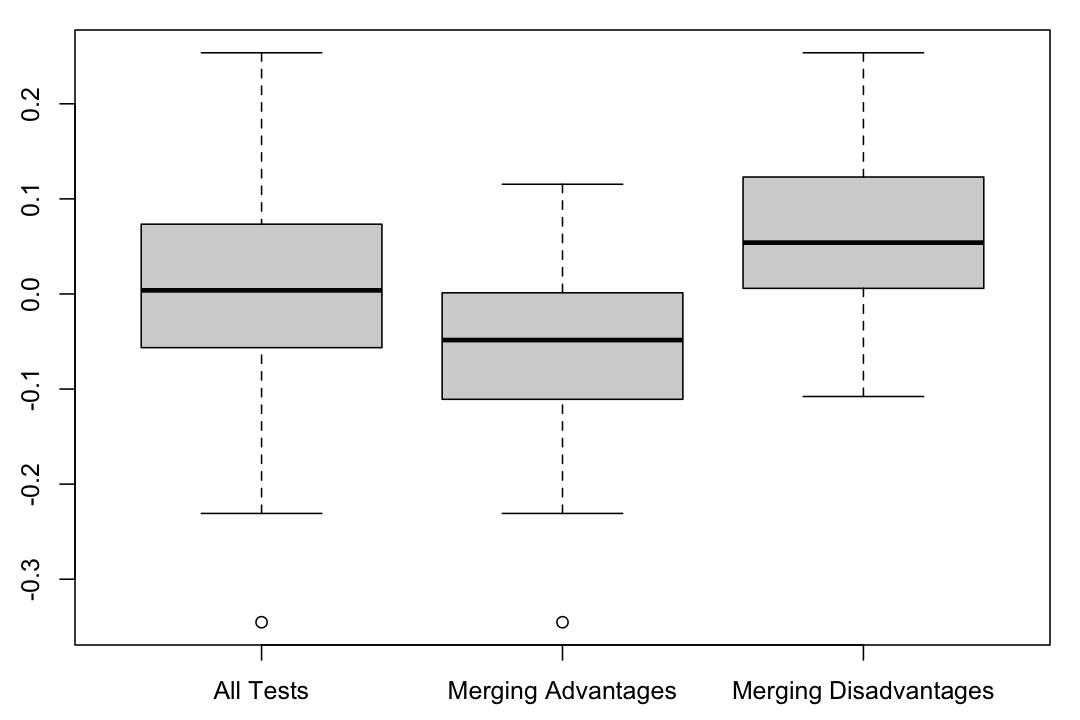}}

    \captionsetup{font=small}\caption{$\rho^m-\rho^b$ for simplification tests. Boxes indicate the IQR of the distribution ($Q_3-Q_1$), and whiskers indicate $Q_1-1.5\times IQR$ and $Q_3+1.5\times IQR$.}
    \label{fig:simpl_rates}
\end{figure}

\subsubsection{Tests of Linearity} Our multiattribute dataset allows for 81 pairwise tests of Linearity. Each involves comparisons of the form $(x,y)$, $(x',y')$, where $x'=x+z$ and $y'=y+z$ for some $z\in \mathbb{R}^n$. Linearity implies $\rho(x,y)=\rho(x',y')$.  We conduct two-sided Fisher's exact tests of the null $\rho(x,y)=\rho(x',y')$ for all 81 comparisons, adjusting $p$-values  using the Benjamini-Hochberg procedure. We fail to reject the null at the 5\% level in all comparisons. 

\subsection{Comparisons to Alternative Noise Specifications}
\label{APP:alternative_metrics}
We compare the $L_1$ complexity model to alternative heteroskedastic noise specifications in our multiattribute choice data. We consider a class of moderate utility specifications where the distance metric is given by the $L_p$ distance: 
\begin{align*}
    \rho(x,y)=G\left(\frac{U(x)-U(y)}{d_{{Lp}}(x,y)}\right)
\end{align*}
for $p\geq 1$, $G$ strictly increasing with $G(t)=1-G(-t)$, $U(x)=\sum_kx_k$, and $d_{Lp}(x,y)=\left(\sum_{k}|x_k-y_k|^p\right)^{1/p}$, and where each attribute is denoted in terms of its induced monetary value. Note that when $p=1$, this model reduces to the $L_1$ complexity model, and when $p=2$, the model is a linear differentiation model \citep{he_random_2023}. \\

\noindent \textbf{\textit{Goodness of Fit.}} Given a value of $p$, we flexibly estimate $G$ on our data using a shape-constrained additive model \citep{pya_shape_2015}, which represents $G$ with a monotone-increasing spline basis; we use a basis with 4 degrees of freedom in our estimation.  

Table \ref{tab:lp_fit} reports the predictive power of different noise specifications, including the $L_1$ model ($p=1$). Goodness of fit is very similar across all models. One reason for this similarity in performance is that on our dataset, the value-dissimilarity ratio $\frac{U(x)-U(y)}{d_{{Lp}}(x,y)}$ is highly correlated across the different noise specifications: for instance, the correlation between the ratios for $p=1$ and $p=2$ is 0.89. 

\begin{table}[h!]
\centering
\caption{Goodness of fit, $L_p$ ratio models} 
\label{tab:lp_fit}
\begin{tabular*}{0.8\linewidth}{@{\extracolsep{\fill}}l|ccccccccc@{}}
  \hline
 & $p=1$ & $p=1.25$ & $p=1.5$ & $p=2$ & $p=3$ & $p=5$ & $p=10$ & $p=20$ & $p=50$ \\ 
  \hline
$R^2$ & 0.357 & 0.365 & 0.366 & 0.362 & 0.361 & 0.362 & 0.362 & 0.361 & 0.361 \\ 
   \hline 
 \multicolumn{10}{l}{\makebox[0pt][l]{\hspace{-\tabcolsep}\parbox{0.8\linewidth}
                        {\scriptsize{$R^2$ values are observation weighted.}}}} 
\end{tabular*}
\end{table}

\noindent \textbf{\textit{Dominance Violations.}} Aggregate predictive power on our dataset aside, one key distinction between these noise specifications is their predictions regarding dominance. Whereas the $L_1$ complexity model respects dominance, for any $p>1$ the resulting noise specification predicts dominance violations: specifically, that one can construct comparisons $(x,y)$ and $(x',y')$ for which $x$ attribute-wise dominates $y$, $x'$ does not dominate $y'$, and yet predicted error rates are higher in $(x,y)$ than $(x',y')$. 

Our data allows for a test of these violations: it contains 12 choice problems $(x,y)$ in which $x$ dominates $y$ and $\frac{U(x)-U(y)}{d_{Lp}(x,y)}=1$ for all $p$, and 13 choice problems $(x',y')$ in which $x'$ does not dominate $y'$ and $\frac{U(x')-U(y')}{d_{Lp}(x',y')}>1$ for $p\geq 2$. That is, a noise specification with $p\geq 2$ predicts that error rates across all 12 dominance problems are \textit{higher} than error rates across all 13 non-dominance problems. Contrary to these predictions, the pooled error rates across the dominance and non-dominance problems are 4.53\% and 12.28\%, respectively, a statistically significant difference (Fisher's exact $p<0.001$).

\subsection{Alternative Value Difference Controls}
\label{APP:incons_model_controls}

Tables \ref{tab:incons_model_controls_risk} and \ref{tab:incons_model_controls_time} show that the relationship between choice inconsistency and the value-dissimilarity ratio in lottery and intertemporal choice is quantitatively similar when controlling for value differences estimated using the set of preference models considered in Section \ref{SEC:benchmarking}.

\begin{table}[h!]
\caption{Inconsistency Rate vs. CDF Ratio with Alternative Value-Difference Controls}
\begin{center}
\begin{tabular*}{0.815\linewidth}{@{\extracolsep{\fill}}l c c c c c@{}}
\hline
 & \multicolumn{5}{c}{\makecell{\textit{Dep. Variable}:\\ Inconsistency Rate}} \\
\cline{2-6}
 & (1) & (2) & (3) & (4) & (5) \\
\hline
Global CDF Ratio & $-0.10^{***}$ & $-0.11^{***}$ & $-0.11^{***}$ & $-0.11^{***}$ & $-0.11^{***}$ \\
                 & $(0.00)$      & $(0.00)$      & $(0.00)$      & $(0.00)$      & $(0.00)$      \\
$\Delta V$, EU   &               & $0.00^{***}$  &               &               &               \\
                 &               & $(0.00)$      &               &               &               \\
$\Delta V$, ST   &               &               & $0.00^{***}$  &               &               \\
                 &               &               & $(0.00)$      &               &               \\
$\Delta V$, RDEU &               &               &               & $0.00^{***}$  &               \\
                 &               &               &               & $(0.00)$      &               \\
$\Delta V$, CPT  &               &               &               &               & $0.00^{*}$    \\
                 &               &               &               &               & $(0.00)$      \\
\hline
Num. obs.        & $10420$       & $10420$       & $10420$       & $10420$       & $10420$       \\
\hline
\multicolumn{6}{l}{\scriptsize{\makebox[0pt][l]{\hspace{-\tabcolsep}\parbox{0.815\linewidth}{\vspace{4pt} OLS estimates. Standard errors (in parentheses) are robust.
                                            ``Global CDF Ratio'' is the representative-agent CDF ratio for each choice problem,
                                            computed using the value of $\alpha$ estimated in the Expected Utility model described
                                            in Appendix \ref{APP:structural_risk}. The $\Delta V$ controls are the absolute
                                            differences between the two lotteries' model-implied values, calculated using the
                                            corresponding estimates reported in the same appendix. ``EU,'' ``ST,'' ``RDEU,'' and
                                            ``CPT'' refer to Expected Utility, Simplicity Theory, Reference-Dependent Expected
                                            Utility, and Cumulative Prospect Theory, respectively. \\ $^{***}p<0.001$; $^{**}p<0.01$; $^{*}p<0.05$.}}}}
\end{tabular*}
\label{tab:incons_model_controls_risk}
\end{center}
\end{table}

\begin{table}[h!]
\caption{Inconsistency Rate vs. CPF Ratio with Alternative Value-Difference Controls}
\begin{center}
\begin{tabular*}{0.65\linewidth}{@{\extracolsep{\fill}}l c c c c@{}}
\hline
 & \multicolumn{4}{c}{\makecell{\textit{Dep. Variable}:\\ Inconsistency Rate}} \\
\cline{2-5}
 & (1) & (2) & (3) & (4) \\
\hline
Global CPF Ratio & $-0.20^{***}$ & $-0.18^{***}$ & $-0.17^{***}$ & $-0.18^{***}$ \\
                 & $(0.01)$      & $(0.02)$      & $(0.02)$      & $(0.02)$      \\
$\Delta V$, EDU  &               & $-0.01^{***}$ &               &               \\
                 &               & $(0.00)$      &               &               \\
$\Delta V$, HDU  &               &               & $-0.01^{***}$ &               \\
                 &               &               & $(0.00)$      &               \\
$\Delta V$, QDU  &               &               &               & $-0.01^{***}$ \\
                 &               &               &               & $(0.00)$      \\
\hline
Num. obs.        & $8290$        & $8290$        & $8290$        & $8290$        \\
\hline
\multicolumn{5}{l}{\scriptsize{\makebox[0pt][l]{\hspace{-\tabcolsep}\parbox{0.65\linewidth}{\vspace{4pt} OLS estimates. Standard errors (in parentheses) are robust.
                                            ``Global CPF Ratio'' is the representative-agent CPF ratio for each choice problem,
                                            computed using the value of $\delta$ estimated in the Exponential Discounting model
                                            described in Appendix \ref{APP:structural_time}. The $\Delta V$ controls are the
                                            absolute differences between the two options' model-implied present values, calculated
                                            using the corresponding estimates reported in the same appendix. ``EDU,'' ``HDU,'' and
                                            ``QDU'' refer to Exponential Discounting, Hyperbolic Discounting, and Quasi-Hyperbolic
                                            Discounting, respectively. \\ $^{***}p<0.001$; $^{**}p<0.01$; $^{*}p<0.05$.}}}}
\end{tabular*}
\label{tab:incons_model_controls_time}
\end{center}
\end{table}

\section{Appendix: Structural Specifications}
\label{APP:structural}
We estimate several standard models of value in multi-attribute objects, intertemporal payoffs, and lotteries, assuming logit choice probabilities $\rho(x, y) = \text{sgm}_{\eta}(V(x)-V(y))$, where $\text{sgm}_{\eta}(t)=1/(1+\exp(-\eta t))$ is the sigmoid function for $\eta\geq0$. For each of these models, we jointly estimate a parameterized $V$ function and the logit noise parameter $\eta$. We also estimate our parameterized model of complexity from Section \ref{SEC:theory_params},
\begin{align*}
    \rho(x,y)&=G\left(\frac{U(x)-U(y)}{d(x,y)}\right),\\
    G(r)&=\begin{dcases}(1-\kappa)-(0.5-\kappa)\frac{(1-r)^{\gamma}}{(r^{\psi}+(1-r)^{\psi})^{1/\psi}} & r\geq 0\\
            \kappa+(0.5-\kappa)\frac{(1+r)^{\gamma}}{((-r)^{\psi}+(1+r)^{\psi})^{1/\psi}} & r<0
            \end{dcases}
\end{align*}
for $\kappa\in [0,0.5], \gamma,\psi>0$. Unless stated otherwise, we will use the 2-parameter functional form of $G$ in which we fix $\psi=1$. In each domain, we jointly estimate the parameterized value-dissimilarity ratio and the $G$-function parameters $\kappa$ and $\gamma$ (and $\psi$, if applicable). Below we give the equations for each model estimated in the paper.

\subsection{Multiattribute Choice}
\label{APP:structural_mac}
In the following structural equations, we write each attribute in terms of its induced monetary value: that is, the true value of option $x$ is given by $U(x)=\sum_{k}x_k$.\\

\noindent\textbf{Distortion-Free Logit.} Choice rates are given by logit noise specification: $\rho(x,y)=\text{sgm}_{\eta}(U(x)-U(y))$. This model is parameterized by $\eta$.\\

\noindent\textbf{Salience.} We use the continuous salience-weighting model in \citet{bordalo_salience_2013}, where $\rho(x,y)=\text{sgm}_{\eta}(V_{BGS}(x|\set{x,y})-V_{BGS}(y|\set{x,y}))$ for $V_{BGS}(x|\set{x,y})\equiv \sum_{k} x_k\left(1+\frac{|x_k-(x_k+y_k)/2|}{|x_k|+|(x_k+y_k)/2|}\right)^{1-\delta}$, where $\delta\leq 1$. This model is parameterized by $(\eta,\delta)$.\\  

\noindent\textbf{Focusing.} We use the parameterization proposed in \citet{koszegi_model_2013}, where $\rho(x,y)=\text{sgm}_{\eta}(V_{KS}(x|\set{x,y})-V_{KS}(y|\set{x,y}))$ for $V_{KS}(x|\set{x,y})=\sum_{k} x_k|x_k-y_k|^{\theta}$, where $\theta\geq 0$. This model is parameterized by $(\eta,\theta)$.\\

\noindent\textbf{Relative Thinking.} We use the parameterization proposed in \citet{bushong_model_2021}, where $\rho(x,y)=\text{sgm}_{\eta}(V_{BRS}(x|\set{x,y})-V_{BRS}(y|\set{x,y}))$ for $V_{BRS}(x|\set{x,y})\equiv\sum_{k}x_k\left[(1-\omega)+\omega\frac{1}{|x_k-y_k|+\xi}\right]$, where $\omega\in [0,1]$, $\xi>0$. This model is parameterized by $(\eta,\omega,\xi)$. \\

\noindent\textbf{$\mathbf{L_1}$ Complexity}. Choice probabilities in our model are given by $\rho(x,y)=G\left(\frac{U(x)-U(y)}{d_{L1}(x,y)}\right)$, where $d_{L1}$ is defined as in Definition \ref{def:L1_complexity}. We estimate both the 2 and 3 parameter versions of $G$; our model is parameterized by $(\kappa,\gamma)$ for the former and $(\kappa,\gamma,\psi)$ for the latter. 

\subsection{Intertemporal Choice}
\label{APP:structural_time}

\textbf{Exponential Discounting.} Choice probabilities are given by $\rho(x,y)=\text{sgm}_{\eta}(PV(x)-PV(y))$, for $PV\equiv\sum_{t}\delta^t m_x(t)$. The parameters of the model are given by $(\eta,\delta)$. \\

\noindent\textbf{Quasi-Hyperbolic Discounting.} Choice probabilities are given by $\rho(x,y)=\text{sgm}_{\eta}(V_{qd}(x)-V_{qd}(y))$, for $V_{qd}\equiv\sum_{t>0}\beta\delta^t m_x(t) + m_x(0)$. The parameters of this model are $(\eta,\delta,\beta)$. \\

\noindent\textbf{Hyperbolic Discounting.}
We use the \citet{loewenstein_anomalies_1992} discount function: $\rho(x,y)=\text{sgm}_{\eta}(V_{hd}(x)-V_{hd}(y))$, for $V_{hd}\equiv\sum_{t}(1+\iota t)^{-\zeta/\iota} m_x(t)$, where $\iota,\zeta>0$. The parameters of this model are $(\eta,\iota,\zeta)$. \\

\noindent\textbf{CPF Complexity.} Choice probabilities in our model are given by $\rho(x,y)=G\left(\frac{PV(x)-PV(y)}{d_{CPF}(x,y)}\right)$, where $d_{CPF}(x,y)$ is defined as in Definition \ref{def:CPF_complexity}. Our model is parameterized by $(\delta,\kappa,\gamma)$.

\subsection{Lottery Choice}
\label{APP:structural_risk}

Define the Bernoulli utility $u_{sym}$ by $u_{sym}(w)=w^{\alpha}$ for $w\geq 0$, and $u_{sym}(w)=-(-w)^{\alpha}$ otherwise. Define $u_{rd}$ by $u_{rd}(w)=w^{\alpha}$ for $w\geq 0$, and $u_{rd}(w)=-\lambda(-w)^{\beta}$ otherwise. \\

\noindent\textbf{Expected Utility.}
To estimate the global preference parameters used in Figure \ref{fig:experiment_plots}  and Table \ref{tab:cdfregs_global}, we assume agents have a Bernoulli utility function that exhibits constant relative risk aversion for both pure-gain and pure-loss lotteries: $\rho(x,y)=\text{sgm}_{\eta}(EU_{sym}(x)-EU_{sym}(y))$, for $EU_{sym}(x)  \equiv \sum_{w}f_x(w)u_{sym}(w)$, where $\alpha> 0$. This model is parameterized by $(\eta,\alpha)$.\\

\noindent\textbf{Simplicity Theory.} The DM has EU preferences, but may penalize (or favor) lotteries with larger support. We follow \citet{puri_simplicity_2025} in parameterizing the penalization term: $\rho(x,y)=\text{sgm}_{\eta}(V_{st}(x)-V_{st}(y))$, for $V_{st}(x) = EU_{sym}(x)+CA(|S_x|)$ and $CA(s) = \frac{\phi}{1+\exp(\upsilon(s-\mu))} -  \frac{\phi}{1+\exp(\upsilon(1-\mu))}$, where $\alpha>0$. This model is parameterized by $(\eta,\alpha, \phi, \upsilon, \mu)$.\\ 

\noindent\textbf{Reference-Dependence.}
The DM has expected utility preferences, where the Bernoulli utility function allows for separate curvature parameters for positive and negative payouts, with a loss-aversion parameter $\lambda$: $\rho(x,y)=\text{sgm}_{\eta}(EU_{rd}(x)-EU_{rd}(y))$, $EU_{rd}(x)\equiv \sum_{w}f_x(w)u_{rd}(w)$, where $\alpha,\beta>0$. This model is parameterized by $(\eta,\alpha,\beta,\lambda)$.\\ 

\noindent\textbf{Cumulative Prospect Theory.}
We also estimate a model where the agent exhibits probability weighting and loss aversion, following \citet{tversky_advances_1992}. We use the probability weighting function given by \citet{gonzalez_shape_1999}. Let the distinct payoffs in a lottery $x$ be ordered by $w_{-m},...,w_{-1},w_{0},w_{1},...,w_{n}$, where $w_{-m},...,w_{0}$ indicate negative payoffs and $w_{0},...,w_{n}$ indicate positive payoffs, with $p_{-m},...,p_{n}$ denoting the associated probabilities. The value of $x$ is given by $V_{cpt}(x)=\sum_{k=-m}^0 u_{rd}(w_k)\pi_k+\sum_{k=0}^n u_{rd}(w_k)\pi_k$, where for $q(p)=\frac{\chi p^\nu}{\chi p^\nu + (1-p)^\nu}$, we define $\pi_n=q(p_n)$, $\pi_{-m}=q(p_{-m})$, $\pi_k=q(p_k+...+p_n)-q(p_{k+1}+...+p_n)$ for $0\leq k< n$, and $\pi_k=q(p_{-m}+...+p_k)-q(p_{-m}+...+p_{k-1})$ for $-m< k< 0$.
for $\alpha,\beta,\chi,\nu,\lambda>0$. Choice probabilities are given by $\rho(x,y)=\text{sgm}_{\eta}(V_{cpt}(x)-V_{cpt}(y))$. This model is parameterized by $(\eta,\alpha,\beta,\chi,\nu,\lambda)$.\\

\noindent\textbf{CDF Complexity.}
We estimate two versions of our model: one that assumes risk neutrality, and one that allows for utility curvature.  In the risk neutral model, we have $\rho(x,y) = G\left(\frac{EV(x)-EV(y)}{d_{CDF}(x,y)}\right)$, where $EV(x)=\sum_{w}wf_x(w)$ and $d_{CDF}$ are defined as in Definition \ref{def:CDF_complexity} with the 
Bernoulli utility function $u(x)=x$. This model is parameterized by $(\kappa,\gamma)$. In the model that allows for utility curvature, we have $\rho(x,y) = G\left(\frac{EU(x)-EU(y)}{d_{CDF}(x,y)}\right)$, where $EU(x)=\sum_{w}u_{sym}(w)f_x(w)$ and $d_{CDF}$ are defined as in Definition \ref{def:CDF_complexity} with the 
Bernoulli utility function $u=u_{sym}$. This model is parameterized by $(\kappa,\gamma,\alpha)$. 

\section{Appendix: Model Completeness and Restrictiveness}
\label{APP:completeness}
We adapt the completeness and restrictiveness measures in 
\citet{fudenberg_measuring_2022,fudenberg_how_2023} to our setting. We have a set of binary choice problems $\mathcal{C}$, each identified by a vector of objective problem features. For each $c\in\mathcal{C}$ there is an associated outcome $q\in \Delta(\set{a,b})$---a distribution over the options in the choice problem; let $\mathcal{Q}\equiv \Delta(\set{a,b})$. Abusing notation, we will also let $q\in[0,1]$ identify the rate of choosing option $a$.

In our dataset, there is a joint distribution over $\mathcal{C}\times\mathcal{Q}$ given by $\mu$. Letting $\mu_{\mathcal{C}}$ denote the marginal distribution over the choice problems in our dataset, we have $\mu_{\mathcal{C}}(c)=n_c/\sum_{c'\in \mathcal{C}}n_{c'}$, where $n_c$ denotes the number of observations for choice problem $c$, and letting $\mu_{\mathcal{Q}|\mathcal{C}}$ denote the conditional distribution over $\set{\delta_a,\delta_b}$ in our choice set, we have 
\begin{align*}
\mu_{\mathcal{Q}|\mathcal{C}}(q|c)=
\begin{cases}
r(c) & q=\delta_a\\
1-r(c) & q=\delta_b\\
0 & \text{otherwise}
\end{cases}
\end{align*}
where $r(c)$ denotes the empirical choice rate for option $a$ for choice problem $c$. 

We consider prediction rules of the form $p:\mathcal{C}\to\mathcal{Q}$ and denote the set of such functions by $\overline{\mathcal{P}}$; $p(c)$ denotes the rate of choosing option $a$ in choice problem $c$ under prediction rule $p$. Completeness and restrictiveness are defined for a parametric model $\mathcal{P}_{\Theta}=\set{p_{\theta}}_{\theta\in\Theta}$ with respect to a base model $p^{b}$: we take $p^b$ to be the constant prediction rule $p^b(c) = \sum_{c\in \mathcal{C}}\frac{n_c}{\sum_{c'\in\mathcal{C}}n_{c'}}r(c)$ outputting the average choice rate in lottery and intertemporal choice, and the fitted Distortion-Free Logit model in multiattribute choice.\footnote{In multiattribute choice, we use the Distortion-Free Logit model instead of the constant predictor as the base model since latter has \textit{lower} prediction loss than the former; in the framework of \citet{fudenberg_measuring_2022,fudenberg_how_2023}, the base model should have higher loss than the model under evaluation.}

\subsection{Completeness}
\textbf{Definition}. Let $l(q,q')=-\left[q'\log(q)+(1-q')\log(1-q)\right]$ denote the negative log-likelihood loss function. Analogous to the maximum-likelihood estimates discussed in Section \ref{SEC:benchmarking}, we measure the prediction loss of a model $p$ by the expected negative log-likelihood: $e(p) = \mathbb{E}_{\mu}[l(p(c),q)]$. Let $p^*$ denote the prediction rule that minimizes expected loss in the data, i.e., $p^*\in\arg\min_{p\in\overline{\mathcal{P}}}e(p)$. The completeness of a model $\mathcal{P}_{\Theta}$ is defined by $\kappa(\mathcal{P}_{\Theta})=\frac{e(p^b)-\min_{\theta\in\Theta}e(p_{\theta})}{e(p^b)-e(p^*)}$.\\

\noindent\textbf{Implementation}. To form $p^*$, we use an ensemble approach to build a predictor that maps problem features into choice probabilities, where we combine parametric model predictions with those of a neural network trained on the data. This ensemble predictor is described in detail in \href{https://jeffreyyang97.github.io/personalwebsite/CC_OA.pdf}{Supplemental Appendix H}.

\subsection{Restrictiveness}
\textbf{Definition}. Restrictiveness captures the degree to which a model is able to fit pre-defined synthetic data---the better the fit, the less restrictive the model. This synthetic data is defined by the \textit{admissible set}  $\mathcal{P}\subseteq\overline{\mathcal{P}}$, characterized by restrictions so that $\mathcal{P}$ constitutes ``reasonable'' choice data, as absent any such restrictions, any model satisfying basic restrictions could have high restrictiveness. Following \citet{fudenberg_how_2023}, we impose restrictions that are shared by the class of models we estimate: 1) \textit{Weak Dominance}: If $x> y$, then $\rho(x,y)\geq 1/2$; 2) \textit{Monotonicity}: If $x'> x$, then $\rho(x',y)\geq \rho(x,y)$, where $>$ denotes the domain-specific dominance notion.\footnote{That is, $>$ denotes attribute-wise dominance, first-order-stochastic dominance, and temporal dominance in multiattribute, lottery, and intertemporal choice, respectively.} Every model in Section \ref{APP:structural} satisfies Weak Dominance except for Simplicity Theory, and every model satisfies Monotonicity except for Simplicity Theory and the Salience model.

Let $d:\overline{\mathcal{P}}\times\overline{\mathcal{P}}\to\mathbb{R}_{+}$ denote the expected Kullback-Leibler divergence $d(p,p')=\mathbb{E}_{\mu_{\mathcal{C}}}[D(p'(c)||p(c))]$. Letting $\lambda_{\mathcal{P}}$ denote the uniform distribution on $\mathcal{P}$ and denoting $d(\mathcal{P}_{\Theta},p)=\inf_{\theta\in\Theta}d(p_{\theta},p)$, the restrictiveness of model $\mathcal{P}_{\Theta}$ is defined as $r(\mathcal{P}_{\Theta})\equiv \frac{\mathbb{E}_{\lambda_{\mathcal{P}}}[d(\mathcal{P}_{\Theta},p)]}{\mathbb{E}_{\lambda_{\mathcal{P}}}[d(p^{b},p)]}$.\\

\noindent\textbf{Implementation}. Following \citep{fudenberg_how_2023}, we generate $K=1000$ synthetic datasets by taking $K$ i.i.d. samples from $\lambda_{\mathcal{P}}$, denoted $\set{p^k}_{k=1}^K$, and form the estimator $\hat{r}(\mathcal{P}_{\Theta})=\frac{\frac{1}{K}\sum_{k=1}^K d(\mathcal{P}_{\Theta},p^k)}{\frac{1}{K}\sum_{k=1}^K d(p^{b},p^k)}$. Standard errors for $\hat{r}(\mathcal{P}_{\Theta})$ are computed following \citet{fudenberg_how_2023}.

Due to the high-dimensional nature of restrictions characterizing $\mathcal{P}$---the Monotonicity restriction generates 648, 6616, and 93529 independent pairwise inequality constraints between choice rates in our multiattribute, intertemporal, and lottery choice data, respectively---standard methods of sampling uniformly from $\mathcal{P}$ such as rejection sampling are computationally infeasible. Instead, we approximate uniform draws from $\mathcal{P}$ by implementing a hit-and-run sampler \citep{smith_efficient_1984}, a Markov chain Monte-Carlo algorithm. \href{https://jeffreyyang97.github.io/personalwebsite/CC_OA.pdf}{Supplemental Appendix I} describes this sampling procedure in full. \\

\newpage


\end{document}